\documentclass[aps,prl,twocolumn,superscriptaddress]{revtex4-2}
\usepackage{fullpage}
\usepackage{tikz}
\usepackage{amsthm,amsmath,amssymb}
\usepackage[hidelinks,draft=false,colorlinks=true]{hyperref}
\usepackage[nameinlink]{cleveref} % After hyperref, otherwise errors!
\usepackage{comment}
\usepackage{tikz-3dplot}
\usetikzlibrary{decorations.pathreplacing}

\newtheorem{theorem}{Theorem}

\theoremstyle{definition}

\newcommand{\tr}{\operatorname{Tr}}
\newcommand{\id}{\mathrm{Id}}

\newcommand{\ket}[1]{\vert #1 \rangle}
\newcommand{\bra}[1]{\langle #1 \vert}

\tikzset{
  tensor/.style={
    inner sep = 0.055cm,
    shape = circle,
    draw,
    fill
  },
  t/.style={
    inner sep = 0.03cm,
    shape = circle,
    draw,
    fill
  },
  every picture/.style = {
    baseline=-1mm,
    font=\scriptsize
    },
}

\begin{document}

\title{The local characterization of global tensor network eigenstates}

\author{Jos\'e Garre Rubio}
\affiliation{Instituto de F\'isica Te\'orica, UAM-CSIC, C. Nicol\'as Cabrera 13-15, Cantoblanco, 28049 Madrid, Spain}
\author{Andr\'as Moln\'ar}
\affiliation{\mbox{University of Vienna, Faculty of Mathematics, Oskar-Morgenstern-Platz 1, 1090 Vienna, Austria}}
\author{Norbert Schuch}
\affiliation{\mbox{University of Vienna, Faculty of Mathematics, Oskar-Morgenstern-Platz 1, 1090 Vienna, Austria}}
\affiliation{University of Vienna, Faculty of Physics, Boltzmanngasse 5, 1090 Vienna, Austria}
\author{Frank Verstraete}
\affiliation{ Department of Physics and Astronomy, Ghent University, Krijgslaan 299, 9000 Gent, Belgium}
\affiliation{ Department of Applied Mathematics and Theoretical Physics, University of Cambridge, Wilberforce Road, Cambridge, CB3 0WA, United Kingdom}
%\date{\today}

\begin{abstract}
We study the conditions under which Matrix Product States (MPS) or Matrix Product Operators are exact eigenvectors of an extensive local operator, such as a Hamiltonian. By suitably choosing the local operator, this covers a wide range of settings: Exact  eigenstates of Hamiltonians, including scar states, exact MPS trajectories for driven quantum systems, steady states of local Lindbladians, generalized symmetries of either Hamiltonians or density matrices, and many more. Our key result is that a local, fixed-size equation---namely, how a single term in the operator acts on a block of tensors---provides a necessary and sufficient condition for exact solutions. 
This local characterization allows us to obtain the full space of MPS solutions to all of the aforementioned problems; in particular, we exemplify the power of our results by recovering the quantum group symmetries of the XXZ model.
We also discuss applications to numerical algorithms with MPS and the generalization of our results to 2D, i.e., projected entangled pair states (PEPS). 

\end{abstract}

\maketitle

\emph{Introduction.---}Understanding the spectrum and the eigenstates of quantum systems is fundamental to condensed matter physics and quantum information. A central question is how to characterize the eigenstates of Hamiltonians that appear in nature: those consisting of local interactions (given by a sum of local terms). 

The characterization of the eigenstates in these systems, both ground states and excited states, is essential because they provide direct access to observable physical properties and enable the classification of quantum phases. Yet solving eigenvalue problems for local Hamiltonians in general is computationally difficult—already determining the ground state energy is QMA-hard \cite{kempe2005complexH}—and a direct solution becomes intractable for large systems.

Given this computational difficulty, physicists have long sought systematic
methods to find exact eigenstates by assuming they belong to a restricted class of states with some inherent structure. The arguably most successful such class are matrix product states (MPS), where the many-body wavefunction coefficients are constructed as a trace over a product of matrices, $\bra{i_1\cdots
i_N}\psi_N\rangle = \tr(A^{i_1}\cdots A^{i_N})$, where $A^i$ is a $D\times D$ matrix for every $i=1,\dots,d$ of the local Hilbert space~\cite{reviewPEPS}. The matrix indices are denoted by Greek letters $(\alpha,\beta,\cdots)$ and are dubbed {\it virtual} indices. Graphically, the state can be represented as 
\begin{equation*}
    \ket{\psi_N}  =
    \begin{tikzpicture}[xscale=0.7, yscale=0.4, font=\scriptsize]
        \draw[rounded corners] (-0.5,0) rectangle (3,-1.2);
        \draw (0,0)--(2.5,0);
        \node[tensor,label=below:$A$] (l) at (0,0) {};
        \draw (l)--++(0,1);
        \node[tensor,label=below:$A$] (m) at (1,0) {};
        \draw (m)--++(0,1);
        \node[tensor,label=below:$A$] (r) at (2.5,0) {};
        \draw (r)--++(0,1);
        \node[fill=white] at (1.75,0) {$\dots$};
    \end{tikzpicture} \ ,
\end{equation*}
where lines denote indices, and connected lines are summed over (contracted).
This ansatz is remarkably powerful for several reasons. First, gapped ground
states of local Hamiltonians in one dimension (1D) satisfy an area law for
the entanglement entropy, meaning their entanglement scales with the boundary of a
region rather than its volume \cite{Hastings07B}: in 1D this results in a constant
entanglement entropy. MPS naturally respect this area law and provide efficient
representations of area law states \cite{PhysRevB.73.094423,Hastings07A}. Second,
the MPS ansatz forms the foundation of the highly successful Density Matrix
Renormalization Group (DMRG) algorithm \cite{White92A}, which has become the
standard numerical tool for studying 1D quantum systems. Third,
many relevant models—such as the AKLT model and its generalizations—have ground
states that admit exact MPS representations with small bond dimension, providing
paradigmatic examples.

The defining advantage of the MPS ansatz is that the global properties of a many-body state—even those in the thermodynamic limit—are encoded within the structure of its individual, local tensors. This local encoding \cite{PhysRevLett.100.167202,Molnar18A} has been the cornerstone for classifying symmetry protected topological (SPT) phases \cite{Chen11, Schuch11}, describing topologically ordered states \cite{Bultinck17A}, understanding density matrix RG fixed points \cite{Cirac17A} and many other applications \cite{reviewPEPS}. It has long been ``folklore'' that not only local and global symmetries can be represented by local constraints, but also whether an MPS is an eigenstate of a local operator (as in Baxter's work \cite{BAXTER19731} for tensor product states). This has been used, for example, to find symmetries in spin chains 
\cite{MingruSym,fendley2025xyzintegrabilityeasyway,fukai2026matrixproductoperatorrepresentations,yashin2026twoparameterfamiliesmpointegrals},
also in non-equilibrium statistical physics \cite{Derrida_1993,Blythe_2007}, even if the MPS involved there have infinite bond dimension, to uncover duality transformations \cite{frassek2026}, in finding scar states \cite{Lin19,Sanada_2023}, as well as in many other examples. Despite this, a rigorous proof of necessity---i.e., that any symmetry can be represented locally---has been missing. Establishing a complete, ``if and only if'' local characterization of symmetries is essential to transform MPS from a successful numerical tool into a robust analytical framework for exploring and characterizing such scenarios in quantum many-body spin systems.

This is precisely what we achieve in this work: We derive a necessary and sufficient local equation---a ``complete characterization''---that describes how a $k$-local Hamiltonian acts on a block of $k$ adjacent MPS tensors. This result provides a unified
framework that covers a remarkably diverse range of physical scenarios. Beyond simply identifying ground states, our characterization allows for the systematic discovery of exact quantum many-body scars, the determination of steady states in dissipative Lindbladian dynamics, and the identification of all generalized matrix product operator (MPO) symmetries for a given Hamiltonian. By reducing a complex global problem to a simple local identity, our framework provides a rigorous foundation for existing numerical methods and offers a powerful new tool for discovering exact solutions in both 1D and 2D quantum systems.

\emph{Injective MPS.---}Our result applies to \emph{injective} MPS \cite{fannes1992finitely}, a central object in the theory of tensor networks. Their importance is highlighted by the following facts: (1) Any MPS, after blocking, decomposes into a sum of injective MPS. (2) Injective MPS are orthogonal to each other in the thermodynamic limit. (3) Injective MPS are unique ground states of local, gapped Hamiltonians. (4) An injective MPS is normalizable in the thermodynamic limit; this distinguishes it from almost all ``excited'' states which are non-normalizable plane waves of quasiparticles (with the exception of scar states). 
%An MPS is called injective if the map $\Gamma_L: X\mapsto \sum_i \tr(XA^{i_1} \dots A^{i_L}) \ket{i_1 \dots i_L}$ is injective for some $L$. If $\Gamma_L$ is injective, then so is $\Gamma_M$ for any $M\geq L$. The minimal such $L$ is called the injectivity length of $A$. For this $L$, there is a left inverse of $A$: a tensor $\tilde{A}$ on $L$ sites defined by the matrices $\tilde{A}^{i_1 \dots i_L}$ such that $\sum_i \tilde{A}^{i_1 \dots i_L}_{\gamma \epsilon} \cdot (A^{i_1} \dots A^{i_L})_{\alpha \beta}  = \delta_{\alpha\gamma} \cdot \delta_{\beta \epsilon}$~\cite{reviewPEPS}. The principal result of this work is the characterization of injective MPS that are the eigenvectors of local operators. 
An MPS is injective, with injectivity length $L$, if there is a left inverse of $A$: a tensor $\tilde{A}$ on $L$ sites defined by the matrices $\tilde{A}^{i_1 \dots i_L}$ such that $\sum_i \tilde{A}^{i_1 \dots i_L}_{\gamma \epsilon} \cdot (A^{i_1} \dots A^{i_L})_{\alpha \beta}  = \delta_{\alpha\gamma} \cdot \delta_{\beta \epsilon}$~\cite{reviewPEPS}. 

Our key result is the characterization of injective MPS that are eigenstates of local operators.

\emph{Main theorem.---}Let $\mathcal{O}_N$ be an operator acting on $N$ particles that is the sum of all translations of a $k$-local operator $O_{[i]}$ acting on sites from $i$ to $i+k-1$, $\mathcal{O}_N= \sum_{i=1}^N O_{[i]}$. An injective MPS $\ket{\psi_N}$ on $N$ sites defined by the matrices $A^i$ with injectivity length $L$ such that $N\geq 2L+2k-1$ is an exact eigenstate of $\mathcal{O}_N$,
\begin{equation}\label{eq:HGS}
    \mathcal{O}_N \ket{\psi_N} = E_N \ket{\psi_N},
\end{equation}
if and only if there exist ($d^k$ of size $D\times D$) matrices $B^{i_1 \dots i_k}$ such that
\begin{equation}
\begin{aligned}\label{eq:hAB}
   &\sum_{j\gamma} (O^{i_1\cdots i_k}_{j_1\cdots j_k}-\epsilon) A^{j_1}_{\alpha\gamma_1}\cdots A^{j_k}_{\gamma_k\beta} = \\
   &\sum_{\gamma} A^{i_1}_{\alpha\gamma}B^{i_2\cdots i_k}_{\gamma\beta}
  -\sum_{\gamma}B^{i_1\cdots i_{k-1}}_{\alpha\gamma}A^{i_k}_{\gamma\beta},
\end{aligned}
\end{equation}
where $\epsilon = E_N/N$. In graphical notation~\cite{reviewPEPS},
\begin{equation*}
    \begin{aligned}
    \begin{tikzpicture}[yscale=0.7, xscale=0.6]
        \draw (-0.5,0)--(3,0);
        \node[tensor,label=below:$A$] (l) at (0,0) {};
        \draw (l)--++(0,1);
        \node[tensor,label=below:$A$] (m) at (1,0) {};
        \draw (m)--++(0,1);
        \node[tensor,label=below:$A$] (r) at (2.5,0) {};
        \draw (r)--++(0,1);
        \draw[line width=1.7mm,line cap=round] (0,0.5) -- (2.5,0.5);
        \node[fill=white,inner sep = 1pt] at (1.75,0) {$\dots$};
        \node[anchor=south,inner sep=5pt] at (1.75,0.5) {$O-\epsilon$};
    \end{tikzpicture}  \:=\:
    \begin{tikzpicture}[yscale = 0.7, xscale=0.6]
        \draw (-0.5,0)--(3,0);
        \draw (1,0)--++(0,0.5);
        \draw (2.5,0)--++(0,0.5);
        \node[tensor,label=below:$A$] (r) at (0,0) {};
        \draw (r)--++(0,0.5);
        \draw[line width=2.1mm,line cap=round] (1,0) -- (2.5,0);
        \draw[red,line width=1.7mm,line cap=round] (1,0) -- (2.5,0);
        \node[anchor=north,inner sep = 1.7mm] at (1.75,0) {$B$};
        \node at (1.75,0.3) {$\dots$};
    \end{tikzpicture}  \:-\:  
    \begin{tikzpicture}[yscale=0.7,xscale=0.6]
        \draw (-0.5,0)--(3,0);
        \node[tensor,label=below:$A$] (l) at (2.5,0) {};
        \draw (l)--++(0,0.5);
        \draw (0,0)--++(0,0.5);
        \draw (1.5,0)--++(0,0.5);
        \draw[line width=2.1mm,line cap=round] (0,0) -- (1.5,0);
        \draw[red,line width=1.7mm,line cap=round] (0,0) -- (1.5,0);
        \node[anchor=north,inner sep = 1.7mm] at (0.75,0) {$B$};
        \node at (0.75,0.3) {$\dots$};
    \end{tikzpicture} \ .
    \end{aligned}
\end{equation*}

Let us emphasize some facets of this theorem. First, it is enough to assume that
Eq.~\eqref{eq:HGS} holds for \emph{one} given system size to conclude that
Eq.~\eqref{eq:hAB} holds. Then, as Eq.~\eqref{eq:hAB} is independent of $N$, the
MPS is an eigenstate for any other system size: 
$\mathcal{O}_K\ket{\psi_K} = E_K
\ket{\psi_K}, \ \forall K$, where $E_K = \epsilon\cdot K$. Second, $\mathcal{O}$
does not need to be Hermitian, and $E_N$ does not need to be its smallest (or largest) eigenvalue. Moreover, we do not assume that $\ket{\psi_N}$ is the unique eigenstate with eigenvalue $E_N$. Third we note that neither the MPS nor the local operator need to be translation invariant; the exact statement is presented below. Fourth, the tensor $B$ is not unique: any tensor of the form $B^{i_1 \dots i_k} + \lambda A^{i_1} \cdots A^{i_k}$ satisfies Eq.~\eqref{eq:hAB}; and this is the only ambiguity in $B$ (see the Supplemental Material (SM)). Fifth, we note that in most applications $O$ is given and the equation provides solutions for $A$ and $B$. Sixth, the result is also valid for one-body local operator in which case $B$ is just a virtual matrix, see SM. Finally, the action of $\mathcal{O}$ in a region $\Lambda$ at the physical boundary is:
\begin{equation} \Big(\sum_{i\in \Lambda} O_{[i]}
-\epsilon|\Lambda|\Big)\,\ket{\psi_N} = (Q_{[i]}-Q_{[i+|\Lambda|]})\,\ket{\psi_N}\ ,
\end{equation} where $\bra{l}Q\ket{j}= \sum_{\alpha,\beta}
B^l_{\alpha\beta}\tilde{A}^j_{\alpha\beta}$ assuming $L=1=k-1$ for simplicity. This form (local boundary action) has been used to classify Hamiltonian types in quantum many-body scar systems \cite{gioia2025}. 

\emph{A simple proof.---} A rigorous proof of the previous theorem can be found in the SM. Here, we provide a simplified and more intuitive argument (restricting for simplicity to the translation invariant case with $k=2$): The operator $e^{-i\mathcal{O}t}\cdot S$, where $S$ is the shift operator and $t$ is small, has an MPO description with the following MPO tensor $T$:
\begin{equation}\label{eq:mpo tensor expansion}
   \begin{tikzpicture}
        \draw (-0.5,0)--(0.5,0);
        \draw (0,-0.5)--(0,0.5);
        \node[tensor,label=below left:$T$] at (0,0) {};
   \end{tikzpicture} = 
    \begin{tikzpicture}
        \draw[rounded corners = 2mm] (-0.5,0)--(0,0)--(0,0.5);
        \draw[rounded corners = 2mm] (0,-0.5)--(0,0)--(0.5,0);
    \end{tikzpicture} - it\cdot
    \begin{tikzpicture}
        \draw[rounded corners = 2mm] (-0.5,0)--(0,0)--(0,0.5);
        \draw[rounded corners = 2mm] (0,-0.5)--(0,0)--(0.5,0);
        \draw[line width =1mm, line cap=round] (-0.15,0.15)--(0.15,-0.15);
        \node at (0.22,0.22) {$O$};
    \end{tikzpicture} + O(t^2)\:.
\end{equation}
As the MPS in the theorem is translation invariant and (after shifting by $\epsilon$) an eigenstate with eigenvalue zero of $\mathcal{O}$, $e^{-i\mathcal{O}t}\cdot S\ket{\psi_N} = \ket{\psi_N}$. This equality between MPSs can be characterized locally by using the ``fundamental theorem'' of MPS \cite{reviewPEPS} that implies the existence of a rank-three fusion tensor $F$ such that~\footnote{This is a simplification. For the exact statement see the Appendix}
\begin{equation}\label{eq:fundamental theorem}
    \begin{tikzpicture}[baseline=1.5mm]
        \draw (-0.5,0)--(0.5,0);
        \draw (-0.5,0.5) --(0.5,0.5);
        \draw (0,0)--(0,1);
        \draw (-1,0.25)--(-0.5,0.25);
        \draw[line width=1.5mm, line cap=round] (-0.5,0)--(-0.5,0.5);
        \node[anchor=north,inner sep = 1.7mm] at (-0.5,1) {$F$};
        \node[tensor,label=below:$A$] at (0,0) {};
        \node[tensor,label=above right:$T$] at (0,0.5) {};
    \end{tikzpicture} = 
    \begin{tikzpicture}[baseline=1.5mm]
        \draw (-0.5,0.25)--(0.5,0.25);
        \draw (0,0.25)--(0,0.75);
        \draw (0.5,0)--(1,0);
        \draw (0.5,0.5)--(1,0.5);
        \draw[line width=1.5mm, line cap=round] (0.5,0)--(0.5,0.5);
        \node[anchor=north,inner sep = 1.7mm] at (0.5,1) {$F$};
        \node[tensor,label=below:$A$] at (0,0.25) {};
    \end{tikzpicture} \ .
\end{equation}
Let us expand Eq.\eqref{eq:fundamental theorem} in orders of $t$ using the decomposition of Eq.\eqref{eq:mpo tensor expansion}. In zeroth order, we obtain that 
\begin{equation*}
    \begin{tikzpicture}[baseline=1.5mm]
        \draw (-0.5,0)--(0.5,0);
        \draw[rounded corners = 2mm] (-0.5,0.5) --(0,0.5) -- (0,1);
        \draw[rounded corners = 2mm] (0,0)--(0,0.5)--(0.5,0.5);
        \draw (-1,0.25)--(-0.5,0.25);
        \draw[line width=1.5mm, line cap=round] (-0.5,0)--(-0.5,0.5);
        \node[anchor=north,inner sep = 1.7mm] at (-0.5,1) {$F$};
        \node[tensor,label=below:$A$] at (0,0) {};
    \end{tikzpicture} = 
    \begin{tikzpicture}[baseline=1.5mm]
        \draw (-0.5,0.25)--(0.5,0.25);
        \draw (0,0.25)--(0,0.75);
        \draw (0.5,0)--(1,0);
        \draw (0.5,0.5)--(1,0.5);
        \draw[line width=1.5mm, line cap=round] (0.5,0)--(0.5,0.5);
        \node[anchor=north,inner sep = 1.7mm] at (0.5,1) {$F$};
        \node[tensor,label=below:$A$] at (0,0.25) {};
    \end{tikzpicture} \ ,
\end{equation*}
and thus, by injectivity, in zeroth order, the fusion tensor is nothing but the MPS tensor. We
can thus write \begin{equation*}
    \begin{tikzpicture}[baseline=1.5mm]
        \draw (0,0.25)--(0.5,0.25);
        \draw (0.5,0)--(1,0);
        \draw (0.5,0.5)--(1,0.5);
        \draw[line width=1.5mm, line cap=round] (0.5,0)--(0.5,0.5);
        \node[anchor=north,inner sep = 1.7mm] at (0.5,1) {$F$};
    \end{tikzpicture} \  = 
    \begin{tikzpicture}
        \draw (-0.5,0)--(0.5,0);
        \draw[rounded corners=2mm] (0,0)--(0,0.5)--(0.5,0.5);
        \node[tensor,label=below:$A$] at (0,0) {};
    \end{tikzpicture} - it \cdot \ 
    \begin{tikzpicture}
        \draw (-0.5,0)--(0.5,0);
        \draw[rounded corners=2mm] (0,0)--(0,0.5)--(0.5,0.5);
        \node[tensor, fill=red,label=below:$B$] at (0,0) {};
    \end{tikzpicture} + O(t^2) .
\end{equation*}
Plugging this back to Eq.~\eqref{eq:fundamental theorem}, expanding in $t$ again, and considering the first order terms, we obtain that
\begin{equation*}
    \begin{tikzpicture}
        \draw (-0.5,0)--(1,0);
        \draw[rounded corners=2mm] (0,0) -- (0,0.5) -- (0.5,0.5) -- (0.5,1);
        \draw[rounded corners=2mm] (0.5,0) -- (0.5,0.5)--(1,0.5);
        \node[tensor,label=below:$A$] at (0,0) {};
        \node[tensor,label=below:$A$] at (0.5,0) {};        
        \draw[line width =1mm, line cap=round] (0.5-0.15,0.5+0.15)--(0.5+0.15,0.5-0.15);
        \node at (0.5+0.22,0.5+0.22) {$O$};
    \end{tikzpicture} = 
    \begin{tikzpicture}
        \draw (-0.5,0)--(1,0);
        \draw (0,0) -- (0,0.5);
        \draw[rounded corners=2mm] (0.5,0) -- (0.5,0.5)--(1,0.5);
        \node[tensor,label=below:$A$] at (0,0) {};
        \node[tensor, fill=red,label=below:$B$] at (0.5,0) {};    
    \end{tikzpicture}
    -
       \begin{tikzpicture}
        \draw (-0.5,0)--(1,0);
        \draw[rounded corners=2mm] (0,0) -- (0,0.5) -- (0.5,0.5) -- (0.5,1);
        \draw[rounded corners=2mm] (0.5,0) -- (0.5,0.5)--(1,0.5);
        \node[tensor, fill=red,label=below:$B$] at (0,0) {};
        \node[tensor,,label=below:$A$] at (0.5,0) {};        
    \end{tikzpicture} ,    
\end{equation*}
which is the desired equation.

\emph{Different applications.---}Our main theorem has a plethora of applications, some of which are already used in the literature where the central principle of telescopic sums has been applied to great success. In these applications Eq.~\eqref{eq:hAB} is used as an \emph{ad hoc} \emph{ansatz}, that is, only the implication Eq.~\eqref{eq:hAB} $\Rightarrow$ Eq.~\eqref{eq:HGS} is used; the result presented in this paper provides the formal justification of these applications, by showing that this indeed presents the most general MPS solution to such an eigenvalue equation.
At the same time, our result provides a unified framework which allows to treat these problems as well as others which we now present.

\emph{(1) Time evolution.} Our theorem can be applied to find MPS solutions of
the (time-dependent) Schrödinger equation $H\ket{\psi} = i \partial_t
\ket{\psi}$. To see this, note that the time derivative of an MPS can be written
as $\partial_t\ket{\psi} = \sum_{j,\{i_k\}} \tr(A^{i_1}\cdots
(\partial_tA^{i_j}) \cdots  A^{i_N}) \ket{i_1\cdots i_N}$. For injective MPS, we
can construct a local operator $V$, $V^i_j= \sum_{\alpha,\beta}
(\partial_t A)^i_{\alpha\beta}\tilde{A}^j_{\alpha\beta}$ where $\tilde{A}$ is the left-inverse tensor, such that $\partial_t
\ket{\psi} = \sum_j V_{[j]} \ket{\psi}$ 
\cite{haegeman2011time,vanderstraeten2019tangent}. Then the Schrödinger equation reads
$(H-i\sum_j V_{[j]})\ket{\psi}= 0$, which is of the form of Eq.~\eqref{eq:HGS}. Our
Theorem then implies that the following must be met by any solution (shown here for 2-body
Hamiltonians):
    \begin{equation*}
        \begin{tikzpicture}[yscale=0.8, xscale=0.7]
            \draw (-0.5,0)--(1.5,0);
            \node[tensor,label=below:$A$] (l) at (1,0) {};
            \draw (l)--++(0,1);
            \node[tensor,label=below:$A$] (r) at (0,0) {};
            \draw (r)--++(0,1);
            \draw[line width=2mm,line cap=round] (0,0.5) -- (1,0.5);
            \node[anchor=south,inner sep=5pt] at (0.5,0.5) {$h$};
        \end{tikzpicture} = i
        \begin{tikzpicture}[yscale=0.8, xscale=0.7]
        \draw (-0.5,0)--(1.5,0);
            \node[tensor,label=below:$ A$] (r) at (1,0) {};
            \node[tensor,label=below:$\partial_t A$] (l) at (0,0) {};
            \draw (r)--++(0,0.5);
            \draw (l)--++(0,0.5);
        \end{tikzpicture}
        +
        \begin{tikzpicture}[yscale=0.8, xscale=0.7]
            \draw (-0.5,0)--(1.5,0);
            \node[tensor,fill=red,label=below:$B$] (l) at (1,0) {};
            \draw (l)--++(0,0.5);
            \node[tensor,label=below:$A$] (r) at (0,0) {};
            \draw (r)--++(0,0.5);
        \end{tikzpicture} - 
        \begin{tikzpicture}[yscale=0.8, xscale=0.7]
            \draw (-0.5,0)--(1.5,0);
            \node[tensor,label=below:$A$] (l) at (1,0) {};
            \draw (l)--++(0,0.5);
            \node[tensor,fill=red,label=below:$B$] (r) at (0,0) {};
            \draw (r)--++(0,0.5);
        \end{tikzpicture} \ ,       
    \end{equation*}
where now $A$ and $B$ (and potentially also $h$) depend on time. For larger injectivity lengths, the operator $V$ will correspondingly act on more sites.

\emph{(2) Continuous symmetries of matrix product density operators (MPDO).} 
Given a unitary symmetry $U$, a density matrix can be weakly symmetric, i.e., $[U,\rho]=0$, or strongly symmetric, i.e., $U\rho=\rho$, according to Ref.~\cite{de_Groot_2022}. Let $U$ be a continuous symmetry generated by 2-local operators, i.e., $U = e^{i\epsilon \sum_i G_{[i,i+1]}}$, and $\rho$  an MPDO, i.e., a density matrix with an MPS form~\cite{verstraete2004matrix,haegeman2017diagonalizing}. Then for weak symmetries,
$$\sum_i(G_{[i,i+1]}\otimes \id -\id\otimes G_{[i,i+1]}^T) |\rho \rangle\! \rangle =
0\ ,$$ while for strong symmetries, $\sum_i (G_{[i,i+1]}\otimes \id) \ket{\rho}\! \rangle=0 $,
where $\ket{\rho}\!\rangle$ denotes the vectorized form, as an MPS, of $\rho$.
Both of these are equations of the form of Eq.~\eqref{eq:HGS}, and thus allow for the application of our Theorem (the resulting ansatz has, for example, been used \emph{ad hoc} in Ref.~\cite{MingruSym}).

\emph{(3) Matrix product operator symmetries of local Hamiltonians.} 
Let us now consider MPO symmetries $O$ of local Hamiltonians: $[H,O]=0$ \cite{Bultinck17A,Lootens2021,Lootens2023dualities,Lootens2024sectors}. By vectorizing the injective MPO into an injective MPS, $O \rightarrow |O\rangle \!\rangle$, we can write $\sum_i(h_{[i]}\otimes \id -\id\otimes h_{[i]}^T) |O\rangle\! \rangle = 0$, whose structure mimics Eq.~\eqref{eq:HGS}. Our Theorem then gives the following local characterization of Hamiltonians with MPO symmetries:
    \begin{equation}\label{MPOsym}
        \begin{tikzpicture}[scale=0.6]
            \draw (-0.5,0)--(1.5,0);
            \node[tensor,label=below right:$T$] (l) at (1,0) {};
            \draw (l)--++(0,1);
            \draw (l)--++(0,-0.5);
            \node[tensor,label=below right:$T$] (r) at (0,0) {};
            \draw (r)--++(0,1);
            \draw (r)--++(0,-0.5);
            \draw[line width=2mm,line cap=round] (0,0.5) -- (1,0.5);
            \node[anchor=south,inner sep=5pt] at (0.5,0.5) {$h$};
        \end{tikzpicture} 
         - 
        \begin{tikzpicture}[scale=0.6]
            \draw (-0.5,0)--(1.5,0);
            \node[tensor,label=above right:$T$] (l) at (1,0) {};
            \draw (l)--++(0,-1);
            \draw (l)--++(0,0.5);
            \node[tensor,label=above right:$T$] (r) at (0,0) {};
            \draw (r)--++(0,-1);
            \draw (r)--++(0,0.5);
            \draw[line width=2mm,line cap=round] (0,-0.5) -- (1,-0.5);
            \node[anchor=north,inner sep=5pt] at (0.5,-0.5) {$h$};
        \end{tikzpicture}
        =
        \begin{tikzpicture}[scale=0.6]
            \draw (-0.5,0)--(1.5,0);
            \node[tensor,fill=red,label=below right:$B$] (l) at (1,0) {};
            \draw (l)--++(0,0.5);
            \draw (l)--++(0,-0.5);
            \node[tensor,label=below right:$T$] (r) at (0,0) {};
            \draw (r)--++(0,0.5);
            \draw (r)--++(0,-0.5);
        \end{tikzpicture} - 
        \begin{tikzpicture}[scale=0.6]
            \draw (-0.5,0)--(1.5,0);
            \node[tensor,label=below right:$T$] (l) at (1,0) {};
            \draw (l)--++(0,0.5);
            \draw (l)--++(0,-0.5);
            \node[tensor,fill=red,label=below right:$B$] (r) at (0,0) {};
            \draw (r)--++(0,0.5);
            \draw (r)--++(0,-0.5);
        \end{tikzpicture} ,       
    \end{equation} 
where $T$ is the tensor of the MPO symmetry $O$, and $B$ is the auxiliary tensor of Eq.~\eqref{eq:hAB}.\\
\emph{(3.1) An example: The quantum group symmetries of the XXZ model.} Let us find MPO symmetries of the XXZ model with the two-body term $h= XX+ YY +{\Delta}ZZ$. We write the MPO tensor $T$ as
\begin{equation*}
    \begin{tikzpicture}[baseline=-1mm, scale=0.3]
        \draw (-1,0) -- (1,0);
        \draw (0,-1) -- (0,1);
        \node[tensor] at (0,0) {};
    \end{tikzpicture} = \ket{0}\bra{0} \otimes a + \ket{0}\bra{1}\otimes b + \ket{1}\bra{0}\otimes c + \ket{1}\bra{1}\otimes d,
\end{equation*}
where the matrices $a,b,c,d$ act on the virtual d.o.f. We express the red tensor in Eq.~\eqref{MPOsym} in a similar form, with virtual matrices $\tilde{a},\tilde{b},\tilde{c},\tilde{d}$. The resulting equations coming from Eq.~\eqref{MPOsym} are given in the SM. Special solutions of these equations are given when $\tilde{a}=\tilde{d}=0$, $\tilde{b}= \lambda b$, $\tilde{c}= -\lambda c$, so that we get
\begin{align*}
    ab = & q^{-1} ba \ , \quad ac = q^{-1} ca \ , \quad bd = q^{-1} db \ , \\
    cd =& q^{-1}dc \ ,\quad bc = cb \ , \quad ad-da = -(q-q^{-1})bc \ ,
\end{align*} 
where $q$ and $\lambda$ satisfy $\Delta = \tfrac{q+q^{-1}}{2}$ and $\lambda = q-q^{-1}$. These equations correspond to the matrices defining the matrix quantum group $\mathrm{SL}(2)_q$, whose explicit connection with the quantum group $U_q(sl_2)$ as symmetries of the XXZ chain is shown in the SM (see also Ref. \cite{yashin2026twoparameterfamiliesmpointegrals}). In a similar way, we also show how the (higher) transfer matrices arise from solutions of Eq.~\eqref{MPOsym}, see SM. 

\emph{(4) Steady states of local Lindbladians.} Let us now consider the problem of finding the steady state, $\mathcal{L}(\rho) = 0$, of a 2-local Lindbladian $\mathcal{L}= \sum_i\mathcal{L}_{[i,i+1]}$. We assume that $\rho$ can be written as a MPDO, see Ref.~\cite{liu25lindbladiansMPO}. An MPDO can be decomposed into injective and locally orthogonal MPDOs $\rho=\sum_a \rho_a$, where we denote the injective tensor of $\rho_a$ by $T_a$. Then, $\mathcal{L}(\rho) = 0$ if and only if there are MPO tensors $M_a$ for each $a$ such that
\begin{align*}
   &\sum_{i'j'k'l',\gamma} \mathcal{L}^{ijkl}_{i'j'k'l'} (T_a)^{i'k'}_{\alpha\gamma}(T_a)^{j'l'}_{\gamma\beta} =  \\ &\qquad \sum_{\gamma,\delta} (T_a)^{ik}_{\alpha\gamma}(M_a)^{jl}_{\gamma\beta}-(M_a)^{ik}_{\alpha\delta}(T_a)^{jl}_{\delta\beta}\:.
\end{align*}

\emph{(5) Local operators and telescopic series.} Given a local operator that adds up to zero, $\sum O_{[i,i+1]}=0 $, we can use our result by applying the operator to the identity operator (which is an injective MPO), which yields that $O_{[i,i+1]}= Q_{[i]}\otimes \id_{[i+1]} - \id_{[i]}\otimes Q_{[i+1]}$, where $Q$ is a one-body operator; see also Ref.~\cite{rai2024}.

\emph{(6) Numerical applications: VUMPS.} Non fine-tuned Hamiltonians do not have exact MPS eigenstates, but the success of tensor networks is based on the fact that ground states can be faithfully approximated by MPS \cite{PhysRevB.73.094423}. One of the standard algorithms to find those MPS approximations is VUMPS~\cite{Zauner18,vanderstraeten2019tangent}.
Interestingly, the fixed points of this recursive algorithm can be understood to satisfy a relaxation of the telescopic identity in the main theorem: instead of satisfying this equality exactly, we only require it to be satisfied when projected on the tangent plane, see SM. The two error terms then correspond to the effective Hamiltonian contributions of the left and right part of the chain, respectively. 
Our central theorem hence gives a justification for the working of that algorithm.

\emph{Generalizations.---}The statement of our main theorem can be generalized in a variety of ways. Here we discuss two possible generalizations: non-translational systems and two-dimensional systems (which can be combined to state our result for non-translational two-dimensional systems). 

\emph{(7) Breaking translation invariance.} We now assume a setting where both
the local terms of the 2-body Hamiltonian $H=\sum h_{[n,n+1]}$  and the MPS
tensors  $A^{[n]}$  depend on the lattice site  $n$. Then the following holds:
\begin{align*}
     \sum_{i'j',\gamma} (h_{[n,n+1]})^{ij}_{i'j'} &A^{[n],i'}_{\alpha\gamma}A^{[n+1],j'}_{\gamma\beta} = \\ & \sum_{\gamma,\delta} A^{[n],i}_{\alpha\gamma}B^{[n+1],j}_{\gamma\beta}-B^{[n],i}_{\alpha\delta}A^{[n+1],j}_{\delta\beta}.
\end{align*}
This setting covers two interesting scenarios. The first one is where the MPS is composed of the 
same tensor everywhere, but a matrix is placed at the virtual indices when closing the boundaries:
this class of MPS describes coalgebras, see Ref.~\cite{molnar22}. The second case is where translation invariance is broken only inside a bigger unit cell. For example, Ref.~\cite{Lin19} constructed scar states with a $2$-site unit cell for the PXP model.

\begin{figure*}[t]
    \centering
    \includegraphics[width=\textwidth]{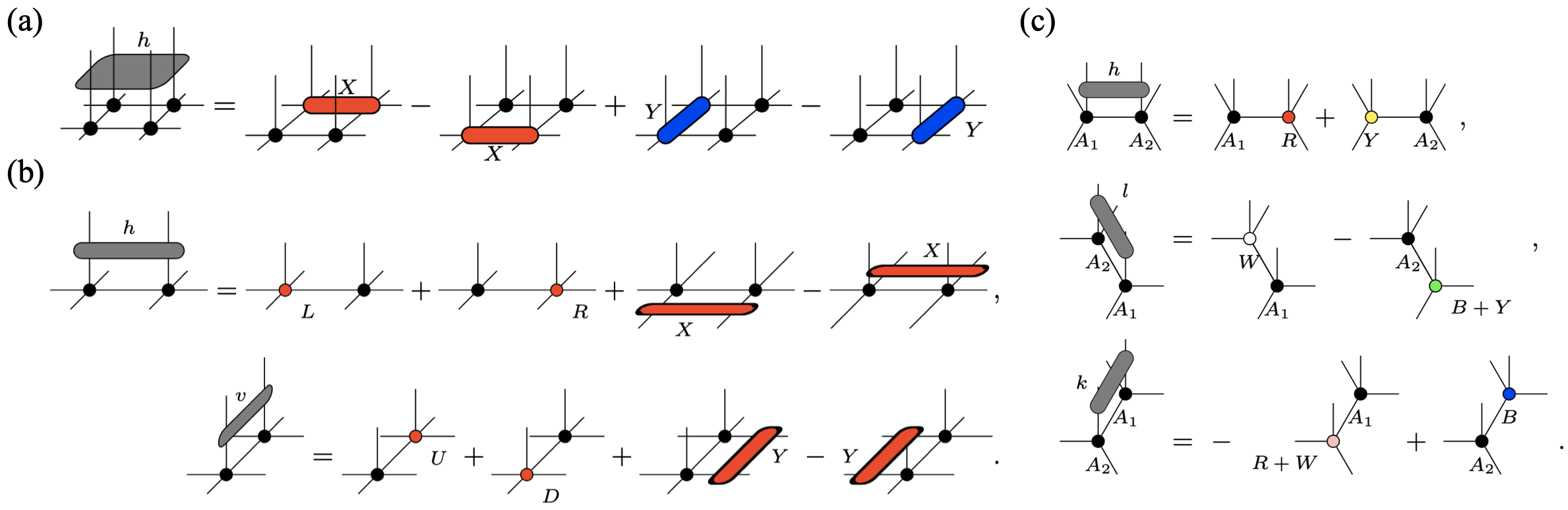}
    \caption{2D local Hamiltonian terms acting on injective PEPS as its ground state. (a) Plaquette term on square lattice, resulting in two virtual tensors $X$ and $Y$. (b) Two-body terms on square lattice, resulting in four tensors satisfying $R+D+L+U=0$ and two virtual tensors $X$ and $Y$.  (c) Two-body terms on hexagonal lattice, given by two site tensors $A_1$ and $A_2$ resulting in four tensors $R,W,B,Y$.}
    \label{fig:2dcases}
\end{figure*}

\emph{(8) Beyond 1D.} Our theorem can be generalized to 2D injective PEPS, see SM for the proof, in particular for two-body and plaquette local terms on the square lattice and two-body terms on the hexagonal lattice. These results are summarized in Fig. \ref{fig:2dcases}.

\emph{Conclusions and outlook.---}In this work, we have derived a necessary and sufficient local condition for an injective MPS to be an exact eigenstate of a finite-range operator, including local Hamiltonians. The resulting equation is formulated as the action of the $k$-local Hamiltonian term on a block of \(k\) adjacent MPS tensors, which is independent of the system size. It gives rise to a telescopic series, and is especially relevant for finding non-frustration free Hamiltonians (the parent Hamiltonian is always frustration free).

Our framework unifies several constructions into a single local principle, covering the time evolution of MPS, MPO symmetries, Lindbladian steady states, zero-sum local operators, and giving a conceptual justification of VUMPS. It further extends beyond translationally invariant 1D systems to site-dependent tensors (including PXP scar eigenstates), and to 2D injective PEPS as ground states of local (plaquette and two-body) Hamiltonians on the square and hexagonal lattices.

An important open direction is to extend these results to more general tensor network structures in two dimensions. In particular, proving analogues of our characterizing equation for generalized-injective \cite{Andras18B} and \(G\)-injective \cite{Schuch10} PEPS would allow to extend our framework to applications in symmetry protected and topologically ordered phases in 2D. From the numerical side, a robust 2D version of our characterization could provide new guidance for designing and diagnosing optimization schemes for PEPS and related ansatzes, potentially addressing some of the core difficulties in variational simulations of two and higher-dimensional quantum systems. 

\emph{Acknowledgments.---} We thank Yuan Miao and Hosho Katsura for helpful comments on quantum groups and pointing out Baxter's work \cite{BAXTER19731}. This research was funded in part by the Austrian
Science Fund (FWF) [Grant DOIs 10.55776/COE1, 10.55776/P36305, 10.55776/F71 and
10.55776/J4796], and the European Union through NextGenerationEU, as well as
through the European Research Council (ERC) under the European Union’s Horizon
2020 research and innovation programme (ERC-CoG SEQUAM, grant agreement
No.~863476).  FV acknowledges funding from the UKRI grant EP/Z003342/1, BOFGOA
(Grant No. BOF23/GOA/021), EOS (Grant No.~40007526), IBOF (Grant No.~IBOF23/064).

\bibliography{bibliography}

\clearpage
\newpage
\onecolumngrid
\appendix
\begin{center}
{\Large\bfseries SUPPLEMENTARY MATERIAL/APPENDICES}
\end{center}

\tableofcontents

\section{Rigorous proof of the Main Theorem}

Let us assume w.l.o.g.\ that the spectral radius of the transfer matrix $T$ of the injective tensor $A$, given by $T(X) = \sum_i A^i X (A^i)^\dagger$, is $1$. For an injective (normal) MPS tensor $A$ there is an operator such that $\rho_R = \sum_i A^i \rho_R (A^i)^\dagger$; this $\rho_R$ is unique up to a scalar multiplication, it is invertible, and it can be chosen to be positive definite. Similarly, there is an operator $\rho_L$ such that $\rho_L = \sum_i (A^i)^\dagger \rho_L A^i$, and this $\rho_L$ is unique up to a scalar multiplication, it is invertible, and it can be chosen to be positive definite. We assume w.l.o.g.\ that the operators $\rho_L$ and $\rho_R$ satisfy $\tr(\rho_L \rho_R) = 1$. Diagrammatically, we have
\begin{equation*}
    \begin{tikzpicture}[xscale=0.6,yscale=0.8,baseline=0.3cm]
        \draw (1,0) -- (-1,0) -- (-1,1) -- (1,1);
        \draw (0,0) -- (0,1);
        \node[tensor,label=above:$\bar{A}$] at (0,1) {};
        \node[tensor,label=below:$A$] at (0,0) {};
        \node[tensor, label=left:$\rho_L$] at (-1,0.5) {};
    \end{tikzpicture} = 
    \begin{tikzpicture}[xscale=0.6,yscale=0.8,baseline=0.3cm]
        \draw (0,0) -- (-1,0) -- (-1,1) -- (0,1);
        \node[tensor, label=left:$\rho_L$] at (-1,0.5) {};
    \end{tikzpicture}, \quad
    \begin{tikzpicture}[xscale=-0.6,yscale=0.8,baseline=0.3cm]
        \draw (1,0) -- (-1,0) -- (-1,1) -- (1,1);
        \draw (0,0) -- (0,1);
        \node[tensor,label=above:$\bar{A}$] at (0,1) {};
        \node[tensor,label=below:$A$] at (0,0) {};
        \node[tensor, label=right:$\rho_R$] at (-1,0.5) {};
    \end{tikzpicture} = 
    \begin{tikzpicture}[xscale=-0.6,yscale=0.8,baseline=0.3cm]
        \draw (0,0) -- (-1,0) -- (-1,1) -- (0,1);
        \node[tensor, label=right:$\rho_R$] at (-1,0.5) {};
    \end{tikzpicture}, \quad
    \begin{tikzpicture}[xscale=0.6,yscale=0.8,baseline=0.3cm]
        \draw (0,0) rectangle (-1,1);
        \node[tensor, label=right:$\rho_R$] at (0,0.5) {};
        \node[tensor, label=left:$\rho_L$] at (-1,0.5) {};
    \end{tikzpicture} = 1.
\end{equation*}
If $A$ is injective with injectivity length $L$, then there is a tensor $X$ such that 
\begin{equation}\label{eq:Ainv}
    \begin{tikzpicture}[xscale=0.6,yscale=0.6,baseline=0.2cm]
        \draw (-0.8,0) -- (3.8,0);
        \foreach \x in {0,1,3} {
            \draw (\x,0) --++(0,1);
            \node[tensor, label=below:$A$] at (\x,0) {};
        }
        \draw (-0.8,1) -- (3.8,1);
        \draw[line width = 1.7mm, line cap=round] (0,1) -- (3,1);
        \node[anchor=south] at (1.5,1.1) {$X$};
        \node[fill=white] at (2,0) {$\dots$};
    \end{tikzpicture} = 
    \begin{tikzpicture}[xscale=0.6,yscale=0.6,baseline=0.2cm]
        \draw (2,0) -- (1,0) -- (1,1) -- (2,1);
        \draw (-1,0) -- (0,0) -- (0,1) -- (-1,1);
        \node[tensor, label=right:$\rho_L$] at (1,0.5) {};
        \node[tensor, label=left:$\rho_R$] at (0,0.5) {};
    \end{tikzpicture} .
\end{equation}
Consequently, for any number of tensors on the left and right, we obtain
\begin{equation}\label{eq:Ainv_with_transfer}
    \begin{tikzpicture}[xscale=0.6,yscale=0.6,baseline=0.2cm]
        \draw (-3.8,0) -- (6.8,0);
        \draw (-3.8,1) -- (6.8,1);
        \foreach \x in {0,1,3} {
            \draw (\x,0) --++(0,1);
            \node[tensor, label=below:$A$] at (\x,0) {};
        }
        \foreach \x in {-1,-3,4,6} {
            \draw (\x,0) --++(0,1);
            \node[tensor, label=below:$A$] at (\x,0) {};
            \node[tensor, label=above:$\bar{A}$] at (\x,1) {};
        }
        \draw[line width = 1.7mm, line cap=round] (0,1) -- (3,1);
        \node[anchor=south] at (1.5,1.1) {$X$};
        \foreach \x in {(-2,0),(2,0),(5,0),(-2,1),(5,1)} {
            \node[fill=white] at \x {$\dots$};
        }
    \end{tikzpicture} = 
    \begin{tikzpicture}[xscale=0.6,yscale=0.6,baseline=0.2cm]
        \draw (2,0) -- (1,0) -- (1,1) -- (2,1);
        \draw (-1,0) -- (0,0) -- (0,1) -- (-1,1);
        \node[tensor, label=right:$\rho_L$] at (1,0.5) {};
        \node[tensor, label=left:$\rho_R$] at (0,0.5) {};
    \end{tikzpicture} .
\end{equation}

\begin{theorem}\label{thm:main}
Let $\mathcal{O}_N$ be an operator acting on $N$ particles that is the sum of the translations of a $k$-local operator $O$, $\mathcal{O}_N= \sum_{i=1}^N O_i$, where $o_i$ denotes the operator $O$ acting on the particles
\footnote{
we index the particles mod $N$, that is, the particle $N+l$ is the same as the particle $l$
}
$(i,i+1 \dots i+k-1)$. An injective MPS on $N$ sites defined by the matrices $A^i$ with injectivity length $L$ such that $N>2L+2k-1$, is an exact eigenstate of $\mathcal{O}_N$,
\begin{equation}\label{eq:HGSAp}
    \mathcal{O}_N \ket{\psi_N} = E_N \ket{\psi_N},
\end{equation}
if and only if there are matrices $B^{i_2 \dots i_k}$ such that
\begin{equation} \label{eq:hABAp}
   \sum_{j\gamma} (O^{i_1\cdots i_k}_{j_1\cdots j_k}-\epsilon) A^{j_1}_{\alpha\gamma_1}\cdots A^{j_k}_{\gamma_k\beta} = 
   \sum_{\gamma} A^{i_1}_{\alpha\gamma}B^{i_2\cdots i_k}_{\gamma\beta}
  -\sum_{\gamma}B^{i_1\cdots i_{k-1}}_{\alpha\gamma}A^{i_k}_{\gamma\beta},
\end{equation}
where $\epsilon = E_N/N$. In graphical notation,
\begin{equation}\label{eq:hABAp graphical}
    \begin{aligned}
    \begin{tikzpicture}[yscale=0.7, xscale=0.6]
        \draw (-0.5,0)--(3,0);
        \node[tensor,label=below:$A$] (l) at (0,0) {};
        \draw (l)--++(0,1);
        \node[tensor,label=below:$A$] (m) at (1,0) {};
        \draw (m)--++(0,1);
        \node[tensor,label=below:$A$] (r) at (2.5,0) {};
        \draw (r)--++(0,1);
        \draw[line width=1.7mm,line cap=round] (0,0.5) -- (2.5,0.5);
        \node[fill=white,inner sep = 1pt] at (1.75,0) {$\dots$};
        \node[anchor=south,inner sep=5pt] at (1.75,0.5) {$O-\epsilon$};
    \end{tikzpicture}  =
    \begin{tikzpicture}[yscale = 0.7, xscale=0.6]
        \draw (-0.5,0)--(3,0);
        \draw (1,0)--++(0,0.5);
        \draw (2.5,0)--++(0,0.5);
        \node[tensor,label=below:$A$] (r) at (0,0) {};
        \draw (r)--++(0,0.5);
        \draw[red,line width=1.7mm,line cap=round] (1,0) -- (2.5,0);
        \node[anchor=north,inner sep = 1.7mm] at (1.75,0) {$B$};
        \node at (1.75,0.3) {$\dots$};
    \end{tikzpicture}  -  
    \begin{tikzpicture}[yscale=0.7,xscale=0.6]
        \draw (-0.5,0)--(3,0);
        \node[tensor,label=below:$A$] (l) at (2.5,0) {};
        \draw (l)--++(0,0.5);
        \draw (0,0)--++(0,0.5);
        \draw (1.5,0)--++(0,0.5);
        \draw[red,line width=1.7mm,line cap=round] (0,0) -- (1.5,0);
        \node[anchor=north,inner sep = 1.7mm] at (0.75,0) {$B$};
        \node at (0.75,0.3) {$\dots$};
    \end{tikzpicture} \ .
    \end{aligned}
\end{equation}    
\end{theorem}

\begin{proof}
Setting $\epsilon = E_N/N$, the eigenvalue equation in graphical notation reads as
\begin{equation}\label{eq:eigenvalue_graphical}
    \sum_i \ 
    \begin{tikzpicture}
        \draw (-2.5,0) rectangle (6.5,-0.5);
        \foreach \x in {1,2,4}{
            \draw (\x,0)--(\x,1);
        }
        \draw[line width=2mm, line cap=round] (1,0.5)--(4,0.5);
        \node at (3,0.8) {$O-\epsilon$};
        \node at (1,1.2) {$i$};
        \node at (2,1.2) {$i+1$};
        \node at (4,1.2) {$i+k-1$};
        \foreach\x in {-2,0,1,2,4, 6}{
            \node[tensor] (t) at (\x,0) {};
            \draw (t)--++(0,0.5);
        }
        \foreach \x in {-1,3,5}{
            \node[fill=white] at (\x,0) {$\dots$};
        }
    \end{tikzpicture} = 0.
\end{equation}
We will contract this equation with different tensors, in four different ways. First, we contract \cref{eq:eigenvalue_graphical} with the contraction of a single tensor $X$ and $N-L$ tensors $\bar A$, contracted as 
\begin{equation*}
    \begin{tikzpicture}[yscale=0.6,xscale=1.1]
        \draw  (-0.6,0) rectangle (6.6,1);
        \foreach \x/\l in {0/1,1/2,3/{L},4/{L+1}, 6/{N}}{
            \draw (\x,0) --++ (0,-1);
            \node[anchor=north] at (\x,-1) {$\l$};
        }
        \draw[line width=2mm, line cap=round] (0,0)--(3,0) node[midway,anchor=south] {$X$};
        \foreach \x in {4,6}{
            \node[tensor, label=above:$\bar{A}$] at (\x,0) {};
        }
        \foreach \x in {(2,-0.7),(5,0)} {
            \node[fill=white] at \x {$\dots$};
        }
    \end{tikzpicture} \ .
\end{equation*}
Evaluating the contraction, we obtain an equation in the form
\begin{equation}\label{eq:cd1}
    0 = d + (N-k-L+1) c,
\end{equation}
where $d$ is the following sum of the $L+k-1$ terms where the support of $o_i$ overlaps with the support of $X$, i.e., the terms with $i=N-k+1,...,L$:
\begin{equation*}
    d = 
        \begin{tikzpicture}[xscale=0.6,baseline=0.4cm]
        \draw (-0.8,0) rectangle (6.8,1);
        \foreach \x in {0,1,3,4,6} {
            \draw (\x,0) --++(0,1);
            \node[tensor,label=below:$A$] at (\x,0) {};
        }
        \foreach \x in {0,1} {
            \node[tensor,label=above:$\bar{A}$] at (\x,1) {};
        }
        \node[fill=white] at (2,0) {$\dots$};
        \node[fill=white] at (2,1) {$\dots$};
        \node[fill=white] at (5,0) {$\dots$};
        \node[fill=white] at (5,0.5) {$\dots$};
        \node[anchor=south] at (2,0.5) {$O-\epsilon$};
        \draw[line width=2mm, line cap=round] (0,0.5) -- (3,0.5);
        \draw[line width=2mm, line cap=round] (3,1) -- (6,1) node[midway,anchor=south] {$X$};
        \node[tensor,label=left:$\rho_L$] at (-0.8,0.5) {};
        \node[tensor,label=right:$\rho_R$] at (6.8,0.5) {};
    \end{tikzpicture} + \dots +
    \begin{tikzpicture}[xscale=0.6,baseline=0.4cm]
        \draw (-3.8,0) rectangle (3.8,1);
        \foreach \x in {-3,-1,0,1,3} {
            \draw (\x,0) --++(0,1);
            \node[tensor,label=below:$A$] at (\x,0) {};
        }
        \foreach \x in {1,3} {
            \node[tensor,label=above:$\bar{A}$] at (\x,1) {};
        }
        \node[fill=white] at (2,0) {$\dots$};
        \node[fill=white] at (2,1) {$\dots$};
        \node[fill=white] at (-2,0) {$\dots$};
        \node[fill=white] at (-2,0.5) {$\dots$};
        \node[anchor=south] at (2,0.5) {$O-\epsilon$};
        \draw[line width=2mm, line cap=round] (0,0.5) -- (3,0.5);
        \draw[line width=2mm, line cap=round] (0,1) -- (-3,1) node[midway,anchor=south] {$X$};
        \node[tensor,label=left:$\rho_L$] at (-3.8,0.5) {};
        \node[tensor,label=right:$\rho_R$] at (3.8,0.5) {};
    \end{tikzpicture},
\end{equation*}
and $c$ is defined as
\begin{equation*}
    c =  
    \begin{tikzpicture}[xscale=0.6,baseline=0.4cm]
        \draw (-0.8,0) rectangle (3.8,1);
        \foreach \x in {0,1,3} {
            \draw (\x,0) --++(0,1);
            \node[tensor,label=below:$A$] at (\x,0) {};
            \node[tensor,label=above:$\bar{A}$] at (\x,1) {};
        }
        \node[fill=white] at (2,0) {$\dots$};
        \node[fill=white] at (2,1) {$\dots$};
        \node[anchor=south] at (2,0.5) {$O-\epsilon$};
        \draw[line width=2mm, line cap=round] (0,0.5) -- (3,0.5);
        \node[tensor,label=left:$\rho_L$] at (-0.8,0.5) {};
        \node[tensor,label=right:$\rho_R$] at (3.8,0.5) {};
    \end{tikzpicture}.
\end{equation*}

Second, we contract \cref{eq:eigenvalue_graphical} with a contraction of two tensors $X$ and $N-2L$ tensors $\bar{A}$ such that the two tensors $X$ are separated from each other by at least $k-1$ sites, in both direction. This is possible as we have assumed $N\geq 2L+2k-1$. One such configuration is depicted below:
\begin{equation*}
    \begin{tikzpicture}[yscale=0.6,xscale=1.1]
        \draw  (-0.6,0) rectangle (13.6,1);
        \foreach \x/\l in {0/1,1/2,3/{L},4/{L+1}, 6/{L+k-1}, 7/{L+k},8/{L+k+1},10/{2L+k-1\; },11/{\; 2L+k},13/N}{
            \draw (\x,0) --++ (0,-1);
            \node[anchor=north] at (\x,-1) {$\l$};
        }
        \draw[line width=2mm, line cap=round] (0,0)--(3,0) node[midway,anchor=south] {$X$};
        \draw[line width=2mm, line cap=round] (7,0)--(10,0) node[midway,anchor=south] {$X$};
        \foreach \x in {11,13,4,6}{
            \node[tensor, label=above:$\bar{A}$] at (\x,0) {};
        }
        \foreach \x in {(2,-0.7),(6,-0.7),(8,-0.7),(12,-0.7),(12,0),(5,0)} {
            \node[fill=white] at \x {$\dots$};
        }
    \end{tikzpicture} \ .
\end{equation*}
After the contraction, we obtain that 
\begin{equation*}
    0 = 2d+ (N-2L-2k+2)c,
\end{equation*}
where $d$ and $c$ are as above. Comparing this with \cref{eq:cd1}, we obtain that $c=d=0$.

Third, we contract \cref{eq:eigenvalue_graphical} with a contraction of the tensor product of two tensors $X$, and $N-2L-k+1\geq k$  tensors $\bar{A}$: 
\begin{itemize}
    \item $X$ on the positions $(1,\dots, L)$ and $(L+k,...,2L+k-1)$,
    \item and $\bar{A}$ on the positions $2L+k,\dots ,N$,
\end{itemize}
where the tensors are contracted as follows:
\begin{equation*}
    \begin{tikzpicture}[yscale=0.6,xscale=1.2]
        \draw  (4.0,0) -- (-0.6,0) -- (-0.6,1) -- (12.6,1) -- (12.6,0)--(5.0,0);
        \foreach \x/\l in {0/1,1/2,3/{L},6/{L+k},7/{L+k+1},9/{2L+k-1},10/{2L+k},12/N}{
            \draw (\x,0) --++ (0,-1);
            \node[anchor=north] at (\x,-1) {$\l$};
        }
        \draw[line width=2mm, line cap=round] (0,0)--(3,0) node[midway,anchor=south] {$X$};
        \draw[line width=2mm, line cap=round] (6,0)--(9,0) node[midway,anchor=south] {$X$};
        \node[tensor, label=above:$\bar{A}$] at (10,0) {};
        \node[tensor, label=above:$\bar{A}$] at (12,0) {};
        \node[tensor, label=above:$\rho_L^{-1}$] at (3.5,0) {};
        \node[tensor, label=above:$\rho_R^{-1}$] at (5.5,0) {};
        \foreach \x in {2,4.5,8,11} {
            \node[fill=white] at (\x,-0.7) {$\dots$};
        }
        \node[fill=white] at (11,0) {$\dots$};        
    \end{tikzpicture}
\end{equation*}
Let us evaluate the resulting equation. As above, we check the position of the local operator $O_i$ in each summand in \cref{eq:eigenvalue_graphical}. We group the summands in three different groups: the first group has overlap with the first $X$, that is, $i=N-k+2,\dots L$. In total there are $L+k-1$ such terms. In the second group, the support of $O_i$ has overlap with the support of the second $X$. This happens for $i=L+1, \dots, 2L+k-1$. In total, there are another $L+k-1$ such terms. The rest of the summands, $N-2L-2k+2$ terms, with $i=2L+k, \dots, N-k+1$, have no overlap with neither of the $X$'s. The equation we obtain from the contraction is thus of the form
\begin{equation*}
    \begin{tikzpicture}[scale=0.6]
        \draw (-0.8,0)--(3.8,0);
        \foreach \x in {0,1,3} {
            \draw (\x,0) --++(0,1);
        }
        \node at (2,0.65) {$\dots$};
        \draw[red,line width=2mm, line cap=round] (0,0)--(3,0) node[midway,below,anchor=north,black] {$B\vphantom{\tilde{B}}$};
    \end{tikzpicture} + 
    \begin{tikzpicture}[scale=0.6]
        \draw (-0.8,0)--(3.8,0);
        \foreach \x in {0,1,3} {
            \draw (\x,0) --++(0,1);
        }
        \node at (2,0.65) {$\dots$};
        \draw[blue,line width=2mm, line cap=round] (0,0)--(3,0) node[midway,below,anchor=north,black] {$\tilde{B}$};
    \end{tikzpicture} 
    + (N-2L-2k+2) c \cdot 
    \begin{tikzpicture}[scale=0.6]
        \draw (-0.8,0)--(3.8,0);
        \foreach \x in {0,1,3} {
            \draw (\x,0) --++(0,1);
            \node[tensor,label=below:$A$] at (\x,0) {};
        }
        \node[fill=white] at (2,0) {$\dots$};
    \end{tikzpicture} = 0,
\end{equation*}
where $B$ contains terms originating from local operators that have overlap with the left $X$, i.e., those that start at positions $N-k+2,\dots, L$:
\begin{equation*}
\begin{aligned}
    \begin{tikzpicture}[scale=0.5]
        \draw (-0.8,0)--(3.8,0);
        \foreach \x in {0,1,3} {
            \draw (\x,0) --++(0,1);
        }
        \node at (2,0.65) {$\dots$};
        \draw[red,line width=2mm, line cap=round] (0,0)--(3,0) node[midway,below,anchor=north,black] {$B\vphantom{\tilde{B}}$};
    \end{tikzpicture} = 
        \begin{tikzpicture}[xscale=0.55,baseline=0.4cm]
        \draw (9.8,0) -- (-0.8,0) -- (-0.8,1) -- (6.5,1) -- (6.5,2) -- (6,2);
        \foreach \x/\y in {0/1,1/1,3/1,4/1,6/1,7/0.5,9/0.5} {
            \draw (\x,0) --++(0,\y);
            \node[tensor,label=below:$A$] at (\x,0) {};
        }
        \foreach \x in {0,1} {
            \node[tensor,label=above:$\bar{A}$] at (\x,1) {};
        }
        \foreach \x in {(2,0), (2,1), (5,0), (5,0.5), (8,0)}{
            \node[fill=white] at \x {$\dots$};
        }
        \node[anchor=south] at (2,0.5) {$O-\epsilon$};
        \draw[line width=2mm, line cap=round] (0,0.5) -- (3,0.5);
        \draw[line width=2mm, line cap=round] (3,1) -- (6,1) node[midway,anchor=south] {$X$};
        \node[tensor,label=left:$\rho_L$] at (-0.8,0.5) {};
        \node[tensor,label=right:$\rho_L^{-1}$] at (6.5,1.5) {};
    \end{tikzpicture} + \dots +
    \begin{tikzpicture}[xscale=0.55,baseline=0.4cm]
        \draw (4.8,0) -- (-3.8,0) -- (-3.8,1) -- (0.5,1) -- (0.5,2) -- (0,2);
        \foreach \x/\y in {-3/1,-1/1,0/1,1/1,3/1,4/0.5} {
            \draw (\x,0) --++(0,\y);
            \node[tensor,label=below:$A$] at (\x,0) {};
        }
        \node[fill=white] at (2,0) {$\dots$};
        \node[fill=white] at (2,1) {$\dots$};
        \node[anchor=south] at (2,0.5) {$O-\epsilon$};
        \draw[line width=2mm, line cap=round] (0,0.5) -- (3,0.5);
        \draw[line width=2mm, line cap=round] (0,1) -- (-3,1) node[midway,anchor=south] {$X$};
        \node[tensor,label=left:$\rho_L$] at (-3.8,0.5) {};
        \node[tensor,label=right:$\rho_L^{-1}$] at (0.5,1.5) {};
    \end{tikzpicture}, \\
\end{aligned}
\end{equation*}
$\tilde{B}$ contatins all terms originating from local operators that overlap with the second $X$, i.e., those that start at positions $L+1, \dots, L+k-1$:
\begin{equation*}
    \begin{tikzpicture}[scale=0.5]
        \draw (-0.8,0)--(3.8,0);
        \foreach \x in {0,1,3} {
            \draw (\x,0) --++(0,1);
        }
        \node at (2,0.65) {$\dots$};
        \draw[blue,line width=2mm, line cap=round] (0,0)--(3,0) node[midway,below,anchor=north,black] {$\tilde{B}$};
    \end{tikzpicture} = 
    \begin{tikzpicture}[xscale=-0.55,baseline=0.4cm]
        \draw (9.8,0) -- (-0.8,0) -- (-0.8,1) -- (6.5,1) -- (6.5,2) -- (6,2);
        \foreach \x/\y in {0/1,1/1,3/1,4/1,6/1,7/0.5,9/0.5} {
            \draw (\x,0) --++(0,\y);
            \node[tensor,label=below:$A$] at (\x,0) {};
        }
        \foreach \x in {0,1} {
            \node[tensor,label=above:$\bar{A}$] at (\x,1) {};
        }
        \foreach \x in {(2,0), (2,1), (5,0), (5,0.5), (8,0)}{
            \node[fill=white] at \x {$\dots$};
        }
        \node[anchor=south] at (2,0.5) {$O-\epsilon$};
        \draw[line width=2mm, line cap=round] (0,0.5) -- (3,0.5);
        \draw[line width=2mm, line cap=round] (3,1) -- (6,1) node[midway,anchor=south] {$X$};
        \node[tensor,label=right:$\rho_R$] at (-0.8,0.5) {};
        \node[tensor,label=left:$\rho_R^{-1}$] at (6.5,1.5) {};
    \end{tikzpicture} + \dots +
    \begin{tikzpicture}[xscale=-0.55,baseline=0.4cm]
        \draw (4.8,0) -- (-3.8,0) -- (-3.8,1) -- (0.5,1) -- (0.5,2) -- (0,2);
        \foreach \x/\y in {-3/1,-1/1,0/1,1/1,3/1,4/0.5} {
            \draw (\x,0) --++(0,\y);
            \node[tensor,label=below:$A$] at (\x,0) {};
        }
        \node[fill=white] at (2,0) {$\dots$};
        \node[fill=white] at (2,1) {$\dots$};
        \node[anchor=south] at (2,0.5) {$O-\epsilon$};
        \draw[line width=2mm, line cap=round] (0,0.5) -- (3,0.5);
        \draw[line width=2mm, line cap=round] (0,1) -- (-3,1) node[midway,anchor=south] {$X$};
        \node[tensor,label=right:$\rho_R$] at (-3.8,0.5) {};
        \node[tensor,label=left:$\rho_R^{-1}$] at (0.5,1.5) {};
    \end{tikzpicture}.    
\end{equation*}
and $c$ is defined as above. As we have seen that $c=0$, we conclude that  
\begin{equation} \label{eq:B-Btilde}
    \begin{tikzpicture}[scale=0.6]
        \draw (-0.8,0)--(3.8,0);
        \foreach \x in {0,1,3} {
            \draw (\x,0) --++(0,1);
        }
        \node at (2,0.65) {$\dots$};
        \draw[blue,line width=2mm, line cap=round] (0,0)--(3,0) node[midway,below,anchor=north,black] {$\tilde{B}$};
    \end{tikzpicture} = - 
    \begin{tikzpicture}[scale=0.6]
        \draw (-0.8,0)--(3.8,0);
        \foreach \x in {0,1,3} {
            \draw (\x,0) --++(0,1);
        }
        \node at (2,0.65) {$\dots$};
        \draw[red,line width=2mm, line cap=round] (0,0)--(3,0) node[midway,below,anchor=north,black] {$B\vphantom{\tilde{B}}$};
    \end{tikzpicture} \ .
\end{equation}
The last (fourth) contraction we consider is contracting \cref{eq:eigenvalue_graphical} with a tensor similar to above, but with the two $X$ tensors separated by $k$ particles. That is, on \cref{eq:eigenvalue_graphical} we apply a contraction of the tensor product of two tensors $X$, and $N-2L-k\geq k-1$  tensors $\bar{A}$: 
\begin{itemize}
    \item $X$ on the positions $(1,\dots, L)$ and $(L+k+1,...,2L+k)$,
    \item and $\bar{A}$ on the positions $2L+k+1,\dots ,N$,
\end{itemize}
where the tensors are contracted as follows:
\begin{equation*}
    \begin{tikzpicture}[yscale=0.6,xscale=1.2]
        \draw  (4,0) -- (-0.6,0) -- (-0.6,1) -- (12.6,1) -- (12.6,0)--(5,0);
        \foreach \x/\l in {0/1,1/2,3/{L},6/{L+k+1\ },7/{\ L+k+2},9/{2L+k},10/{2L+k+1},12/N}{
            \draw (\x,0) --++ (0,-1);
            \node[anchor=north] at (\x,-1) {$\l$};
        }
        \draw[line width=2mm, line cap=round] (0,0)--(3,0) node[midway,anchor=south] {$X$};
        \draw[line width=2mm, line cap=round] (6,0)--(9,0) node[midway,anchor=south] {$X$};
        \node[tensor, label=above:$\bar{A}$] at (10,0) {};
        \node[tensor, label=above:$\bar{A}$] at (12,0) {};
        \node[tensor, label=above:$\rho_L^{-1}$] at (3.5,0) {};
        \node[tensor, label=above:$\rho_R^{-1}$] at (5.5,0) {};
        \foreach \x in {2,4.5,8,11} {
            \node[fill=white] at (\x,-0.7) {$\dots$};
        }
        \node[fill=white] at (11,0) {$\dots$};        
    \end{tikzpicture}
\end{equation*}
Evaluating the resulting equation, we arrive at
\begin{equation*}
    \begin{tikzpicture}[scale=0.6]
        \draw (-1.8,0)--(3.8,0);
        \foreach \x in {-1,0,1,3} {
            \draw (\x,0) --++(0,1);
        }
        \node at (2,0.65) {$\dots$};
        \node[tensor,label=below:$A$] at (-1,0) {};
        \draw[red,line width=2mm, line cap=round] (0,0)--(3,0) node[midway,below,anchor=north,black] {$B\vphantom{\tilde{B}}$};
    \end{tikzpicture} + 
    \begin{tikzpicture}[scale=0.55]
        \draw (-0.8,0)--(4.8,0);
        \foreach \x in {0,1,3,4} {
            \draw (\x,0) --++(0,1);
        }
        \node at (2,0.65) {$\dots$};
        \node[tensor,label=below:$A$] at (4,0) {};
        \draw[blue,line width=2mm, line cap=round] (0,0)--(3,0) node[midway,below,anchor=north,black] {$\tilde{B}$};
    \end{tikzpicture} + (N-2L-2k+1)c \cdot 
    \begin{tikzpicture}[scale=0.55]
        \draw (-0.8,0)--(3.8,0);
        \foreach \x in {0,1,3} {
            \draw (\x,0) --++(0,1);
            \node[tensor,label=below:$A$] at (\x,0) {};
        }
        \node[fill=white] at (2,0) {$\dots$};
    \end{tikzpicture} + 
    \begin{tikzpicture}[scale=0.55]
        \draw (-0.8,0)--(3.8,0);
        \foreach \x in {0,1,3} {
            \draw (\x,0) --++(0,1);
            \node[tensor,label=below:$A$] at (\x,0) {};
        }
        \node[fill=white] at (2,0) {$\dots$};
        \draw[line width=1.7mm,line cap=round] (0,0.5) -- (3,0.5);
        \node[anchor=south,inner sep=5pt] at (2,0.5) {$O-\epsilon$};
    \end{tikzpicture} = 0.
\end{equation*}
Finally, using $c=0$ and \cref{eq:B-Btilde}, we arrive at 
\begin{equation*}
    \begin{aligned}
    \begin{tikzpicture}[yscale=1, xscale=0.6]
        \draw (-0.5,0)--(3,0);
        \node[tensor,label=below:$A$] (l) at (0,0) {};
        \draw (l)--++(0,1);
        \node[tensor,label=below:$A$] (m) at (1,0) {};
        \draw (m)--++(0,1);
        \node[tensor,label=below:$A$] (r) at (2.5,0) {};
        \draw (r)--++(0,1);
        \draw[line width=1.7mm,line cap=round] (0,0.5) -- (2.5,0.5);
        \node[fill=white,inner sep = 1pt] at (1.75,0) {$\dots$};
        \node[anchor=south,inner sep=5pt] at (1.75,0.5) {$O-\epsilon$};
    \end{tikzpicture}  =
    \begin{tikzpicture}[yscale = 1, xscale=0.6]
        \draw (-0.5,0)--(3,0);
        \draw (1,0)--++(0,0.5);
        \draw (2.5,0)--++(0,0.5);
        \node[tensor,label=below:$A$] (r) at (0,0) {};
        \draw (r)--++(0,0.5);
        \draw[red,line width=1.7mm,line cap=round] (1,0) -- (2.5,0);
        \node[anchor=north,inner sep = 1.7mm] at (1.75,0) {$B$};
        \node at (1.75,0.3) {$\dots$};
    \end{tikzpicture}  -  
    \begin{tikzpicture}[yscale=0.7,xscale=0.6]
        \draw (-0.5,0)--(3,0);
        \node[tensor,label=below:$A$] (l) at (2.5,0) {};
        \draw (l)--++(0,0.5);
        \draw (0,0)--++(0,0.5);
        \draw (1.5,0)--++(0,0.5);
        \draw[red,line width=1.7mm,line cap=round] (0,0) -- (1.5,0);
        \node[anchor=north,inner sep = 1.7mm] at (0.75,0) {$B$};
        \node at (0.75,0.3) {$\dots$};
    \end{tikzpicture} \ ,
    \end{aligned}
\end{equation*}    
which is the desired equation.
\end{proof}

\subsection{Redundancies in the tensor \texorpdfstring{$B$}{B}}

In this Section we characterize the redundancy in the tensor $B$ that appears in Theorem~\ref{thm:main}. 
\begin{theorem}
Let the tensor $B$ a solution for \cref{eq:hABAp}, or equivalently, \cref{eq:hABAp graphical}, given the injective MPS tensor $A$ and local operator $O$. Then the tensor $\tilde{B}$ is also a solution of \cref{eq:hABAp graphical} for the same $A$ and $O$ if and only if  
\begin{equation*}
  \tilde{B}^{i_1 \dots i_k} = B^{i_1\dots i_k} + \lambda A^{i_1} \dots A^{i_k}.  
\end{equation*}
\end{theorem}

\begin{proof} ($\Leftarrow$) The term $\lambda A^{i_1} \dots A^{i_k}$ cancels from the telescopic difference (i.e., from the r.h.s.\ of \cref{eq:hABAp graphical}):
\begin{equation}\label{eq:ADDA}
    \begin{aligned}
    \begin{tikzpicture}[yscale = 1, xscale=0.6]
        \draw (-0.5,0)--(3,0);
        \draw (1,0)--++(0,0.5);
        \draw (2.5,0)--++(0,0.5);
        \node[tensor,label=below:$A$] (r) at (0,0) {};
        \draw (r)--++(0,0.5);
        \draw[red,line width=1.7mm,line cap=round] (1,0) -- (2.5,0);
        \node[anchor=north,inner sep = 1.7mm] at (1.75,0) {$B$};
        \node at (1.75,0.3) {$\dots$};
    \end{tikzpicture}  -
    \begin{tikzpicture}[yscale=0.7,xscale=0.6]
        \draw (-0.5,0)--(3,0);
        \node[tensor,label=below:$A$] (l) at (2.5,0) {};
        \draw (l)--++(0,0.5);
        \draw (0,0)--++(0,0.5);
        \draw (1.5,0)--++(0,0.5);
        \draw[red,line width=1.7mm,line cap=round] (0,0) -- (1.5,0);
        \node[anchor=north,inner sep = 1.7mm] at (0.75,0) {$B$};
        \node at (0.75,0.3) {$\dots$};
    \end{tikzpicture}  = \ 
    \end{aligned}
    \begin{aligned}
    \begin{tikzpicture}[yscale = 1, xscale=0.6]
        \draw (-0.5,0)--(3,0);
        \draw (1,0)--++(0,0.5);
        \draw (2.5,0)--++(0,0.5);
        \node[tensor,label=below:$A$] (r) at (0,0) {};
        \draw (r)--++(0,0.5);
        \draw[red,line width=1.7mm,line cap=round] (1,0) -- (2.5,0);
        \node[anchor=north,inner sep = 1.7mm] at (1.75,0) {$\tilde B$};
        \node at (1.75,0.3) {$\dots$};
    \end{tikzpicture}  -
    \begin{tikzpicture}[yscale=0.7,xscale=0.6]
        \draw (-0.5,0)--(3,0);
        \node[tensor,label=below:$A$] (l) at (2.5,0) {};
        \draw (l)--++(0,0.5);
        \draw (0,0)--++(0,0.5);
        \draw (1.5,0)--++(0,0.5);
        \draw[red,line width=1.7mm,line cap=round] (0,0) -- (1.5,0);
        \node[anchor=north,inner sep = 1.7mm] at (0.75,0) {$\tilde B$};
        \node at (0.75,0.3) {$\dots$};
    \end{tikzpicture},
    \end{aligned}
\end{equation}  
and thus if $B$ is a solution to \cref{eq:hABAp graphical}, then so is $\tilde B$. 

($\Rightarrow$) $D = B-\tilde{B}$ satisfies the equation     
\begin{equation}\label{eq:ADDA}
    \begin{aligned}
    \begin{tikzpicture}[yscale = 1, xscale=0.6]
        \draw (-0.5,0)--(3,0);
        \draw (1,0)--++(0,0.5);
        \draw (2.5,0)--++(0,0.5);
        \node[tensor,label=below:$A$] (r) at (0,0) {};
        \draw (r)--++(0,0.5);
        \draw[red,line width=1.7mm,line cap=round] (1,0) -- (2.5,0);
        \node[anchor=north,inner sep = 1.7mm] at (1.75,0) {$D$};
        \node at (1.75,0.3) {$\dots$};
    \end{tikzpicture}  -
    \begin{tikzpicture}[yscale=0.7,xscale=0.6]
        \draw (-0.5,0)--(3,0);
        \node[tensor,label=below:$A$] (l) at (2.5,0) {};
        \draw (l)--++(0,0.5);
        \draw (0,0)--++(0,0.5);
        \draw (1.5,0)--++(0,0.5);
        \draw[red,line width=1.7mm,line cap=round] (0,0) -- (1.5,0);
        \node[anchor=north,inner sep = 1.7mm] at (0.75,0) {$D$};
        \node at (0.75,0.3) {$\dots$};
    \end{tikzpicture}  =0 ,
    \end{aligned}
\end{equation}  
or equivalently,
\begin{equation}\label{eq:AD=DA}
    \begin{aligned}
    \begin{tikzpicture}[yscale = 1, xscale=0.6]
        \draw (-0.5,0)--(3,0);
        \draw (1,0)--++(0,0.5);
        \draw (2.5,0)--++(0,0.5);
        \node[tensor,label=below:$A$] (r) at (0,0) {};
        \draw (r)--++(0,0.5);
        \draw[red,line width=1.7mm,line cap=round] (1,0) -- (2.5,0);
        \node[anchor=north,inner sep = 1.7mm] at (1.75,0) {$D$};
        \node at (1.75,0.3) {$\dots$};
    \end{tikzpicture}  =
    \begin{tikzpicture}[yscale=0.7,xscale=0.6]
        \draw (-0.5,0)--(3,0);
        \node[tensor,label=below:$A$] (l) at (2.5,0) {};
        \draw (l)--++(0,0.5);
        \draw (0,0)--++(0,0.5);
        \draw (1.5,0)--++(0,0.5);
        \draw[red,line width=1.7mm,line cap=round] (0,0) -- (1.5,0);
        \node[anchor=north,inner sep = 1.7mm] at (0.75,0) {$D$};
        \node at (0.75,0.3) {$\dots$};
    \end{tikzpicture} .
    \end{aligned}
\end{equation}  
Repeatedly applying this equation, we arrive at 
\begin{equation*}
    \begin{tikzpicture}[xscale=0.6,yscale=0.5]
        \draw (-0.5,0) -- (8,0);
        \foreach \x in {0,2,3,5} {
            \node[tensor,label=below:$A$] at (\x,0) {};
            \draw (\x,0) --++(0,1);
        }
        \node[fill=white] at (1,0) {$\dots$}; 
        \node[fill=white] at (4,0) {$\dots$}; 
        \node[fill=white] at (6.75,0.6) {$\dots$}; 
        \draw (6,0) --++ (0,1);
        \draw (7.5,0) --++ (0,1);
        \draw[red,line width=1.7mm,line cap=round] (6,0) --++ (1.5,0);
        \draw[decorate,decoration={brace,amplitude=4pt},yshift=2pt] (-0.1,1) -- (2.1,1)
            node[midway,yshift=10pt] {$k$};
        \draw[decorate,decoration={brace,amplitude=4pt},yshift=2pt] (3-0.1,1) -- (5.1,1)
            node[midway,yshift=10pt] {$L$};
        \draw[decorate,decoration={brace,amplitude=4pt},yshift=2pt] (6-0.1,1) -- (7.6,1)
            node[midway,yshift=10pt] {$k$};
        \node[anchor=north,inner sep = 1.7mm] at (6.75,0) {$D$};
    \end{tikzpicture}   =
    \begin{tikzpicture}[xscale=-0.6,yscale=0.5]
        \draw (-0.5,0) -- (8,0);
        \foreach \x in {0,2,3,5} {
            \node[tensor,label=below:$A$] at (\x,0) {};
            \draw (\x,0) --++(0,1);
        }
        \node[fill=white] at (1,0) {$\dots$}; 
        \node[fill=white] at (4,0) {$\dots$}; 
        \node[fill=white] at (6.75,0.6) {$\dots$}; 
        \draw (6,0) --++ (0,1);
        \draw (7.5,0) --++ (0,1);
        \draw[red,line width=1.7mm,line cap=round] (6,0) --++ (1.5,0);
        \draw[decorate,decoration={brace,mirror,amplitude=4pt},yshift=2pt] (-0.1,1) -- (2.1,1)
            node[midway,yshift=10pt] {$k$};
        \draw[decorate,decoration={brace,mirror,amplitude=4pt},yshift=2pt] (3-0.1,1) -- (5.1,1)
            node[midway,yshift=10pt] {$L$};
        \draw[decorate,decoration={brace,mirror,amplitude=4pt},yshift=2pt] (6-0.1,1) -- (7.6,1)
            node[midway,yshift=10pt] {$k$};
        \node[anchor=north,inner sep = 1.7mm] at (6.75,0) {$D$};
    \end{tikzpicture}     ,
\end{equation*}
where $L$ is the injectivity length of $A$. Due to injectivity of $A$, there is a tensor $A^{-1}$ such that 
\begin{equation*}
    \begin{tikzpicture}[xscale=0.6,yscale=0.5,baseline=1.6mm]
        \draw (-0.5,0) -- (2.5,0);
        \draw (-0.5,1) -- (2.5,1);
        \foreach \x in {0,2} {
            \node[tensor,label=below:$A$] at (\x,0) {};
            \draw (\x,0) --++(0,1);
        }
        \node[anchor=south,inner sep = 2mm] at (1,1) {$A^{-1}$};
        \node[fill=white] at (1,0) {$\dots$}; 
        \node[fill=white] at (1,0.5) {$\dots$}; 
        \draw[line width=1.7mm,line cap=round] (0,1) --++ (2,0);
    \end{tikzpicture} = 
    \begin{tikzpicture}[xscale=0.6,yscale=0.5,baseline=1.6mm]
        \draw (0,0) -- (1,0) -- (1,1) -- (0,1);
        \draw (3,0) -- (2,0) -- (2,1) -- (3,1);
    \end{tikzpicture}.
\end{equation*}
Applying $A^{-1}$ on the middle $L$ tensors, and obtain
\begin{equation*}
    \begin{tikzpicture}[yscale=0.5,xscale=0.6]
        \draw (-0.5,0) -- (2.5,0);
        \foreach \x in {0,2} {
            \node[tensor,label=below:$A$] at (\x,0) {};
            \draw (\x,0) --++(0,1);
        }
        \node[fill=white] at (1,0) {$\dots$}; 
    \end{tikzpicture} \otimes 
    \begin{tikzpicture}[yscale = 1, xscale=0.6]
        \draw (0.5,0)--(3,0);
        \draw (1,0)--++(0,0.5);
        \draw (2.5,0)--++(0,0.5);
        \draw[red,line width=1.7mm,line cap=round] (1,0) -- (2.5,0);
        \node[anchor=north,inner sep = 1.7mm] at (1.75,0) {$D$};
        \node at (1.75,0.3) {$\dots$};
    \end{tikzpicture}  = 
    \begin{tikzpicture}[yscale = 1, xscale=0.6]
        \draw (0.5,0)--(3,0);
        \draw (1,0)--++(0,0.5);
        \draw (2.5,0)--++(0,0.5);
        \draw[red,line width=1.7mm,line cap=round] (1,0) -- (2.5,0);
        \node[anchor=north,inner sep = 1.7mm] at (1.75,0) {$D$};
        \node at (1.75,0.3) {$\dots$};
    \end{tikzpicture}  \otimes 
    \begin{tikzpicture}[yscale=0.5,xscale=0.6]
        \draw (-0.5,0) -- (2.5,0);
        \foreach \x in {0,2} {
            \node[tensor,label=below:$A$] at (\x,0) {};
            \draw (\x,0) --++(0,1);
        }
        \node[fill=white] at (1,0) {$\dots$}; 
    \end{tikzpicture}.
\end{equation*}
This implies that there is $\lambda\in\mathbb{C}$ such that 
\begin{equation*}
    \begin{tikzpicture}[yscale = 1, xscale=0.6]
        \draw (0.5,0)--(3,0);
        \draw (1,0)--++(0,0.5);
        \draw (2.5,0)--++(0,0.5);
        \draw[red,line width=1.7mm,line cap=round] (1,0) -- (2.5,0);
        \node[anchor=north,inner sep = 1.7mm] at (1.75,0) {$D$};
        \node at (1.75,0.3) {$\dots$};
    \end{tikzpicture}  = \lambda \cdot
    \begin{tikzpicture}[yscale=0.5, xscale=0.6]
        \draw (-0.5,0)--(3,0);
        \node[tensor,label=below:$A$] (l) at (0,0) {};
        \draw (l)--++(0,1);
        \node[tensor,label=below:$A$] (m) at (1,0) {};
        \draw (m)--++(0,1);
        \node[tensor,label=below:$A$] (r) at (2.5,0) {};
        \draw (r)--++(0,1);
        \node[fill=white,inner sep = 1pt] at (1.75,0) {$\dots$};
    \end{tikzpicture} .
\end{equation*}
\end{proof}
\subsection{The simplest case: two-body operators on injective MPS}

In this subsection we show a simplified proof of the main theorem for the case $L=1$ and $k=2$. This proof is slightly different from the general case because for $L=1$, we can apply the inverse tensor on all MPS tensors simultaneously.  

\begin{theorem}
    An MPS defined by an injective MPS tensor $A$ is the exact zero energy eigenstate of a translation invariant two-body operator $\mathcal{O}_N=\sum_i O_{[i,i+1]}$ for system size $N\geq 5$ if and only if there is an MPS tensor $B$ such that
    \begin{equation*}
        \begin{tikzpicture}
            \draw (-0.5,0)--(1.5,0);
            \node[tensor,label=below:$A$] (l) at (1,0) {};
            \draw (l)--++(0,1);
            \node[tensor,label=below:$A$] (r) at (0,0) {};
            \draw (r)--++(0,1);
            \draw[line width=2mm,line cap=round] (0,0.5) -- (1,0.5);
            \node[anchor=south,inner sep=5pt] at (0.5,0.5) {$O$};
        \end{tikzpicture} =
        \begin{tikzpicture}
            \draw (-0.5,0)--(1.5,0);
            \node[tensor,color=red,label=below:$B$] (l) at (1,0) {};
            \draw (l)--++(0,0.5);
            \node[tensor,label=below:$A$] (r) at (0,0) {};
            \draw (r)--++(0,0.5);
        \end{tikzpicture} - 
        \begin{tikzpicture}
            \draw (-0.5,0)--(1.5,0);
            \node[tensor,label=below:$A$] (l) at (1,0) {};
            \draw (l)--++(0,0.5);
            \node[tensor,color=red,label=below:$B$] (r) at (0,0) {};
            \draw (r)--++(0,0.5);
        \end{tikzpicture} \ \Longleftrightarrow      \sum_{i'j',\gamma} O_{ij,i'j'} A^{i'}_{\alpha\gamma}A^{j'}_{\gamma\beta} = \sum_{\gamma,\delta} A^i_{\alpha\gamma}B^j_{\gamma\beta}-B^i_{\alpha\delta}A^j_{\delta\beta}
    \end{equation*}    
\end{theorem}

\begin{proof}

The eigenvalue equation reads as
\begin{equation*}
    \sum_i \ 
    \begin{tikzpicture}
        \draw (-0.5,0) rectangle (4.5,-0.5);
        \draw (1,0.5)--(1,1);
        \draw (2,0.5)--(2,1);
        \draw[line width=2mm, line cap=round] (1,0.5)--(2,0.5);
        \node at (1.5,0.8) {$O$};
        \node at (1,1.2) {$i$};
        \node at (2,1.2) {$i+1$};
        \foreach\x in {0,1,2,4}{
            \node[tensor] (t) at (\x,0) {};
            \draw (t)--++(0,0.5);
        }
        \node[fill=white] at (3,0) {$\dots$};
    \end{tikzpicture} = 0.    
\end{equation*}

Let us apply the inverse of $A$ in the whole region. We obtain that
$
    \begin{tikzpicture}[baseline=4mm]
        \draw (0,0) -- (0,1);
        \draw (1,0) -- (1,1);
        \draw (0,0) -- (1,0) arc (-90:90:0.5) -- (0,1) arc (90:270:0.5);
        \node[tensor] at (0,0) {};
        \node[tensor] at (1,0) {};
        \node[tensor] at (0,1) {};
        \node[tensor] at (1,1) {};
        \draw[line width=2mm, line cap=round] (0,0.5)--(1,0.5);        
    \end{tikzpicture} = 0.
$
Let us apply the inverse tensor on every site except one. We obtain
\begin{equation*}
    \begin{tikzpicture}[baseline=4mm]
        \draw (0,0) -- (0,1);
        \draw (1,0) -- (1,1.5);
        \draw[rounded corners =3mm] (0.5,1.5) -- (0.5,1) -- (-0.5,1)--(-0.5,0)--(1.5,0)--(1.5,1.5);
        \node[tensor] at (0,0) {};
        \node[tensor] at (1,0) {};
        \node[tensor] at (0,1) {};
        \draw[line width=2mm, line cap=round] (0,0.5)--(1,0.5);        
    \end{tikzpicture} +
    \begin{tikzpicture}[baseline=4mm]
        \draw (1,0) -- (1,1);
        \draw (0,0) -- (0,1.5);
        \draw[rounded corners =3mm] (0.5,1.5) -- (0.5,1) --(1.5,1)--(1.5,0)--(-0.5,0)--(-0.5,1.5);
        \node[tensor] at (0,0) {};
        \node[tensor] at (1,0) {};
        \node[tensor] at (1,1) {};
        \draw[line width=2mm, line cap=round] (0,0.5)--(1,0.5);        
    \end{tikzpicture} +
    \begin{tikzpicture}[baseline=4mm]
        \draw (0,0) -- (0,1);
        \draw (1,0) -- (1,1);
        \draw (2,0) -- (2,1.5);
        \draw[rounded corners=3mm] (1.4,1.5)--(1.4,1)--(-0.5,1) -- (-0.5,0)--(2.5,0) -- (2.5,1.5);
        \node[tensor] at (0,0) {};
        \node[tensor] at (1,0) {};
        \node[tensor] at (0,1) {};
        \node[tensor] at (1,1) {};
        \node[tensor] at (2,0) {};
        \draw[line width=2mm, line cap=round] (0,0.5)--(1,0.5);        
    \end{tikzpicture} +
    \begin{tikzpicture}[baseline=4mm]
        \draw (0,0) -- (0,1);
        \draw (1,0) -- (1,1);
        \draw (-1,0) -- (-1,1.5);
        \draw[rounded corners=3mm] (-0.5,1.5)--(-0.5,1)--(1.5,1) -- (1.5,0)--(-1.5,0) -- (-1.5,1.5);
        \node[tensor] at (-1,0) {};
        \node[tensor] at (0,0) {};
        \node[tensor] at (1,0) {};
        \node[tensor] at (0,1) {};
        \node[tensor] at (1,1) {};
        \draw[line width=2mm, line cap=round] (0,0.5)--(1,0.5);        
    \end{tikzpicture} = 0.
\end{equation*}
Let us define now 
\begin{equation*}
    \begin{tikzpicture}[baseline=4mm]
        \draw (0,0) -- (0,0.5);
        \draw[rounded corners =3mm] (-0.5,0.5)--(-0.5,0)--(0.5,0)--(0.5,0.5);
        \node[tensor,red,label=below:$B$] at (0,0) {};
    \end{tikzpicture} =
    \begin{tikzpicture}[baseline=4mm]
        \draw (0,0) -- (0,1);
        \draw (1,0) -- (1,1.5);
        \draw[rounded corners =3mm] (0.5,1.5) -- (0.5,1) --(-0.5,1)--(-0.5,0)--(1.5,0)--(1.5,1.5);
        \node[tensor] at (0,0) {};
        \node[tensor] at (1,0) {};
        \node[tensor] at (0,1) {};
        \draw[line width=2mm, line cap=round] (0,0.5)--(1,0.5);        
    \end{tikzpicture} +
    \begin{tikzpicture}[baseline=4mm]
        \draw (0,0) -- (0,1);
        \draw (1,0) -- (1,1);
        \draw (2,0) -- (2,1.5);
        \draw[rounded corners=3mm] (1.4,1.5)--(1.4,1)--(-0.5,1) -- (-0.5,0)--(2.5,0) -- (2.5,1.5);
        \node[tensor] at (0,0) {};
        \node[tensor] at (1,0) {};
        \node[tensor] at (0,1) {};
        \node[tensor] at (1,1) {};
        \node[tensor] at (2,0) {};
        \draw[line width=2mm, line cap=round] (0,0.5)--(1,0.5);        
    \end{tikzpicture} = -
    \begin{tikzpicture}[baseline=4mm]
        \draw (1,0) -- (1,1);
        \draw (0,0) -- (0,1.5);
        \draw[rounded corners =3mm] (0.5,1.5) -- (0.5,1) --(1.5,1)--(1.5,0)--(-0.5,0)--(-0.5,1.5);
        \node[tensor] at (0,0) {};
        \node[tensor] at (1,0) {};
        \node[tensor] at (1,1) {};
        \draw[line width=2mm, line cap=round] (0,0.5)--(1,0.5);        
    \end{tikzpicture} -
    \begin{tikzpicture}[baseline=4mm]
        \draw (0,0) -- (0,1);
        \draw (1,0) -- (1,1);
        \draw (-1,0) -- (-1,1.5);
        \draw[rounded corners=3mm] (-0.5,1.5)--(-0.5,1)--(1.5,1) -- (1.5,0)--(-1.5,0) -- (-1.5,1.5);
        \node[tensor] at (-1,0) {};
        \node[tensor] at (0,0) {};
        \node[tensor] at (1,0) {};
        \node[tensor] at (0,1) {};
        \node[tensor] at (1,1) {};
        \draw[line width=2mm, line cap=round] (0,0.5)--(1,0.5);        
    \end{tikzpicture}.
\end{equation*}
Finally we can apply the inverse on every site except for two. We obtain 
\begin{equation*}
    \begin{tikzpicture}[baseline=4mm,xscale=0.7]
        \draw (0,0) -- (0,1.5);
        \draw (1,0) -- (1,1.5);
        \draw[rounded corners =3mm] (-0.5,1.5) --(-0.5,0)--(1.5,0)--(1.5,1.5);
        \node[tensor] at (0,0) {};
        \node[tensor] at (1,0) {};
        \draw[line width=2mm, line cap=round] (0,0.5)--(1,0.5);        
    \end{tikzpicture} +
    \begin{tikzpicture}[baseline=4mm,xscale=0.7]
        \draw (0,0) -- (0,1);
        \draw (1,0) -- (1,1.5);
        \draw (2,0) -- (2,1.5);
        \draw[rounded corners =3mm] (0.5,1.5) -- (0.5,1) --(-0.5,1)--(-0.5,0)--(2.5,0)--(2.5,1.5);
        \node[tensor] at (0,0) {};
        \node[tensor] at (1,0) {};
        \node[tensor] at (2,0) {};
        \node[tensor] at (0,1) {};
        \draw[line width=2mm, line cap=round] (0,0.5)--(1,0.5);        
    \end{tikzpicture} +
    \begin{tikzpicture}[baseline=4mm,xscale=0.7]
        \draw (1,0) -- (1,1);
        \draw (0,0) -- (0,1.5);
        \draw (-1,0) -- (-1,1.5);
        \draw[rounded corners =3mm] (0.5,1.5) -- (0.5,1) --(1.5,1)--(1.5,0)--(-1.5,0)--(-1.5,1.5);
        \node[tensor] at (0,0) {};
        \node[tensor] at (1,0) {};
        \node[tensor] at (-1,0) {};
        \node[tensor] at (1,1) {};
        \draw[line width=2mm, line cap=round] (0,0.5)--(1,0.5);        
    \end{tikzpicture} +
    \begin{tikzpicture}[baseline=4mm,xscale=0.7]
        \draw (0,0) -- (0,1);
        \draw (1,0) -- (1,1);
        \draw (2,0) -- (2,1.5);
        \draw (3,0) -- (3,1.5);
        \draw[rounded corners=3mm] (1.5,1.5)--(1.5,1)--(-0.5,1) -- (-0.5,0)--(3.5,0) -- (3.5,1.5);
        \node[tensor] at (0,0) {};
        \node[tensor] at (1,0) {};
        \node[tensor] at (0,1) {};
        \node[tensor] at (1,1) {};
        \node[tensor] at (2,0) {};
        \node[tensor] at (3,0) {};
        \draw[line width=2mm, line cap=round] (0,0.5)--(1,0.5);        
    \end{tikzpicture} +
    \begin{tikzpicture}[baseline=4mm,xscale=0.7]
        \draw (0,0) -- (0,1);
        \draw (1,0) -- (1,1);
        \draw (-1,0) -- (-1,1.5);
        \draw (-2,0) -- (-2,1.5);
        \draw[rounded corners=3mm] (-0.5,1.5)--(-0.5,1)--(1.5,1) -- (1.5,0)--(-2.5,0) -- (-2.5,1.5);
        \node[tensor] at (-2,0) {};
        \node[tensor] at (-1,0) {};
        \node[tensor] at (0,0) {};
        \node[tensor] at (1,0) {};
        \node[tensor] at (0,1) {};
        \node[tensor] at (1,1) {};
        \draw[line width=2mm, line cap=round] (0,0.5)--(1,0.5);        
    \end{tikzpicture} = 0.
\end{equation*}
Using the definition of the red tensor $B$, we obtain 
\begin{equation*}
    \begin{tikzpicture}[baseline=4mm]
        \draw (0,0) -- (0,1);
        \draw (1,0) -- (1,1);
        \draw[rounded corners =3mm] (-0.5,1) --(-0.5,0)--(1.5,0)--(1.5,1);
        \node[tensor] at (0,0) {};
        \node[tensor] at (1,0) {};
        \draw[line width=2mm, line cap=round] (0,0.5)--(1,0.5);        
    \end{tikzpicture} +
    \begin{tikzpicture}[baseline=4mm]
        \draw (1,0) -- (1,1);
        \draw (2,0) -- (2,1);
        \draw[rounded corners =3mm] (0.5,1) -- (0.5,0) --(2.5,0)--(2.5,1);
        \node[tensor,red] at (1,0) {};
        \node[tensor] at (2,0) {};
    \end{tikzpicture} -
    \begin{tikzpicture}[baseline=4mm]
        \draw (1,0) -- (1,1);
        \draw (2,0) -- (2,1);
        \draw[rounded corners =3mm] (0.5,1) -- (0.5,0) --(2.5,0)--(2.5,1);
        \node[tensor] at (1,0) {};
        \node[tensor,red] at (2,0) {};
    \end{tikzpicture} =0,
\end{equation*}
which is the desired equation.
\end{proof}
\subsection{The case of one-body operators}

\begin{theorem}
    An MPS defined by an injective MPS tensor $A$ is the exact zero energy eigenstate of a translation invariant one-body operator $\mathcal{O}_N=\sum_i O_{[i]}$ for system size $N\geq 3$ if and only if there is a virtual matrix $B$ such that
    \begin{equation*}
        \begin{tikzpicture}
            \draw (-0.5,0)--(0.5,0);
            \node[tensor,label=below:$A$] (r) at (0,0) {};
            \node[tensor,label=right:$O$]  at (0,0.5) {};
            \draw (r)--++(0,1);
        \end{tikzpicture} =
        \begin{tikzpicture}
            \draw (-0.5,0)--(1.5,0);
            \node[tensor,color=red,label=below:$B$] (l) at (1,0) {};
            \node[tensor,label=below:$A$] (r) at (0,0) {};
            \draw (r)--++(0,0.5);
        \end{tikzpicture} - 
        \begin{tikzpicture}
            \draw (-0.5,0)--(1.5,0);
            \node[tensor,label=below:$A$] (l) at (1,0) {};
            \draw (l)--++(0,0.5);
        \node[tensor,color=red,label=below:$B$] (r) at (0,0) {};
        \end{tikzpicture} \ \Longleftrightarrow      \sum_{j} O_{i,j} A^{j}_{\alpha\beta} = \sum_{\gamma} A^i_{\alpha\gamma}B_{\gamma\beta}-B_{\alpha\gamma}A^i_{\gamma\beta}
    \end{equation*}    
\end{theorem}

\begin{proof}

The eigenvalue equation reads as
\begin{equation}
    \sum_i \ 
    \begin{tikzpicture}
        \draw (-0.5,0) rectangle (4.5,-0.5);
        \draw (1,0.5)--(1,1);
        \draw (2,0.5)--(2,1);
\node[tensor,label=right:$O$]  at (1,0.5) {};
        \node at (1,1.2) {$i$};
        \foreach\x in {0,1,2,4}{
            \node[tensor] (t) at (\x,0) {};
            \draw (t)--++(0,0.5);
        }
        \node[fill=white] at (3,0) {$\dots$};
    \end{tikzpicture} = 0.
\end{equation}
Let us apply the inverse of $A$ in the whole region. We obtain that
$
    \begin{tikzpicture}[baseline=4mm]
        \draw (0,0) -- (0,1);
\draw (-0.1,0) -- (0.1,0) arc (-90:90:0.5) -- (-0.1,1) arc (90:270:0.5);
        \node[tensor] at (0,0) {};
        \node[tensor] at (0,1) {};
    \node[tensor,label=right:$O$]  at (0,0.5) {};       
    \end{tikzpicture} = 0.
$
Let us apply the inverse tensor on every site but leave on edge open. We obtain
\begin{equation*}
    \begin{tikzpicture}[baseline=4mm]
        \draw (0,0) -- (0,1);
            \draw[rounded corners =3mm] (0.5,1.5) -- (0.5,1) -- (-0.5,1)--(-0.5,0)--(1,0)--(1,1.5);
        \node[tensor] at (0,0) {};
        \node[tensor] at (0,1) {};
           \node[tensor,label=right:$O$]  at (0,0.5) {};         
    \end{tikzpicture} +
   \begin{tikzpicture}[baseline=4mm]
        \draw (1,0) -- (1,1);
        \draw[rounded corners =3mm] (0.5,1.5) -- (0.5,1) --(1.5,1)--(1.5,0)--(0,0)--(0,1.5);
        \node[tensor] at (1,0) {};
        \node[tensor] at (1,1) {};
        \node[tensor,label=right:$O$]  at (1,0.5) {};         
    \end{tikzpicture}  = 0.
\end{equation*}
Let us define now 
\begin{equation*}
    \begin{tikzpicture}[baseline=4mm]
        \draw[rounded corners =3mm] (-0.5,0.5)--(-0.5,0)--(0.5,0)--(0.5,0.5);
        \node[tensor,red,label=below:$B$] at (0,0) {};
    \end{tikzpicture} =
    \begin{tikzpicture}[baseline=4mm]
        \draw (0,0) -- (0,1);
            \draw[rounded corners =3mm] (0.5,1.5) -- (0.5,1) -- (-0.5,1)--(-0.5,0)--(1,0)--(1,1.5);
        \node[tensor] at (0,0) {};
        \node[tensor] at (0,1) {};
           \node[tensor,label=right:$O$]  at (0,0.5) {};         
    \end{tikzpicture} = \ - \ 
   \begin{tikzpicture}[baseline=4mm]
        \draw (1,0) -- (1,1);
        \draw[rounded corners =3mm] (0.5,1.5) -- (0.5,1) --(1.5,1)--(1.5,0)--(0,0)--(0,1.5);
        \node[tensor] at (1,0) {};
        \node[tensor] at (1,1) {};
        \node[tensor,label=right:$O$]  at (1,0.5) {};         
    \end{tikzpicture}.
\end{equation*}
Finally we can apply the inverse on every site except for one site. We obtain 

  \begin{equation*}
    \begin{tikzpicture}[baseline=4mm]
        \draw[rounded corners =3mm] (-0.5,1)--(-0.5,0)--(0.5,0)--(0.5,1);
        \node[tensor] at (0,0) {};
         \draw (0,0) -- (0,1);
         \node[tensor,label=right:$O$]  at (0,0.5) {};
    \end{tikzpicture} +
    \begin{tikzpicture}[baseline=4mm]
    \node[tensor] at (1,0) {};
         \draw (1,0) -- (1,1.5);
        \draw (0,0) -- (0,1);
            \draw[rounded corners =3mm] (0.5,1.5) -- (0.5,1) -- (-0.5,1)--(-0.5,0)--(1.5,0)--(1.5,1.5);
        \node[tensor] at (0,0) {};
        \node[tensor] at (0,1) {};
           \node[tensor,label=right:$O$]  at (0,0.5) {};         
    \end{tikzpicture} +
   \begin{tikzpicture}[baseline=4mm]
        \draw (1,0) -- (1,1);
        \draw[rounded corners =3mm] (0.5,1.5) -- (0.5,1) --(1.5,1)--(1.5,0)--(-0.5,0)--(-0.5,1.5);
        \node[tensor] at (1,0) {};
        \node[tensor] at (1,1) {};
        \node[tensor,label=right:$O$]  at (1,0.5) {};       
        \node[tensor] at (0,0) {};
         \draw (0,0) -- (0,1.5);
    \end{tikzpicture}
     = 0
\end{equation*}

Using the definition of the virtual matrix $B$, we obtain the desired equation.
\end{proof}

The previous result can be generalized, and proved similarly, for one-body operators acting on any n-dimensional square lattice for injective tensor network states as follows:

\begin{theorem} Let $|\psi_A\rangle$ be an injective tensor network state on the $n$-dimensional square lattice given by the tensor $A^i_{\alpha_1\beta_1\cdots \alpha_k\beta_k \cdots \alpha_n\beta_n}$, where $(\alpha_k,\beta_k)$ are the virtual indices corresponding to the $k$-th dimension, one per each direction. Then, if a one-body operator annihilate the state, i.e. $(\sum_i O_{[i]})|\psi_A\rangle = 0$ then there exist $n$ virtual matrices, $B^{[k]}\in \mathbb{C}^{D^2_k}$ for $k=1,\cdots, n$, such that: 
$$ 
\sum_j O_{i,j} A^{j}_{\alpha_1\beta_1\cdots \alpha_n \beta_n} = \sum_{k\in[1,n],\gamma_k\in [1,D_k]} A^i_{\alpha_1\beta_1\cdots \alpha_k\gamma_k \cdots \alpha_n\beta_n}B^{[k]}_{\gamma_k\beta_k}-B^{[k]}_{\alpha_k\gamma_k}A^i_{\alpha_1\beta_1\cdots \gamma_k \beta_k \cdots \alpha_n\beta_n}
$$
\end{theorem}

\section{Application (3). Example: The XXZ model}

In this section we provide more details on the calculation of the symmetries and the (higher) transfer matrices of the XXZ Hamiltonian. Both are MPOs that commute with the Hamiltonian; That is, we are looking for solutions (that is, the black and red MPO tensors) of the equation 
\begin{equation}\label{MPOsym2}
    \begin{tikzpicture}[scale=0.7]
        \draw (-0.5,0)--(1.5,0);
        \node[tensor] (l) at (1,0) {};
        \draw (l)--++(0,1);
        \draw (l)--++(0,-0.5);
        \node[tensor] (r) at (0,0) {};
        \draw (r)--++(0,1);
        \draw (r)--++(0,-0.5);
        \draw[line width=2mm,line cap=round] (0,0.5) -- (1,0.5);
        \node[anchor=south,inner sep=5pt] at (0.5,0.5) {$h$};
    \end{tikzpicture} 
    \ - \ 
    \begin{tikzpicture}[scale=0.7]
        \draw (-0.5,0)--(1.5,0);
        \node[tensor] (l) at (1,0) {};
        \draw (l)--++(0,-1);
        \draw (l)--++(0,0.5);
        \node[tensor] (r) at (0,0) {};
        \draw (r)--++(0,-1);
        \draw (r)--++(0,0.5);
        \draw[line width=2mm,line cap=round] (0,-0.5) -- (1,-0.5);
        \node[anchor=north,inner sep=5pt] at (0.5,-0.5) {$h$};
    \end{tikzpicture}
    =
    \begin{tikzpicture}[scale=0.7]
        \draw (-0.5,0)--(1.5,0);
        \node[tensor,color=red] (l) at (1,0) {};
        \draw (l)--++(0,0.5);
        \draw (l)--++(0,-0.5);
        \node[tensor] (r) at (0,0) {};
        \draw (r)--++(0,0.5);
        \draw (r)--++(0,-0.5);
    \end{tikzpicture} - 
    \begin{tikzpicture}[scale=0.7]
        \draw (-0.5,0)--(1.5,0);
        \node[tensor] (l) at (1,0) {};
        \draw (l)--++(0,0.5);
        \draw (l)--++(0,-0.5);
        \node[tensor,color=red] (r) at (0,0) {};
        \draw (r)--++(0,0.5);
        \draw (r)--++(0,-0.5);
    \end{tikzpicture} \ .
\end{equation} 
We write the black and the red tensors as 
\begin{align*}
    \begin{tikzpicture}[baseline=-1mm, scale=0.3]
        \draw (-1,0) -- (1,0);
        \draw (0,-1) -- (0,1);
        \node[tensor] at (0,0) {};
    \end{tikzpicture} &= 
    \ket{0}\bra{0} \otimes a + \ket{0}\bra{1}\otimes b + \ket{1}\bra{0}\otimes c + \ket{1}\bra{1}\otimes d, \\
    \begin{tikzpicture}[baseline=-1mm, scale=0.3]
        \draw (-1,0) -- (1,0);
        \draw (0,-1) -- (0,1);
        \node[tensor,red] at (0,0) {};
    \end{tikzpicture} &= 
    \ket{0}\bra{0} \otimes \tilde a + \ket{0}\bra{1}\otimes \tilde b + \ket{1}\bra{0}\otimes \tilde c + \ket{1}\bra{1}\otimes \tilde d,
\end{align*}
where the matrices $a,b,c,d$ and $\tilde a, \tilde b, \tilde c, \tilde d$ act on the virtual d.o.f. We are looking for matrices $a,b,c,d$ and $\tilde a, \tilde b, \tilde c, \tilde d$ such that the equation Eq.~\eqref{MPOsym2} holds. Below we expand the equation component-wise, that is, in the computational basis of the physical indices. For that, notice first that the Hamiltonian reads 
\begin{equation*}
    h = \Delta \ket{00}\bra{00} - \Delta \ket{01}\bra{01} - \Delta \ket{10}\bra{10} + 2 \ket{10}\bra{01} + 2 \ket{01}\bra{10} + \Delta \ket{11}\bra{11}. 
\end{equation*}
That is, 
\begin{align*}
    &\bra{00}h = \Delta \bra{00}, \quad \bra{01}h = 2\bra{10} - \Delta \bra{01}, \quad \bra{10}h = 2\bra{01} - \Delta \bra{10}, \quad \bra{11} h = \Delta \bra{11},\\
    &h\ket{00} = \Delta \ket{00}, \quad h\ket{01} = 2\ket{10} - \Delta \ket{01}, \quad h\ket{10} = 2\ket{01} - \Delta \ket{10}, \quad h\ket{11} = \Delta \ket{11}.
\end{align*}
Using these formulas, we can read off the 16 components of Eq.~\eqref{MPOsym2}, in alphabetical ordering:
\begin{equation}\label{eq:XXZ_main}
\begin{aligned}
 0 &= a\tilde{a} - \tilde{a}a 
&\quad 2\Delta ab - 2ba &= a \tilde{b} - \tilde{a} b 
&\quad 2\Delta ba - 2ab &= b \tilde{a} - \tilde{b} a 
&\quad 0 &= b \tilde b - \tilde b b \\[4pt]
2 ca - 2\Delta ac &= a \tilde c - \tilde a c 
&\quad 2 cb- 2bc &= a \tilde d - \tilde a d
&\quad 2 da - 2ad &= b \tilde c - \tilde b c
&\quad 2 d b - 2\Delta bd &= b \tilde d - \tilde b d \\[4pt]
2 ac - 2\Delta ca & = c \tilde a - \tilde c a
&\quad 2 ad - 2da &= c \tilde b - \tilde c b
&\quad 2 bc - 2 cb &= d \tilde a - \tilde d a
&\quad 2bd - 2\Delta db &= d \tilde b - \tilde d b \\[4pt]
0 &= c \tilde c - \tilde c c
&\quad 2 \Delta cd - 2 dc &= c \tilde d - \tilde c d
&\quad 2 \Delta dc - 2 cd &= d \tilde c - \tilde d c
&\quad 0 &= d \tilde d - \tilde d d.
\end{aligned}
\end{equation}

Both the $U_q(sl_2)$ symmetries and the higher transfer matrices arise as special solutions to these equations. In the following we describe these solutions.

\subsection{Symmetries of the XXZ model}

We will obtain a special solution to \cref{eq:XXZ_main} by setting $\tilde a =0$,  $\tilde d = 0$, $\tilde b = \lambda b$ and $\tilde c = - \lambda c$. Substituting this ansatz in the previous equations, we obtain
\begin{align*}
\begin{aligned}
 0 &=0 
&\quad 2\Delta ab - 2ba &= \lambda a b  
&\quad 2\Delta ba - 2ab &=  - \lambda b a 
&\quad 0 &= 0 \\[4pt]
2 ca - 2\Delta ac &=  - \lambda a  c 
&\quad 2 cb- 2bc &= 0
&\quad 2 da - 2ad &= -2 \lambda b  c 
&\quad 2 d b - 2\Delta bd &=  - \lambda b d \\[4pt]
2 ac - 2\Delta ca & = \lambda c a
&\quad 2 ad - 2da &= 2 \lambda c  b 
&\quad 2 bc - 2 cb &= 0
&\quad 2bd - 2\Delta db &= \lambda d  b  \\[4pt]
0 &= 0
&\quad 2 \Delta cd - 2 dc &=  \lambda c d
&\quad 2 \Delta dc - 2 cd &=  - \lambda d  c 
&\quad 0 &= 0.
\end{aligned}
\end{align*}
That is, $bc = cb$, $ad - da = \lambda bc$, and
\begin{equation*}
\begin{aligned}
    ab &= \left(\frac{2\Delta - \lambda }{2}\right)^{-1}ba = \frac{2\Delta + \lambda}{2} ba,\\
    ac &= \left(\frac{2\Delta -\lambda}{2}\right)^{-1} ca = \frac{2\Delta + \lambda}{2} ca, \\
    bd &= \left(\frac{2\Delta -\lambda}{2}\right)^{-1} db = \frac{2\Delta + \lambda}{2} db,\\
    cd &= \left(\frac{2\Delta -\lambda}{2}\right)^{-1} dc = \frac{2\Delta + \lambda}{2} dc.
\end{aligned}
\end{equation*}
That is, we have found solutions as long as 
\begin{equation*}
    q:=\frac{2\Delta - \lambda }{2} = \left(\frac{2\Delta + \lambda}{2}\right)^{-1},
\end{equation*}
or equivalently, if we can write $\Delta = \tfrac{q+q^{-1}}{2}$ and $\lambda = -(q-q^{-1})$. In this case $a,b,c,d$ satisfy
\begin{align*}
    ab = & q^{-1} ba \ , \quad ac = q^{-1} ca \ , \quad bd = q^{-1} db \ , \\
    cd =& q^{-1}dc \ ,\quad bc = cb \ , \quad ad-da = -(q-q^{-1})bc \ ,
\end{align*} 
i.e., they are the generators of the algebra $\mathrm{SL}_q(2)$.

So far we have seen that the XXZ Hamiltonian with periodic boundary conditions, $H=\sum_{i=0}^{n-1} h_{i,i+1}$, where the indexing is taken $\mod n$, and the local term $h$ is defined as 
\begin{equation*}
    h = X \otimes X + Y\otimes Y + \frac{q+q^{-1}}{2} Z\otimes Z,
\end{equation*}
has MPO-like symmetries: the MPO is taken with periodic boundary conditions, and it is defined by the tensor
\begin{align}
    \begin{tikzpicture}[baseline=-1mm, scale=0.3]
        \draw (-1,0) -- (1,0);
        \draw (0,-1) -- (0,1);
        \node[tensor] at (0,0) {};
    \end{tikzpicture} &= 
    \ket{0}\bra{0} \otimes a + \ket{0}\bra{1}\otimes b + \ket{1}\bra{0}\otimes c + \ket{1}\bra{1}\otimes d, \label{eq:mpo_tensor_ansatz}
\end{align}
where $a,b,c,d$ are the generators of $SL_q(2)$. A few remarks are in place. First of all, to obtain an actual MPO with finite bond dimension, we need to take a finite dimensional representation of $SL_q(2)$, which is only possible (non-tivially) when $q$ is a root of unity. Second, the $XXZ$ model is known to have $U_q(sl_2)$ symmetries -- we still have to show that the MPOs obtained are representations of this algebra. (Note that $a,b,c,d \in SL_q(2)$ are on the virtual level of the MPO tensor, whereas here we ask whether the MPO as an operator on $(\mathbb{C}^2)^{\otimes n}$ is a representation of $U_q(sl_2)$). Third, with periodic boundary conditions we do not recover the whole $U_q(sl_2)$ symmetry. In the following we will show how to overcome these three problems and recover the full $U_q(sl_2)$ symmetry of the XXZ model.

As a first step, notice that we have obtained that the red tensor in \cref{MPOsym2} is defined as  
\begin{align*}
    \begin{tikzpicture}[baseline=-1mm, scale=0.3]
        \draw (-1,0) -- (1,0);
        \draw (0,-1) -- (0,1);
        \node[tensor,red] at (0,0) {};
    \end{tikzpicture} &= 
    \ket{0}\bra{1}\otimes \lambda b - \ket{1}\bra{0}\otimes \lambda c =  [k, \ket{0}\bra{0}] \otimes a + [k,\ket{0}\bra{1}]\otimes b + [k,\ket{1}\bra{0}]\otimes c + [k,\ket{1}\bra{1}]\otimes d, 
\end{align*}
where $\lambda = (q-q^{-1})$ and 
\begin{equation*}
    k = q^Z = q \ket{0}\bra{0} + q^{-1} \ket{1}\bra{1} = \begin{pmatrix} q & 0 \\ 0 & q^{-1}\end{pmatrix}.
\end{equation*}
We can use this decomposition and add $q^z$ to the Hamiltonian as a telescopic sum, i.e., we define 
\begin{equation*}
    \tilde{h} =  X \otimes X + Y\otimes Y + \frac{q+q^{-1}}{2} Z\otimes Z + q^Z \otimes 1 - 1\otimes q^Z,
\end{equation*}
which is the Pasquier-Saleur Hamiltonian \cite{PASQUIER1990523} with the well known quantum group symmetries $U_q(sl_2)$. Note that this ``Hamiltonian'' term is only Hermitian for $\Delta\geq 1$, where $q$ is real. Because of the telescopic nature of the local terms, the Pasquier-Saleur Hamiltonian on $N$ particles is the same as the XXZ Hamiltonian if we take periodic boundaries,
\begin{equation*}
    \sum_{i=1}^N h_i = \sum_{i=1}^N \tilde{h}_i,
\end{equation*}
where $h_N$ and $\tilde{h}_N$ act on the particles $N$ and $1$. The advantage of the Pasquier-Saleur Hamiltonian is that on OBC it has the full $U_q(sl_2)$ symmetries. For that, notice first that using $\tilde{h}$ we can rewrite \cref{MPOsym2} as
\begin{equation}\label{MPOsym3}
    \begin{tikzpicture}[scale=0.7]
        \draw (-0.5,0)--(1.5,0);
        \node[tensor] (l) at (1,0) {};
        \draw (l)--++(0,1);
        \draw (l)--++(0,-0.5);
        \node[tensor] (r) at (0,0) {};
        \draw (r)--++(0,1);
        \draw (r)--++(0,-0.5);
        \draw[line width=2mm,line cap=round] (0,0.5) -- (1,0.5);
        \node[anchor=south,inner sep=5pt] at (0.5,0.5) {$\tilde{h}$};
    \end{tikzpicture} 
    \ - \ 
    \begin{tikzpicture}[scale=0.7]
        \draw (-0.5,0)--(1.5,0);
        \node[tensor] (l) at (1,0) {};
        \draw (l)--++(0,-1);
        \draw (l)--++(0,0.5);
        \node[tensor] (r) at (0,0) {};
        \draw (r)--++(0,-1);
        \draw (r)--++(0,0.5);
        \draw[line width=2mm,line cap=round] (0,-0.5) -- (1,-0.5);
        \node[anchor=north,inner sep=5pt] at (0.5,-0.5) {$\tilde{h}$};
    \end{tikzpicture}
    = 0 .
\end{equation} 
Due to this equation, the OBC Hamiltonian $H_{OBC}^n = \sum_{i=0}^{n-1} h_{i,i+1}$ commutes with any MPO given above, with aribtrary boundary condition: it has the symmetry $[H_{OBC}^n,O_n(x)]$ with
\begin{equation*}
    O_n(x)  =
    \begin{tikzpicture}[xscale=0.7, yscale=0.4, font=\scriptsize]
        \draw[rounded corners] (-1.5,0) rectangle (3,-1.2);
        \node[tensor,label=below:$x$] (l) at (-1,0) {};
        \draw (0,0)--(2.5,0);
        \node[tensor] (l) at (0,0) {};
        \draw (0,-1)--++(0,2);
        \node[tensor] (m) at (1,0) {};
        \draw (1,-1)--++(0,2);
        \node[tensor] (r) at (2.5,0) {};
        \draw (2.5,-1)--++(0,2);
        \node[fill=white] at (1.75,0) {$\dots$};
    \end{tikzpicture} \ ,
\end{equation*}
whenever $x$ is a (nice) linear functional on $SL_q(2)$. As a second step we show how elements of $U_q(sl_2)$ are linear functionals over $SL_q(2)$ and how the resulting MPO-like object is a representation of $U_q(sl_2)$.

The algebra $SL_q(2)$ is actually a Hopf algebra. The coproduct is given by
\begin{align*}
    \Delta(a) = a\otimes a + b\otimes c, \quad \Delta(b) = a \otimes b + b\otimes d, \\
    \Delta(c) = c \otimes a + d \otimes c, \quad \Delta(d) = c\otimes b + d \otimes d. 
\end{align*}
The algebra $U_q(sl_2)$ is generated by four elements $\{K,K^{-1}, E,F\}$ (plus the identity) subject to the relations
\begin{align*}
    KE = q^2 EK, \quad KF = q^{-2} FK, \quad EF - FE = \frac{K-K^{-1}}{q-q^{-1}}, \quad KK^{-1} = K^{-1}K=1.
\end{align*}
Just as $SL_q(2)$, this algebra is also a Hopf algebra, with coproduct
\begin{equation*}
    \begin{gathered}
        \Delta(K) = K\otimes K, \quad \Delta(K^{-1}) = K^{-1} \otimes K^{-1} \\
        \Delta(E) = E \otimes 1 + K \otimes E, \quad \Delta(F) = F \otimes K^{-1} + 1 \otimes F.
    \end{gathered}
\end{equation*}
The two Hopf algebras $U_q(sl_2)$ and $SL_q(2)$ are in fact related to each other: there is a pairing between them that respects the Hopf structure, that is, a bilinear functional $\langle .,.\rangle: U_q(sl_2) \times SL_q(2) \to \mathbb{C}$ such that 
\begin{equation*}
    \langle x, uv \rangle = \sum_{(x)} \langle x_{(1)}, u\rangle \langle x_{(2)} , v\rangle , \quad \text{and} \quad \langle xy, u \rangle = \sum_{(u)}\langle x, u_{(1)} \rangle \langle y, u_{(2)}\rangle ,
\end{equation*}
for all $x,y\in U_q(sl_2)$ and $u,v\in SL_q(2)$, where we used Sweedler's notation for the coproduct writing $\Delta(w) = \sum_{(w)} w_{(1)} \otimes w_{(2)}$. This pairing is given by 
\[
\begin{aligned}
&\langle K,a\rangle = q^{-1}, \quad && \langle K,b\rangle = 0, \quad && \langle K,c\rangle = 0, \quad && \langle K,d\rangle = q, \\
&\langle K^{-1},a\rangle = q, \quad && \langle K^{-1},b\rangle = 0, \quad && \langle K^{-1},c\rangle = 0, \quad && \langle K^{-1},d\rangle = q^{-1}, \\
&\langle E,a\rangle = 0, && \langle E,b\rangle = 0, && \langle E,c\rangle = 1, && \langle E,d\rangle = 0, \\
&\langle F,a\rangle = 0, && \langle F,b\rangle = 1, && \langle F,c\rangle = 0, && \langle F,d\rangle = 0.
\end{aligned}
\]
From these relations the value of  $\langle x, y\rangle $ for general $x\in U_q(sl_2)$ and $y \in SL_q(2)$ can be found by repeatedly using linearity and the compatibility of the pairing with the Hopf structures. 

Using this pairing we can now make sense of the MPO-like object we obtained, even without taking an explicit finite dimensional representation of $SL_q(2)$: writing the MPO tensor as $ \sum_{ij}t_{ij} \otimes \ket{i}\bra{j}$, 
\begin{equation*}
    O_n(x) = \sum_{ij} \langle x, t_{i_1j_1}t_{i_2j_2} \dots t_{i_nj_n} \rangle \cdot \ket{i_1 i_2 \dots i_n} \bra{j_1 j_2 \dots j_n}. 
\end{equation*}
For example, $O_1(x)$ reads as 
\begin{equation*}
    O_1(x) = \langle x, a\rangle \ket{0}\bra{0} +  \langle x, b\rangle \ket{0}\bra{1} + \langle x, c\rangle \ket{1}\bra{0} + \langle x, d\rangle \ket{1}\bra{1}, 
\end{equation*}
or explicitly, $O_1(1) = \id$ and 
\[
O_1(K) =
\begin{pmatrix}
q^{-1} & 0 \\
0 & q
\end{pmatrix}=q^{-Z},
\quad
O_1(K^{-1}) =
\begin{pmatrix}
q & 0 \\
0 & q^{-1}
\end{pmatrix}=q^{Z},
\quad
O_1(E) =
\begin{pmatrix}
0 & 0 \\
1 & 0
\end{pmatrix}=\sigma^+ ,
\quad
O_1(F) =
\begin{pmatrix}
0 & 1 \\
0 & 0
\end{pmatrix}=\sigma^- .
\]
That is, $O_1$ is the defining representation of $SL_q(2)$. It is easy now to see that $O_n(x) = O_1^{\otimes n} \circ \Delta^{n-1}(x)$:
\begin{align*}
    O_n(x) = &\sum_{ij} \langle x, t_{i_1j_1}t_{i_2j_2} \dots t_{i_nj_n} \rangle \cdot \ket{i_1 i_2 \dots i_n} \bra{j_1 j_2 \dots j_n} = \\
     = & \sum_{ij} \sum_{(x)} \langle x_{(1)}, t_{i_1j_1}\rangle \langle x_{(2)}, t_{i_2j_2}\rangle\cdots \langle  x_{(n)},t_{i_nj_n} \rangle \cdot \ket{i_1 i_2 \dots i_n} \bra{j_1 j_2 \dots j_n} = \\ 
     =& O_1(x_{(1)}) \otimes O_1(x_{(2)}) \otimes \dots \otimes O_1(x_{(n)}).
\end{align*}
As $O_1$ is a representation and $\Delta$ is  a homomorphism, $O_n$ is a representation of $U_q(sl_2)$. In particular:

\begin{gather*}
O_n(1) = \id \otimes \id \otimes \dots \otimes \id,\\
O_n(K)= q^{-Z} \otimes q^{-Z} \otimes \dots \otimes q^{-Z},  \\      
O_n(K^{-1})= q^Z \otimes q^Z \otimes \dots \otimes q^Z,  \\     
O_n(E)= \sum_i \big(q^{-Z} \big)^{\otimes i}\otimes \sigma^+\otimes \id^{\otimes n-i-1} , \\     
O_n(F)= \sum_i \id^{\otimes n-i-1} \otimes \sigma^-\otimes \big(q^{Z} \big)^{\otimes i} .     
\end{gather*}

We have thus recovered the well-known result that the XXZ Hamiltonian has $U_q(sl_2)$ symmetries, and these symmetries are given by the (infinite bond dimensional) MPOs
\begin{equation*}
    O_n(x)  =
    \begin{tikzpicture}[xscale=0.7, yscale=0.4, font=\scriptsize]
        \draw[rounded corners] (-1.5,0) rectangle (3,-1.2);
        \node[tensor,label=below:$x$] (l) at (-1,0) {};
        \draw (0,0)--(2.5,0);
        \node[tensor] (l) at (0,0) {};
        \draw (0,-1)--++(0,2);
        \node[tensor] (m) at (1,0) {};
        \draw (1,-1)--++(0,2);
        \node[tensor] (r) at (2.5,0) {};
        \draw (2.5,-1)--++(0,2);
        \node[fill=white] at (1.75,0) {$\dots$};
    \end{tikzpicture} \ ,
\end{equation*}
where $x\in U_q(sl_2)$ and the MPO tensor is given in \cref{eq:mpo_tensor_ansatz}.

\subsection{The transfer matrix of the XXZ model}

We can also obtain the higher transfer matrices of the XXZ model as another set of special solutions of \cref{eq:XXZ_main}. For that, let us first extend \footnote{Let us emphasize that the algebras $U_q(sl_2)$ and $SL_q(2)$ do not have a topology, and thus we are not allowed to take infinite sums in them. This is why adding $K^{1/2}$ is an extension of the algebra.} $U_q(sl_2)$ to include $G = K^{1/2}$. That is, we define the algebra $\overline{U_q(sl_2)}$ as the algebra generated by $E,F,1,K,K^{-1}, G,G^{-1}$, where $E,F,1,K,K^{-1}$ satisfy the usual $U_q(sl_2)$ relations and
$$
  G^2 = K, \quad G E  = q E G, \quad G F = q^{-1} F G,
$$
and $GG^{-1} = G^{-1} G = 1$. 

It is now straightforward to check that the MPO tensors 
\begin{equation*}
    \begin{pmatrix}
      a & b \\ c & d
    \end{pmatrix}
    = 
    \begin{pmatrix}
      \frac{\mu G - \mu^{-1} G^{-1}}{2}  &  \frac{q-q^{-1}}{2}\,E  \\[1ex]
      \frac{q-q^{-1}}{2}\,F &  \frac{\mu G^{-1} - \mu^{-1} G}{2}
    \end{pmatrix}, \qquad
    \begin{pmatrix}
      \tilde{a} & \tilde{b} \\ \tilde{c} & \tilde{d}
    \end{pmatrix}
    = 
    \frac{q-q^{-1}}{2}
    \begin{pmatrix}
    \mu G + \mu^{-1} G^{-1} & 0 \\
    0 & \mu^{-1} G + \mu G^{-1}
    \end{pmatrix},
\end{equation*}
provide a solution to \cref{eq:XXZ_main}. The MPO tensor
\begin{equation*}
    L(\mu) = 
    \begin{pmatrix}
      \frac{\mu \phi(G) - \mu^{-1} \phi(G^{-1})}{2}  &  \frac{q-q^{-1}}{2}\,\phi(E)  \\[1ex]
      \frac{q-q^{-1}}{2}\,\phi(F) &  \frac{\mu \phi(G^{-1}) - \mu^{-1} \phi(G)}{2}
    \end{pmatrix}, 
\end{equation*}
where $\phi$ is a finite dimensional representation of $\overline{U_q(sl_2)}$, defines the higher transfer matrix. For the 2-dimensional representation $\phi$ defined as
\[
\phi(G)=
\begin{pmatrix}
q^{1/2} & 0 \\
0 & q^{-1/2}
\end{pmatrix},
\qquad
\phi(E)=
\begin{pmatrix}
0 & 0 \\
1 & 0
\end{pmatrix},
\qquad
\phi(F)=
\begin{pmatrix}
0 & 1 \\
0 & 0
\end{pmatrix},
\]
then $L(\mu)$ is the Lax operator of the XXZ model,
\[
L(\mu)= \frac{1}{2}
\begin{pmatrix}
\mu q^{1/2}-\mu^{-1}q^{-1/2} & 0 & 0 & 0\\
0 & \mu q^{-1/2}-\mu^{-1}q^{1/2} & q-q^{-1} & 0\\
0 & q-q^{-1} & \mu q^{-1/2}-\mu^{-1}q^{1/2} & 0\\
0 & 0 & 0 & \mu q^{1/2}-\mu^{-1}q^{-1/2}
\end{pmatrix},
\]
and the MPO is the transfer matrix of the XXZ model. 
\section{Application (5). Local operators and telescopic series }

Let us consider a $k$-local operator that adds up to a constant, i.e. $\sum_{i=1}^N O_{[i]}=\lambda \cdot \id$, where $O_{[i]}$ acts on sites $i,\cdots, i+k-1$.

We can use our result by applying the operator to the identity operator. The identity operator is trivially an injective matrix product operator that, when vectorized, it corresponds to a tensor product state. This is an injective MPS with tensor $A^{i,j}=\delta_{i,j}$ i.e., with bond dimension 1.

Then our result yields that there exists a $(k-1)$-body operator $Q$, corresponding to the tensor $B$, which has bond dimension 1, like $A$, such that
$$ O_{[i]} - \lambda/N= \id_{[i]}\otimes Q_{[i+1]} -  Q_{[i]}\otimes \id_{[i+k-1]}  \ , $$
where $Q_{[i]}$ acts on sites $i,\cdots, i+k-2$ and $\id_{[i]}$ is the identity operator acting on site $i$.

\section{Application (6). VUMPS}

Non fine-tuned Hamiltonians do not have exact MPS eigenstates, but their ground states can be faithfully approximated by MPS \cite{PhysRevB.73.094423}. This is the reason behind the success of tensor networks based algorithms. One of the standard algorithms to find those MPS approximations for infinite translational invariant spin chains is VUMPS~\cite{Zauner18}. The goal is to minimize the energy with states belonging to the MPS manifold: $\frac{\langle \psi_A|H|\psi_A\rangle}{\langle \psi_A|\psi_A\rangle}$. The algorithm is formulated in the mixed canonical form of an MPS:
$$|\psi_A\rangle = \sum_{\{\pmb s\}} (\cdots A_L^{s_{n-2}} A_L^{s_{n-1}} A_C^{s_{n}} A_R^{s_{n+1}}A_R^{s_{n+2}}\cdots) | {\pmb s} \rangle,$$
where $A_L$ is the left canonical form of $A$, $A_R$ the right canonical form, and $A_C$, the central canonical form, is defined through
$$ A_L^i \cdot C = C\cdot A_R^i = A_C^i \ , $$
where $C$ is the matrix with the Schmidt coefficients. That is, the transfer matrices satisfy:
\begin{equation*}
\sum_{i,\alpha}(\bar{A}_L^i)_{\alpha,\beta} (A_L^i)_{\alpha,\gamma}=\delta_{\beta,\gamma}
\longleftrightarrow
\begin{tikzpicture}[xscale=0.6,yscale=0.8,baseline=0.3cm]
        \draw[rounded corners=3mm] (1,0) -- (-1,0) -- (-1,1) -- (1,1);
        \draw (0,0) -- (0,1);
        \node[tensor,label=above:$\bar{A}_L$] at (0,1) {};
        \node[tensor,label=below:$A_L$] at (0,0) {};
    \end{tikzpicture} = 
    \begin{tikzpicture}[xscale=0.6,yscale=0.8,baseline=0.3cm]
        \draw[rounded corners=3mm] (0,0) -- (-1,0) -- (-1,1) -- (0,1);
    \end{tikzpicture} \quad \& \quad
    \begin{tikzpicture}[xscale=-0.6,yscale=0.8,baseline=0.3cm]
        \draw[rounded corners=3mm] (1,0) -- (-1,0) -- (-1,1) -- (1,1);
        \draw (0,0) -- (0,1);
        \node[tensor,label=above:$\bar{A}_L$] at (0,1) {};
        \node[tensor,label=below:$A_L$] at (0,0) {};
        \node[tensor, label=right:$C^2$] at (-1,0.5) {};
    \end{tikzpicture} = 
    \begin{tikzpicture}[xscale=-0.6,yscale=0.8,baseline=0.3cm]
        \draw[rounded corners=3mm] (0,0) -- (-1,0) -- (-1,1) -- (0,1);
        \node[tensor, label=right:$C^2$] at (-1,0.5) {};
    \end{tikzpicture} \ . 
\end{equation*}

%The minimum of the energy is found where the gradient is zero:
%$$(\partial_{\bar{A}_C}\langle  \psi_A| )H-\epsilon(A)|\psi_A\rangle=0$$.

Let us now show how such a stationary solution can be found starting from our fundamental relation. In the case that the injective MPS $|\psi_A\rangle$ is an exact solution, we have

 \begin{equation*}
        \begin{tikzpicture}
            \draw (-0.5,0)--(1.5,0);
            \node[tensor,label=below:$A_L$] (l) at (1,0) {};
            \draw (l)--++(0,1);
            \node[tensor,label=below:$A_L$] (r) at (0,0) {};
            \draw (r)--++(0,1);
            \draw[line width=2mm,line cap=round] (0,0.5) -- (1,0.5);
            \node[anchor=south,inner sep=5pt] at (0.5,0.5) {$\tilde{h}$};
        \end{tikzpicture} =
        \begin{tikzpicture}
            \draw (-0.5,0)--(1.5,0);
        \node[tensor,color=red,label=below:$B_L$] (l) at (1,0) {};
            \draw (l)--++(0,0.5);
            \node[tensor,label=below:$A_L$] (r) at (0,0) {};
            \draw (r)--++(0,0.5);
        \end{tikzpicture} - 
        \begin{tikzpicture}
            \draw (-0.5,0)--(1.5,0);
            \node[tensor,label=below:$A_L$] (l) at (1,0) {};
            \draw (l)--++(0,0.5);
        \node[tensor,color=red,label=below:$B_L$] (r) at (0,0) {};
            \draw (r)--++(0,0.5);
        \end{tikzpicture} 
        \ ,\ \tilde{h}=h- \epsilon \ .
    \end{equation*}    

where $\epsilon$ is the energy density.  Note that $B_L$ is only defined up to a mutiple of $A_L$. The same equation holds for $A_R$ by defining  
$$ B_L^i \cdot C = C\cdot B_R^i \ . $$

In general, this equation cannot be satisfied exactly. The central idea is to relax this equation in that we only impose projected versions of it (and those versions can be satisfied exactly):

\begin{equation}
    \begin{tikzpicture}[baseline=4mm]
        \draw (0,0) -- (0,1);
        \draw (1,0) -- (1,1.5);
        \draw[rounded corners =3mm] (0.5,1.5) -- (0.5,1) -- (-0.5,1)--(-0.5,0)--(1.5,0)--(1.5,1.5);
        \node[tensor,label=below:$A_L$] at (0,0) {};
        \node[tensor,label=below:$A_C$] at (1,0) {};
        \node[tensor,label=above:$\bar{A}_L$] at (0,1) {};
        \draw[line width=2mm, line cap=round] (0,0.5)--(1,0.5);        
        \node[anchor=south,inner sep=5pt] at (0.5,0.5) {$\tilde{h}$};
    \end{tikzpicture}
    =
    \begin{tikzpicture}[baseline=4mm]
    \draw (0,0) -- (0,1);
    \draw[rounded corners=3mm] (-0.5,1) --++ (0,-1) --++ (1.25,0) --++ (0,1);
    \node[tensor,color=red,label=below:$B_L$] at (0,0) {};    
    \node[tensor,label=below:$C$] at (0.5,0) {}; 
    \end{tikzpicture}
    -
       \begin{tikzpicture}[baseline=4mm]
        \draw (0,0) -- (0,1);
        \draw (1,0) -- (1,1.5);
        \draw[rounded corners =3mm] (0.5,1.5) -- (0.5,1) -- (-0.5,1)--(-0.5,0)--(1.5,0)--(1.5,1.5);
        \node[tensor,color=red,label=below:$B_L$] at (0,0) {};
        \node[tensor,label=below:$A_C$] at (1,0) {};
        \node[tensor,label=above:$\bar{A}_L$] at (0,1) {};        
    \end{tikzpicture}
    \label{eq1_vumps}
\end{equation}
and 
\begin{equation}
    \begin{tikzpicture}[baseline=4mm]
        \draw (1,0) -- (1,1);
        \draw (0,0) -- (0,1.5);
        \draw[rounded corners =3mm] (0.5,1.5) -- (0.5,1) --(1.5,1)--(1.5,0)--(-0.5,0)--(-0.5,1.5);
        \node[tensor,label=below:$A_C$] at (0,0) {};
        \node[tensor,label=below:$A_R$] at (1,0) {};
        \node[tensor,label=above:$A_R$] at (1,1) {};
        \draw[line width=2mm, line cap=round] (0,0.5)--(1,0.5);       
        \node[anchor=south,inner sep=5pt] at (0.5,0.5) {$\tilde{h}$};
    \end{tikzpicture} 
    =
        \begin{tikzpicture}[baseline=4mm]
        \draw (1,0) -- (1,1);
        \draw (0,0) -- (0,1.5);
        \draw[rounded corners =3mm] (0.5,1.5) -- (0.5,1) --(1.5,1)--(1.5,0)--(-0.5,0)--(-0.5,1.5);
        \node[tensor,label=below:$A_C$] at (0,0) {};
        \node[tensor,color=red,label=below:$B_R$] at (1,0) {};
        \node[tensor,label=above:$A_R$] at (1,1) {};       
    \end{tikzpicture} 
    -
    \begin{tikzpicture}[baseline=4mm]
    \draw (0,0) -- (0,1);
    \draw[rounded corners=3mm] (-0.5,1) --++ (0,-1) --++ (1.25,0) --++ (0,1);
    \node[tensor,color=red,label=below:$B_L$] at (0,0) {};    
    \node[tensor,label=below:$C$] at (0.5,0) {}; 
    \end{tikzpicture} \ .
\label{eq2_vumps}
\end{equation}
VUMPS can now be understood as an iterative procedure for finding a solution to these two equations (\ref{eq1_vumps},\ref{eq2_vumps}). 

The algorithm  proceeds as follows: given the iterates $A_L,A_C$ and $\epsilon$ at iterate $k$, we solve for $B_L$ at iterate $k$  by solving the linear equation (\ref{eq1_vumps}). Same for $B_R$ and equation (\ref{eq2_vumps}). 

The other equations defining the VUMPS algorithm can be written in terms of  the auxiliary $B$ tensors, up to suitable gauge transformation, as follows: 
%equations involve a geometric sum, their solutions satisfy:

\begin{equation*} 
\begin{tikzpicture}[baseline=4mm]
        \draw[rounded corners=3mm] (0.5,0) -- (-0.5,0) -- (-0.5,1) -- (0.5,1);
        \draw (0,0) -- (0,1);
        \node[tensor,label=above:$\bar{A}_L$] at (0,1) {};
        \node[tensor,color=red,label=below:$B_L$] at (0,0) {};
    \end{tikzpicture}
=
\sum_{i<i_C}
\begin{tikzpicture}[baseline=4mm]
 \foreach \x  in {0,1,2,-1,-2}{
         \draw (\x,0) --++ (0,1);
         \node[tensor] at (\x,0) {};
         \node[tensor] at (\x,1) {};
 }
  \node at (-1.5,0.5) {$\cdots$};
        \draw[rounded corners=3mm] (2.4,1)--(-2.5,1) -- (-2.5,0)--(2.4,0);
        \node[anchor=south,inner sep=5pt] at (0.5,0.45) {$\tilde{h}_{i-1,i}$};
        \draw[line width=2mm, line cap=round] (0,0.45)--(1,0.45);        
\end{tikzpicture}
\end{equation*}

and

\begin{equation*}  
\begin{tikzpicture}[baseline=4mm]
        \draw[rounded corners=3mm] (-0.5,0) -- (0.5,0) -- (0.5,1) -- (-0.5,1);
        \draw (0,0) -- (0,1);
        \node[tensor,label=above:$\bar{A}_R$] at (0,1) {};
        \node[tensor,color=red,label=below:$B_R$] at (0,0) {};
    \end{tikzpicture}
=-
\sum_{i>i_C}
\begin{tikzpicture}[baseline=4mm]
 \foreach \x  in {0,1,2,-1,-2}{
         \draw (\x,0) --++ (0,1);
         \node[tensor] at (\x,0) {};
         \node[tensor] at (\x,1) {};
 }
  \node at (1.5,0.5) {$\cdots$};
        \draw[rounded corners=3mm] (-2.4,1)--(2.5,1) -- (2.5,0)--(-2.4,0);
        \node[anchor=south,inner sep=5pt] at (-0.5,0.45) {$\tilde{h}_{i,i+1}$};
        \draw[line width=2mm, line cap=round] (0,0.45)--(-1,0.45);        
\end{tikzpicture}
\end{equation*}

In other words, the terms $B_L$ and $B_R$ encode all the information of the Hamiltonian terms to the left respectively to the right of the sites of interest. This gives a very interesting interpretation to the $B$ terms.

In the next iteration, we fix $B_L,B_R,A_L,A_R$ and solve the eigenvalue problem with unknown eigenvector $A_C$ and eigenvalue $\epsilon$ obtained by summing equations (\ref{eq1_vumps}) and (\ref{eq2_vumps}):  
\begin{equation}
    \begin{tikzpicture}[baseline=4mm]
        \draw (0,0) -- (0,1);
        \draw (1,0) -- (1,1.5);
        \draw[rounded corners =3mm] (0.5,1.5) -- (0.5,1) -- (-0.5,1)--(-0.5,0)--(1.5,0)--(1.5,1.5);
        \node[tensor,label=below:$A_L$] at (0,0) {};
        \node[tensor,label=below:$A_C$] at (1,0) {};
        \node[tensor,label=above:$\bar{A}_L$] at (0,1) {};
        \draw[line width=2mm, line cap=round] (0,0.5)--(1,0.5);    
        \node[anchor=south,inner sep=5pt] at (0.5,0.5) {${h}$};
    \end{tikzpicture} +
    \begin{tikzpicture}[baseline=4mm]
        \draw (1,0) -- (1,1);
        \draw (0,0) -- (0,1.5);
        \draw[rounded corners =3mm] (0.5,1.5) -- (0.5,1) --(1.5,1)--(1.5,0)--(-0.5,0)--(-0.5,1.5);
        \node[tensor,label=below:$A_C$] at (0,0) {};
        \node[tensor,label=below:$A_R$] at (1,0) {};
        \node[tensor,label=above:$A_R$] at (1,1) {};
        \draw[line width=2mm, line cap=round] (0,0.5)--(1,0.5);       
        \node[anchor=south,inner sep=5pt] at (0.5,0.5) {${h}$};
    \end{tikzpicture} 
    -
    \begin{tikzpicture}[baseline=4mm]
        \draw (1,0) -- (1,1);
        \draw (0,0) -- (0,1.5);
        \draw[rounded corners =3mm] (0.5,1.5) -- (0.5,1) --(1.5,1)--(1.5,0)--(-0.5,0)--(-0.5,1.5);
        \node[tensor,label=below:$A_C$] at (0,0) {};
        \node[tensor,color=red,label=below:$B_R$] at (1,0) {};
        \node[tensor,label=above:$A_R$] at (1,1) {};       
    \end{tikzpicture} 
    +
       \begin{tikzpicture}[baseline=4mm]
        \draw (0,0) -- (0,1);
        \draw (1,0) -- (1,1.5);
        \draw[rounded corners =3mm] (0.5,1.5) -- (0.5,1) -- (-0.5,1)--(-0.5,0)--(1.5,0)--(1.5,1.5);
        \node[tensor,color=red,label=below:$B_L$] at (0,0) {};
        \node[tensor,label=below:$A_C$] at (1,0) {};
        \node[tensor,label=above:$\bar{A}_L$] at (0,1) {};        
    \end{tikzpicture}
    =
        2\epsilon \cdot \begin{tikzpicture}[baseline=0mm]
        \draw (0,0) -- (0,0.5);
        \draw[rounded corners =3mm] (-0.5,0.5)--(-0.5,0)--(0.5,0)--(0.5,0.5);
        \node[tensor,label=below:$A_C$] at (0,0) {};
    \end{tikzpicture}
    \label{eq3_VUMPS}
\end{equation}

This fixes the next iterate $A_C$ and $\epsilon$. To update $A_L$ and $A_R$, we substitute $A_C$ with $A^i_L \cdot C=C \cdot A^i_R$ in equations (\ref{eq1_vumps},\ref{eq2_vumps}), do an extra projection with $A_L$ respectively $A_R$, sum both equations, and we get an eigenvalue problem for $C,\epsilon$. Given $A_C$ and $C$, it is now a standard problem to determine the associated $A_L$ and $A_R$. This is one full iteration step. We can now reiterate by solving anew for $B_L,B_R, A_C, C,\epsilon$ until convergence. The conditions under which this iteration converges is explained in the standard VUMPS reference~\cite{Zauner18}.

\section{Application (8). Beyond 1D.}

\subsection{Plaquette operators at the square lattice on injective PEPS}

In this section we generalize our main theorem to the square lattice. We call a PEPS tensor injective if there is an ``inverse'' tensor such that 
\begin{equation*}
    \begin{tikzpicture}[baseline=4mm]
        \draw[canvas is xz plane at y=0] (-0.5,-0.5) grid (0.5,0.5);
        \draw[canvas is xz plane at y=1] (-0.5,-0.5) grid (0.5,0.5);
        \draw (0,0,0) -- (0,1,0);
        \node[tensor] at (0,0,0) {};
        \node[tensor,fill=white] at (0,1,0) {};
    \end{tikzpicture} = 
    \begin{tikzpicture}[baseline=4mm]
        \draw[canvas is xy plane at z=0] (-0.5,0) -- (-0.1,0) -- (-0.1,1) -- (-0.5,1);
        \draw[canvas is xy plane at z=0] (0.5,0) -- (0.1,0) -- (0.1,1) -- (0.5,1);
        \draw[canvas is zy plane at x=0] (-0.5,0) -- (-0.1,0) -- (-0.1,1) -- (-0.5,1);
        \draw[canvas is zy plane at x=0] (0.5,0) -- (0.1,0) -- (0.1,1) -- (0.5,1);
    \end{tikzpicture}    = \id^{\otimes 4} .
\end{equation*}

\begin{theorem}
An injective PEPS $\ket{\psi}$ on the square lattice is an eigenstate of a translationally invariant operator $\mathcal{O}=\sum_i O_i$, i.e., $\mathcal{O} \ket{\psi} = \epsilon \ket{\psi}$, with $O$ being a local $2\times 2$ plaquette term
if and only if there are two tensors $X$ and $Y$, depicted by red
and blue shapes, respectively, such that: 
\begin{equation}    
    \begin{tikzpicture}[scale=0.7]
        \draw[canvas is xz plane at y=0] (-0.5,-0.5) grid (1.5,1.5);
        \node at (0.5,1.1) {$h$};
        \draw[canvas is xz plane at y=0.75, rounded corners, fill=gray] (-0.25,-0.25) rectangle (1.25,1.25);
        \foreach \x in {(0,0,0),(1,0,0),(0,0,1),(1,0,1)}{
            \node[tensor] at \x {};
            \draw \x --++ (0,1.3);        } 
   \end{tikzpicture} = 
    \begin{tikzpicture}[scale=0.7,z={(-0.6cm,-0.5cm)}]
        \draw[canvas is xz plane at y=0] (-0.5,-0.5) grid (1.5,1.5);
        \foreach \x in {(0,0,0),(1,0,0),(0,0,1),(1,0,1)}{
            \node[tensor] at \x {};
            \draw \x --++ (0,1);
        } 
        \draw[line width=2.1mm,line cap=round] (0,0,0) -- (1,0,0);
        \draw[red,line width=1.7mm,line cap=round] (0,0,0) -- (1,0,0);
        \node at (0.6,0.3) {$X$};
        \draw (1,0,1) --++ (0,1);
    \end{tikzpicture} - 
    \begin{tikzpicture}[scale=0.7,z={(-0.6cm,-0.5cm)}]
        \draw[canvas is xz plane at y=0] (-0.5,-0.5) grid (1.5,1.5);
        \foreach \x in {(0,0,0),(1,0,0),(0,0,1),(1,0,1)}{
            \node[tensor] at \x {};
            \draw \x --++ (0,1);
        } 
        \draw[line width=2.1mm,line cap=round] (0,0,1) -- (1,0,1);
        \draw[red,line width=1.7mm,line cap=round] (0,0,1) -- (1,0,1);
        \node at (-0.2,-0.8) {$X$};
    \end{tikzpicture} +    
    \begin{tikzpicture}[scale=0.7,z={(-0.6cm,-0.5cm)}]
        \draw[canvas is xz plane at y=0] (-0.5,-0.5) grid (1.5,1.5);
        \foreach \x in {(0,0,0),(1,0,0),(0,0,1),(1,0,1)}{
            \node[tensor] at \x {};
            \draw \x --++ (0,1);
        } 
        \draw[line width=2.1mm,line cap=round] (0,0,1) -- (0,0,0);
        \draw[blue,line width=1.7mm,line cap=round] (0,0,1) -- (0,0,0);
        \node at (-0.8,-0.1) {$Y$};
    \end{tikzpicture} -    
    \begin{tikzpicture}[scale=0.7,z={(-0.6cm,-0.5cm)}]
        \draw[canvas is xz plane at y=0] (-0.5,-0.5) grid (1.5,1.5);
        \foreach \x in {(0,0,0),(1,0,0),(0,0,1),(1,0,1)}{
            \node[tensor] at \x {};
            \draw \x --++ (0,1);
        } 
        \draw[line width=2.1mm,line cap=round] (1,0,0) -- (1,0,1);
        \draw[blue,line width=1.7mm,line cap=round] (1,0,0) -- (1,0,1);
        \node at (1.3,-0.3) {$Y$};
    \end{tikzpicture} 
 \quad ,
\end{equation}
where $h=O-\epsilon/N$ and $N$ is the number of sites.
\end{theorem}

\begin{proof}
    Just as in the case of MPS, we consider the equation 
    \begin{equation}\label{eq:h_PEPS_0}
        \sum_i h_i \ket{\psi} = 0
    \end{equation}
    On this equation, we apply first the inverse of the PEPS tensor on every lattice site and contract each neighboring pairs of indices in the resulting tensor network (such that the result is a scalar). We obtain that $h$ sandwiched between four PEPS tensors and their inverses is 0:
    \begin{equation}\label{eq:plaquette_h_0}
        \begin{tikzpicture}[scale=0.7]
            \draw[canvas is xz plane at y=0] (-0.7,-0.7) grid (1.7,1.7);
            \draw[canvas is xz plane at y=1.6] (-0.7,-0.7) grid (1.7,1.7);
            \foreach \x in {0,1}{
                \draw[canvas is xy plane at z=\x] (-0.7, 0) rectangle (1.7, 1.6);
            }
            \draw[canvas is zy plane at x=0] (-1.1, 0) rectangle (2.1, 1.6);
            \draw[canvas is xz plane at y=0.75, rounded corners, fill=gray] (-0.25,-0.25) rectangle (1.25,1.25);
            \draw[canvas is zy plane at x==1] (-1.1, 0) rectangle (2.1, 1.6);
            \foreach \x in {(0,0,0),(1,0,0),(0,0,1),(1,0,1)}{
                \node[tensor] at \x {};
                \draw \x --++ (0,1.6);
            } 
            \def\y{1.6}
            \foreach \x in {(0,\y,0),(1,\y,0),(0,\y,1),(1,\y,1)}{
                \node[tensor,fill=white] at \x {};
            }             
        \end{tikzpicture} = 0.
    \end{equation}

    Second, let us consider \cref{eq:h_PEPS_0} again and apply the inverse of the PEPS tensor on every lattice site except one. We contract all neighboring indices of the resulting tensor network. The result is a rank-5 tensor, the same shape as the original PEPS tensor. Because of \cref{eq:plaquette_h_0}, in the resulting equation there are precisely 12 Hamiltonian terms that have non-zero contributions. Below we marked the missing inverse tensor's position by a small red circle, and the 12 relevant plaquettes by small black (and green) circles. 
    \begin{equation*}
        \begin{tikzpicture}[scale=0.5]
            \draw (-0.5,-0.5) grid (4.5,4.5);
            \fill[red] (2,2) circle (3pt);
            \foreach \x in {(1.5,0.5),(2.5,0.5), (1.5,3.5),(2.5,3.5)}{
              \fill \x circle (2pt);
            }
            \foreach \x in {0.5, 1.5, 2.5, 3.5}{
              \fill (\x, 1.5) circle (2pt);
              \fill (\x, 2.5) circle (2pt);
            }
            \fill[green!70!black] (0.5,2.5) circle (2pt);
            \draw[blue,densely dashed,rounded corners=5pt] (1.1,1.1) -- (1.1,0.1) -- (1.9,0.1) -- (1.9,1.9) -- (0.1,1.9) -- (0.1,1.1) -- (1.1,1.1);
            \draw[blue,densely dashed,rounded corners=5pt] (2.9,1.1) -- (2.9,0.1) -- (2.1,0.1) -- (2.1,1.9) -- (3.9,1.9) -- (3.9,1.1) -- (2.9,1.1);
            \draw[blue,densely dashed,rounded corners=5pt] (2.9,2.9) -- (3.9,2.9) -- (3.9,2.1) -- (2.1,2.1) -- (2.1,3.9) -- (2.9,3.9) -- (2.9,2.9);
            \draw[blue,densely dashed,rounded corners=5pt] (1.1,2.9) -- (0.1,2.9) -- (0.1,2.1) -- (1.9,2.1) -- (1.9,3.9) -- (1.1,3.9) -- (1.1,2.9);
        \end{tikzpicture}
    \end{equation*}
    For example, the contribution of the plaquette denoted by the green dot is the following rank-5 tensor:
    \begin{equation*}
        \begin{tikzpicture}
            \draw (-1.5,0) -- (0.5,0);
            \draw (0,0,-0.7) -- (0,0,0.7);
            \draw (0,0) -- (0,0.5);
            \node[tensor] at (0,0) {};
            \node[tensor,label=below:$m$] at (-1,0) {};
        \end{tikzpicture} 
        \quad \text{where} \quad
        m = 
        \begin{tikzpicture}[scale=0.7,baseline=3mm]
            \draw[canvas is xz plane at y=0] (-0.7,-0.7) grid (1.7,1.7);
            \draw[canvas is xz plane at y=1.6] (-0.7,-0.7) grid (1.7,1.7);
            \draw[canvas is xy plane at z=1] (1.7, 0) -- (-0.7, 0) -- (-0.7, 1.6) -- (1.7, 1.6);
            \draw[canvas is xy plane at z=0] (-0.7, 0) rectangle (1.7, 1.6);
            \draw[canvas is zy plane at x=0] (-1.1, 0) rectangle (2.1, 1.6);
            \draw[canvas is xz plane at y=0.75, rounded corners, fill=gray] (-0.25,-0.25) rectangle (1.25,1.25);
            \draw[canvas is zy plane at x==1] (-1.1, 0) rectangle (2.1, 1.6);
            \foreach \x in {(0,0,0),(1,0,0),(0,0,1),(1,0,1)}{
                \node[tensor] at \x {};
                \draw \x --++ (0,1.6);
            } 
            \def\y{1.6}
            \foreach \x in {(0,\y,0),(1,\y,0),(0,\y,1),(1,\y,1)}{
                \node[tensor,fill=white] at \x {};
            }             
        \end{tikzpicture} \ .        
    \end{equation*}
    The 12 terms can be grouped into 4 groups, marked by dashed blue lines. The rank-5 tensors corresponding to the groups on the up-left, up-right, down-right and down-left positions of the red tensor are called $x_{ul},x_{ur},x_{dr},x_{dl}$, respectively.  After this grouping, we obtain the equation 
    \begin{equation*}
        \begin{tikzpicture}[baseline=-1mm]
            \draw[canvas is xz plane at y=0] (-0.5,-0.5) grid (0.5,0.5);
            \draw (0,0,0) -- (0,0.5,0);
            \node[tensor, red, label=below:$\quad x_{ul}$] at (0,0,0) {};
        \end{tikzpicture} + 
        \begin{tikzpicture}[baseline=-1mm]
            \draw[canvas is xz plane at y=0] (-0.5,-0.5) grid (0.5,0.5);
            \draw (0,0,0) -- (0,0.5,0);
            \node[tensor, red, label=below:$\quad x_{ur}$] at (0,0,0) {};
        \end{tikzpicture} + 
        \begin{tikzpicture}[baseline=-1mm]
            \draw[canvas is xz plane at y=0] (-0.5,-0.5) grid (0.5,0.5);
            \draw (0,0,0) -- (0,0.5,0);
            \node[tensor, red, label=below:$\quad x_{dr}$] at (0,0,0) {};
        \end{tikzpicture} + 
        \begin{tikzpicture}[baseline=-1mm]
            \draw[canvas is xz plane at y=0] (-0.5,-0.5) grid (0.5,0.5);
            \draw (0,0,0) -- (0,0.5,0);
            \node[tensor, red, label=below:$\quad x_{dl}$] at (0,0,0) {};
        \end{tikzpicture} = 0.
    \end{equation*}
    Let us consider \cref{eq:h_PEPS_0} again and apply the inverse of the PEPS tensor on every lattice site except for two neighboring sites. Just as above, we contract the neighboring indices of the resulting network until we are left with a rank-8 tensor. Again, because of \cref{eq:plaquette_h_0}, there are only 16 terms in the Hamiltonian that have non-zero contributions. Below we marked the missing inverse tensors' positions by small red circles, and the 16 relevant plaquettes by small black circles. 
    \begin{equation*}
        \begin{tikzpicture}[scale=0.5]
            \draw (-0.5,-0.5) grid (5.5,4.5);
            \fill[red] (2,2) circle (3pt);
            \fill[red] (3,2) circle (3pt);
            \foreach \y in {0.5,3.5}{
                \foreach \x in {1.5,2.5,3.5}{
                  \fill (\x,\y) circle (2pt);
                }
            }
            \foreach \x in {0.5, 1.5, 2.5, 3.5,4.5}{
              \fill (\x, 1.5) circle (2pt);
              \fill (\x, 2.5) circle (2pt);
            }
            \draw[blue,densely dashed,rounded corners=5pt] (1.1,1.1) -- (1.1,0.1) -- (1.9,0.1) -- (1.9,1.9) -- (0.1,1.9) -- (0.1,1.1) -- (1.1,1.1);
            \draw[blue,densely dashed,rounded corners=5pt] (3.9,1.1) -- (3.9,0.1) -- (3.1,0.1) -- (3.1,1.9) -- (4.9,1.9) -- (4.9,1.1) -- (3.9,1.1);
            \draw[blue,densely dashed,rounded corners=5pt] (3.9,2.9) -- (4.9,2.9) -- (4.9,2.1) -- (3.1,2.1) -- (3.1,3.9) -- (3.9,3.9) -- (3.9,2.9);
            \draw[blue,densely dashed,rounded corners=5pt] (1.1,2.9) -- (0.1,2.9) -- (0.1,2.1) -- (1.9,2.1) -- (1.9,3.9) -- (1.1,3.9) -- (1.1,2.9);
            \draw[blue,densely dashed, rounded corners=5pt] (2.1,0.1) rectangle (2.9,1.9);
            \draw[blue,densely dashed, rounded corners=5pt] (2.1,2.1) rectangle (2.9,3.9);
        \end{tikzpicture}
    \end{equation*}
    We group the resulting terms into 6 groups. We obtain
    \begin{equation*}
        \begin{tikzpicture}[baseline=-1mm]
            \draw (-0.5,0) -- (1.5,0);
            \foreach\x in {0,1}{
                \draw (\x,0) -- (\x,0.5);
                \draw (\x,0,-0.7) -- (\x,0,0.7);
            }
            \node[tensor, red, label=below:$\quad x_{ul}$] at (0,0) {};
            \node[tensor] at (1,0) {};
        \end{tikzpicture} + 
        \begin{tikzpicture}[baseline=-1mm]
            \draw (-0.5,0) -- (1.5,0);
            \foreach\x in {0,1}{
                \draw (\x,0) -- (\x,0.5);
                \draw (\x,0,-0.7) -- (\x,0,0.7);
            }
            \node[tensor, red, label=below:$\quad x_{dl}$] at (0,0) {};
            \node[tensor] at (1,0) {};
        \end{tikzpicture} + 
        \begin{tikzpicture}[baseline=-1mm]
            \draw (-0.5,0) -- (1.5,0);
            \foreach\x in {0,1}{
                \draw (\x,0) -- (\x,0.5);
                \draw (\x,0,-0.7) -- (\x,0,0.7);
            }
            \node[tensor, red, label=below:$\quad x_{ur}$] at (1,0) {};
            \node[tensor] at (0,0) {};
        \end{tikzpicture} + 
        \begin{tikzpicture}[baseline=-1mm]
            \draw (-0.5,0) -- (1.5,0);
            \foreach\x in {0,1}{
                \draw (\x,0) -- (\x,0.5);
                \draw (\x,0,-0.7) -- (\x,0,0.7);
            }
            \node[tensor, red, label=below:$\quad x_{dr}$] at (1,0) {};
            \node[tensor] at (0,0) {};
        \end{tikzpicture} + 
        \begin{tikzpicture}[baseline=-1mm]
            \draw (-0.5,0) -- (1.5,0);
            \foreach\x in {0,1}{
                \draw (\x,0) -- (\x,0.5);
                \draw (\x,0,-0.7) -- (\x,0,0.7);
            }
            \draw[fill,rounded corners=0.5mm,red] (-0.1,-0.05) rectangle (1.1, 0.05);
            \node at (0.5,-0.2) {$x_u$};
        \end{tikzpicture}  +
        \begin{tikzpicture}[baseline=-1mm]
            \draw (-0.5,0) -- (1.5,0);
            \foreach\x in {0,1}{
                \draw (\x,0) -- (\x,0.5);
                \draw (\x,0,-0.7) -- (\x,0,0.7);
            }
            \draw[fill,rounded corners=0.5mm,red] (-0.1,-0.05) rectangle (1.1, 0.05);
            \node at (0.5,-0.2) {$x_d$};
        \end{tikzpicture}         
        =0,
    \end{equation*}
    where $x_{ul},x_{ur},x_{dr},x_{dl}$ are as above and $x_u$, $x_d$ are rank-8 tensors arising from the plaquettes above and below between the two marked vertices.

    Rotating the picture above by 90 degrees, we obtain the equation
    \begin{equation*}
        \begin{tikzpicture}[baseline=-3mm]
            \draw (0,0,-0.6) -- (0,0,1.6);
            \foreach\x in {0,1}{
                \draw (0,0,\x) -- (0,0.5,\x);
                \draw (-0.5,0, \x) -- (0.5,0,\x);
            }
            \node[tensor, red, label=below:$\quad x_{ul}$] at (0,0) {};
            \node[tensor] at (0,0,1) {};
        \end{tikzpicture} + 
        \begin{tikzpicture}[baseline=-3mm]
            \draw (0,0,-0.6) -- (0,0,1.6);
            \foreach\x in {0,1}{
                \draw (0,0,\x) -- (0,0.5,\x);
                \draw (-0.5,0, \x) -- (0.5,0,\x);
            }
            \node[tensor, red, label=below:$\quad x_{ur}$] at (0,0) {};
            \node[tensor] at (0,0,1) {};
        \end{tikzpicture} + 
        \begin{tikzpicture}[baseline=-3mm]
            \draw (0,0,-0.6) -- (0,0,1.6);
            \foreach\x in {0,1}{
                \draw (0,0,\x) -- (0,0.5,\x);
                \draw (-0.5,0, \x) -- (0.5,0,\x);
            }
            \node[tensor, red, label=below:$\quad x_{dl}$] at (0,0,1) {};
            \node[tensor] at (0,0) {};
        \end{tikzpicture} + 
        \begin{tikzpicture}[baseline=-3mm]
            \draw (0,0,-0.6) -- (0,0,1.6);
            \foreach\x in {0,1}{
                \draw (0,0,\x) -- (0,0.5,\x);
                \draw (-0.5,0, \x) -- (0.5,0,\x);
            }
            \node[tensor, red, label=below:$\quad x_{dr}$] at (0,0,1) {};
            \node[tensor] at (0,0) {};
        \end{tikzpicture} + 
        \begin{tikzpicture}[baseline=-3mm]
            \draw (0,0,-0.6) -- (0,0,1.6);
            \foreach\x in {0,1}{
                \draw (0,0,\x) -- (0,0.5,\x);
                \draw (-0.5,0,\x) -- (0.5,0,\x);
            }
            \draw[fill,red,rounded corners=0.5mm,canvas is zx plane at y=0] (-0.1,-0.05) rectangle (1.1, 0.05);
            \node at (0.4,-0.25) {$x_l$};
        \end{tikzpicture}  +
        \begin{tikzpicture}[baseline=-3mm]
            \draw (0,0,-0.6) -- (0,0,1.6);
            \foreach\x in {0,1}{
                \draw (0,0,\x) -- (0,0.5,\x);
                \draw (-0.5,0,\x) -- (0.5,0,\x);
            }
            \draw[fill,red,rounded corners=0.5mm,canvas is zx plane at y=0] (-0.1,-0.05) rectangle (1.1, 0.05);
            \node at (0.4,-0.25) {$x_r$};
        \end{tikzpicture}  
        =0,
    \end{equation*}
    where the tensors $x_l$ and $x_r$ are the ones arising from the plaquette terms between the two vertices left out, on the left, and on the right, respectively.

    Finally let us consider \cref{eq:h_PEPS_0} again and apply the inverse of the PEPS tensor on every lattice site except for four sites at the corners of a plaquette. Just as above, we contract the neighboring indices of the resulting network until we are left with a rank-12 tensor. Again, because of \cref{eq:plaquette_h_0}, there are only 21 terms in the Hamiltonian that have non-zero contributions. Below we marked the missing inverse tensors' positions by small red circles, and the 21 relevant plaquettes by small black circles. 
    \begin{equation*}
        \begin{tikzpicture}[scale=0.5]
            \draw (-0.5,-0.5) grid (5.5,5.5);
            \fill[red] (2,2) circle (3pt);
            \fill[red] (3,2) circle (3pt);
            \fill[red] (2,3) circle (3pt);
            \fill[red] (3,3) circle (3pt);
            \foreach \y in {0.5,4.5}{
                \foreach \x in {1.5,2.5,3.5}{
                  \fill (\x,\y) circle (2pt);
                }
            }
            \foreach \x in {0.5, 1.5, 2.5, 3.5,4.5}{
              \fill (\x, 1.5) circle (2pt);
              \fill (\x, 2.5) circle (2pt);
              \fill (\x, 3.5) circle (2pt);
            }
            \draw[blue,densely dashed,rounded corners=5pt] (1.1,1.1) -- (1.1,0.1) -- (1.9,0.1) -- (1.9,1.9) -- (0.1,1.9) -- (0.1,1.1) -- (1.1,1.1);
            \draw[blue,densely dashed,rounded corners=5pt] (3.9,1.1) -- (3.9,0.1) -- (3.1,0.1) -- (3.1,1.9) -- (4.9,1.9) -- (4.9,1.1) -- (3.9,1.1);
            \draw[blue,densely dashed,rounded corners=5pt] (3.9,3.9) -- (4.9,3.9) -- (4.9,3.1) -- (3.1,3.1) -- (3.1,4.9) -- (3.9,4.9) -- (3.9,3.9);
            \draw[blue,densely dashed,rounded corners=5pt] (1.1,3.9) -- (0.1,3.9) -- (0.1,3.1) -- (1.9,3.1) -- (1.9,4.9) -- (1.1,4.9) -- (1.1,3.9);
            \draw[blue,densely dashed, rounded corners=5pt] (2.1,0.1) rectangle (2.9,1.9);
            \draw[blue,densely dashed, rounded corners=5pt] (2.1,3.1) rectangle (2.9,4.9);
            \draw[blue,densely dashed, rounded corners=5pt] (0.1,2.1) rectangle (1.9,2.9);
            \draw[blue,densely dashed, rounded corners=5pt] (3.1,2.1) rectangle (4.9,2.9);
        \end{tikzpicture}
    \end{equation*}
    After blocking the different non-zero terms into 9 different groups as indicated on the figure, we obtain the equation 
    \begin{align*}    
        \begin{tikzpicture}[scale=0.7]
            \draw[canvas is xz plane at y=0] (-0.5,-0.5) grid (1.5,1.5);
            \node at (0.5,1.1) {$h$};
            \draw[canvas is xz plane at y=0.75, rounded corners, fill=gray] (-0.25,-0.25) rectangle (1.25,1.25);
            \foreach \x in {(0,0,0),(1,0,0),(0,0,1),(1,0,1)}{
                \node[tensor] at \x {};
                \draw \x --++ (0,1.3);
            } 
        \end{tikzpicture} +
        \begin{tikzpicture}[scale=0.7]
            \draw[canvas is xz plane at y=0] (-0.5,-0.5) grid (1.5,1.5);
            \foreach \x in {(0,0,0),(1,0,0),(0,0,1),(1,0,1)}{
                \node[tensor] at \x {};
                \draw \x --++ (0,1);
            } 
            \draw[red,line width=1.7mm,line cap=round] (0,0,0) -- (1,0,0);
            \node at (0.3,0.3) {$x_u$};
        \end{tikzpicture} + 
        \begin{tikzpicture}[scale=0.7]
            \draw[canvas is xz plane at y=0] (-0.5,-0.5) grid (1.5,1.5);
            \foreach \x in {(0,0,0),(1,0,0),(0,0,1),(1,0,1)}{
                \node[tensor] at \x {};
                \draw \x --++ (0,1);
            } 
            \draw[red,line width=1.7mm,line cap=round] (0,0,1) -- (1,0,1);
            \node at (0.8,0,1.7) {$x_d$};
        \end{tikzpicture} +    
        \begin{tikzpicture}[scale=0.7]
            \draw[canvas is xz plane at y=0] (-0.5,-0.5) grid (1.5,1.5);
            \foreach \x in {(0,0,0),(1,0,0),(0,0,1),(1,0,1)}{
                \node[tensor] at \x {};
                \draw \x --++ (0,1);
            } 
            \draw[red,line width=1.7mm,line cap=round] (0,0,1) -- (0,0,0);
            \node at (-0.5,0,0.5) {$x_l$};
        \end{tikzpicture} +    
        \begin{tikzpicture}[scale=0.7]
            \draw[canvas is xz plane at y=0] (-0.5,-0.5) grid (1.5,1.5);
            \foreach \x in {(0,0,0),(1,0,0),(0,0,1),(1,0,1)}{
                \node[tensor] at \x {};
                \draw \x --++ (0,1);
            } 
            \draw[red,line width=1.7mm,line cap=round] (1,0,0) -- (1,0,1);
            \node at (1.6, 0, 0.5) {$x_r$};
        \end{tikzpicture} & \\ +
        \begin{tikzpicture}[scale=0.7]
            \draw[canvas is xz plane at y=0] (-0.5,-0.5) grid (1.5,1.5);
            \foreach \x in {(0,0,0),(1,0,0),(0,0,1),(1,0,1)}{
                \node[tensor] at \x {};
                \draw \x --++ (0,0.8);
            } 
            \node[tensor,red] at (0,0) {};
            \node at (0.3,0.3) {$x_{ul}$};
        \end{tikzpicture} +
        \begin{tikzpicture}[scale=0.7]
            \draw[canvas is xz plane at y=0] (-0.5,-0.5) grid (1.5,1.5);
            \foreach \x in {(0,0,0),(1,0,0),(0,0,1),(1,0,1)}{
                \node[tensor] at \x {};
                \draw \x --++ (0,0.8);
            } 
            \node[tensor,red] at (1,0) {};
            \node at (1.35,0.3) {$x_{ur}$};
        \end{tikzpicture} +
        \begin{tikzpicture}[scale=0.7]
            \draw[canvas is xz plane at y=0] (-0.5,-0.5) grid (1.5,1.5);
            \foreach \x in {(0,0,0),(1,0,0),(0,0,1),(1,0,1)}{
                \node[tensor] at \x {};
                \draw \x --++ (0,0.8);
            } 
            \node[tensor,red] at (1,0,1) {};
            \node at (1.5,0,1.7) {$x_{dr}$};
        \end{tikzpicture} +
        \begin{tikzpicture}[scale=0.7]
            \draw[canvas is xz plane at y=0] (-0.5,-0.5) grid (1.5,1.5);
            \foreach \x in {(0,0,0),(1,0,0),(0,0,1),(1,0,1)}{
                \node[tensor] at \x {};
                \draw \x --++ (0,0.8);
            } 
            \node[tensor,red] at (0,0,1) {};
            \node at (0.4,0,1.7) {$x_{dl}$};
        \end{tikzpicture} 
        & = 0.
    \end{align*}
    Defining 
    \begin{align*}
        \begin{tikzpicture}[baseline=-1mm]
            \draw (-0.5,0) -- (1.5,0);
            \foreach\x in {0,1}{
                \draw (\x,0) -- (\x,0.5);
                \draw (\x,0,-0.7) -- (\x,0,0.7);
            }
            \draw[fill,red,rounded corners=0.5mm] (-0.1,-0.05) rectangle (1.1, 0.05);
            \node at (0.5,-0.2) {$X$};
        \end{tikzpicture}     &    =
        \begin{tikzpicture}[baseline=-1mm]
            \draw (-0.5,0) -- (1.5,0);
            \foreach\x in {0,1}{
                \draw (\x,0) -- (\x,0.5);
                \draw (\x,0,-0.7) -- (\x,0,0.7);
            }
            \draw[fill,red,rounded corners=0.5mm] (-0.1,-0.05) rectangle (1.1, 0.05);
            \node at (0.5,-0.2) {$x_d$};
        \end{tikzpicture}   - 
        \begin{tikzpicture}[baseline=-1mm]
            \draw (-0.5,0) -- (1.5,0);
            \foreach\x in {0,1}{
                \draw (\x,0) -- (\x,0.5);
                \draw (\x,0,-0.7) -- (\x,0,0.7);
            }
            \node[tensor, red,label=below:$\quad x_{dr}$] at (0,0) {};
            \node[tensor] at (1,0) {};
        \end{tikzpicture} +
        \begin{tikzpicture}[baseline=-1mm]
            \draw (-0.5,0) -- (1.5,0);
            \foreach\x in {0,1}{
                \draw (\x,0) -- (\x,0.5);
                \draw (\x,0,-0.7) -- (\x,0,0.7);
            }
            \node[tensor, red,label=below:$\quad x_{dr}$] at (1,0) {};
            \node[tensor] at (0,0) {};
        \end{tikzpicture}  
        \\ & = 
        -\begin{tikzpicture}[baseline=-1mm]
            \draw (-0.5,0) -- (1.5,0);
            \foreach\x in {0,1}{
                \draw (\x,0) -- (\x,0.5);
                \draw (\x,0,-0.7) -- (\x,0,0.7);
            }
            \node[tensor, red, label=below:$\quad x_{ul}$] at (0,0) {};
            \node[tensor] at (1,0) {};
        \end{tikzpicture} -
        \begin{tikzpicture}[baseline=-1mm]
            \draw (-0.5,0) -- (1.5,0);
            \foreach\x in {0,1}{
                \draw (\x,0) -- (\x,0.5);
                \draw (\x,0,-0.7) -- (\x,0,0.7);
            }
            \node[tensor, red, label=below:$\quad x_{ur}$] at (1,0) {};
            \node[tensor] at (0,0) {};
        \end{tikzpicture} - 
        \begin{tikzpicture}[baseline=-1mm]
            \draw (-0.5,0) -- (1.5,0);
            \foreach\x in {0,1}{
                \draw (\x,0) -- (\x,0.5);
                \draw (\x,0,-0.7) -- (\x,0,0.7);
            }
            \node[tensor, red, label=below:$\quad x_{dl}$] at (0,0) {};
            \node[tensor] at (1,0) {};
        \end{tikzpicture} -
        \begin{tikzpicture}[baseline=-1mm]
            \draw (-0.5,0) -- (1.5,0);
            \foreach\x in {0,1}{
                \draw (\x,0) -- (\x,0.5);
                \draw (\x,0,-0.7) -- (\x,0,0.7);
            }
            \draw[fill,red, rounded corners=0.5mm] (-0.1,-0.05) rectangle (1.1, 0.05);
            \node at (0.5,-0.2) {$x_u$};
        \end{tikzpicture} - 
        \begin{tikzpicture}[baseline=-1mm]
            \draw (-0.5,0) -- (1.5,0);
            \foreach\x in {0,1}{
                \draw (\x,0) -- (\x,0.5);
                \draw (\x,0,-0.7) -- (\x,0,0.7);
            }
            \node[tensor, red, label=below:$\quad x_{dr}$] at (0,0) {};
            \node[tensor] at (1,0) {};
        \end{tikzpicture} \\ &= 
        \begin{tikzpicture}[baseline=-1mm]
            \draw (-0.5,0) -- (1.5,0);
            \foreach\x in {0,1}{
                \draw (\x,0) -- (\x,0.5);
                \draw (\x,0,-0.7) -- (\x,0,0.7);
            }
            \node[tensor, red, label=below:$\quad x_{ur}$] at (0,0) {};
            \node[tensor] at (1,0) {};
        \end{tikzpicture} -
        \begin{tikzpicture}[baseline=-1mm]
            \draw (-0.5,0) -- (1.5,0);
            \foreach\x in {0,1}{
                \draw (\x,0) -- (\x,0.5);
                \draw (\x,0,-0.7) -- (\x,0,0.7);
            }
            \node[tensor, red, label=below:$\quad x_{ur}$] at (1,0) {};
            \node[tensor] at (0,0) {};
        \end{tikzpicture} -
        \begin{tikzpicture}[baseline=-1mm]
            \draw (-0.5,0) -- (1.5,0);
            \foreach\x in {0,1}{
                \draw (\x,0) -- (\x,0.5);
                \draw (\x,0,-0.7) -- (\x,0,0.7);
            }
            \draw[fill,red, rounded corners=0.5mm] (-0.1,-0.05) rectangle (1.1, 0.05);
            \node at (0.5,-0.2) {$x_u$};
        \end{tikzpicture}         
    \end{align*}
    and 
    \begin{align*}
        \begin{tikzpicture}[baseline=-3mm]
            \draw (0,0,-0.6) -- (0,0,1.6);
            \foreach\x in {0,1}{
                \draw (0,0,\x) -- (0,0.5,\x);
                \draw (-0.5,0,\x) -- (0.5,0,\x);
            }
            \draw[fill,red,rounded corners=0.5mm,canvas is zx plane at y=0] (-0.1,-0.05) rectangle (1.1, 0.05);
            \node at (0.4,-0.25) {$Y$};
        \end{tikzpicture}  & = 
        \begin{tikzpicture}[baseline=-3mm]
            \draw (0,0,-0.6) -- (0,0,1.6);
            \foreach\x in {0,1}{
                \draw (0,0,\x) -- (0,0.5,\x);
                \draw (-0.5,0,\x) -- (0.5,0,\x);
            }
            \draw[fill,red,rounded corners=0.5mm,canvas is zx plane at y=0] (-0.1,-0.05) rectangle (1.1, 0.05);
            \node at (0.4,-0.25) {$x_r$};
        \end{tikzpicture}  \\ & = -
        \begin{tikzpicture}[baseline=-3mm]
            \draw (0,0,-0.6) -- (0,0,1.6);
            \foreach\x in {0,1}{
                \draw (0,0,\x) -- (0,0.5,\x);
                \draw (-0.5,0, \x) -- (0.5,0,\x);
            }
            \node[tensor, red,label=below:$\quad x_{ul}$] at (0,0) {};
            \node[tensor] at (0,0,1) {};
        \end{tikzpicture} - 
        \begin{tikzpicture}[baseline=-3mm]
            \draw (0,0,-0.6) -- (0,0,1.6);
            \foreach\x in {0,1}{
                \draw (0,0,\x) -- (0,0.5,\x);
                \draw (-0.5,0, \x) -- (0.5,0,\x);
            }
            \node[tensor, red,label=below:$\quad x_{dl}$] at (0,0,1) {};
            \node[tensor] at (0,0) {};
        \end{tikzpicture} - 
        \begin{tikzpicture}[baseline=-3mm]
            \draw (0,0,-0.6) -- (0,0,1.6);
            \foreach\x in {0,1}{
                \draw (0,0,\x) -- (0,0.5,\x);
                \draw (-0.5,0,\x) -- (0.5,0,\x);
            }
            \draw[fill,red, rounded corners=0.5mm,canvas is zx plane at y=0] (-0.1,-0.05) rectangle (1.1, 0.05);
            \node at (0.4,-0.25) {$x_l$};
        \end{tikzpicture}  -
        \begin{tikzpicture}[baseline=-3mm]
            \draw (0,0,-0.6) -- (0,0,1.6);
            \foreach\x in {0,1}{
                \draw (0,0,\x) -- (0,0.5,\x);
                \draw (-0.5,0, \x) -- (0.5,0,\x);
            }
            \node[tensor, red,label=below:$\quad x_{ur}$] at (0,0) {};
            \node[tensor] at (0,0,1) {};
        \end{tikzpicture} -
        \begin{tikzpicture}[baseline=-3mm]
            \draw (0,0,-0.6) -- (0,0,1.6);
            \foreach\x in {0,1}{
                \draw (0,0,\x) -- (0,0.5,\x);
                \draw (-0.5,0, \x) -- (0.5,0,\x);
            }
            \node[tensor, red,label=below:$\quad x_{dr}$] at (0,0,1) {};
            \node[tensor] at (0,0) {};
        \end{tikzpicture} ,
    \end{align*}
    we obtain 
    \begin{align*}    
        \begin{tikzpicture}[scale=0.7]
            \draw[canvas is xz plane at y=0] (-0.5,-0.5) grid (1.5,1.5);
            \node at (0.5,1.1) {$h$};
            \draw[canvas is xz plane at y=0.75, rounded corners, fill=gray] (-0.25,-0.25) rectangle (1.25,1.25);
            \foreach \x in {(0,0,0),(1,0,0),(0,0,1),(1,0,1)}{
                \node[tensor] at \x {};
                \draw \x --++ (0,1.3);
            } 
        \end{tikzpicture} -
        \begin{tikzpicture}[scale=0.7]
            \draw[canvas is xz plane at y=0] (-0.5,-0.5) grid (1.5,1.5);
            \foreach \x in {(0,0,0),(1,0,0),(0,0,1),(1,0,1)}{
                \node[tensor] at \x {};
                \draw \x --++ (0,1);
            } 
            \draw[red,line width=1.7mm,line cap=round] (0,0,0) -- (1,0,0);
            \node at (0.3,0.3) {$X$};
        \end{tikzpicture} +
        \begin{tikzpicture}[scale=0.7]
            \draw[canvas is xz plane at y=0] (-0.5,-0.5) grid (1.5,1.5);
            \foreach \x in {(0,0,0),(1,0,0),(0,0,1),(1,0,1)}{
                \node[tensor] at \x {};
                \draw \x --++ (0,1);
            } 
            \draw[red,line width=1.7mm,line cap=round] (0,0,1) -- (1,0,1);
            \node at (0.8,0,1.7) {$X$};
        \end{tikzpicture} -    
        \begin{tikzpicture}[scale=0.7]
            \draw[canvas is xz plane at y=0] (-0.5,-0.5) grid (1.5,1.5);
            \foreach \x in {(0,0,0),(1,0,0),(0,0,1),(1,0,1)}{
                \node[tensor] at \x {};
                \draw \x --++ (0,1);
            } 
            \draw[red,line width=1.7mm,line cap=round] (0,0,1) -- (0,0,0);
            \node at (-0.5,0,0.5) {$Y$};
        \end{tikzpicture} +
        \begin{tikzpicture}[scale=0.7]
            \draw[canvas is xz plane at y=0] (-0.5,-0.5) grid (1.5,1.5);
            \foreach \x in {(0,0,0),(1,0,0),(0,0,1),(1,0,1)}{
                \node[tensor] at \x {};
                \draw \x --++ (0,1);
            } 
            \draw[red,line width=1.7mm,line cap=round] (1,0,0) -- (1,0,1);
            \node at (1.6, 0, 0.5) {$Y$};
        \end{tikzpicture} = 0,
    \end{align*}
    which is the desired equation.
\end{proof}

\subsection{Nearest neighbor operators at the square lattice on injective PEPS}

In this section we prove
\begin{theorem}
    Let $O$ be a nearest neighbor operator on an $n \times m$ square lattice: $O = \sum_i h_i + \sum_j v_j$, where $h_i$ are the horizontal terms and $v_j$ are the vertical ones. Let $\ket{\psi}$ be a PEPS generated by an injective tensor $A$ with bond dimension $D_h$ in the horizontal direction and $D_v$ in the vertical direction. Then $O\ket{\psi} = 0$ if and only if there are PEPS tensors $L,R,U,D$ of the same shape as $A$, a matrix $X$ of size $D_v^2 \times D_v^2$  and a matrix $Y$ of size $D_h^2 \times D_h^2$  such that $L+U+R+D=0$ and  
    \begin{align*}
        \begin{tikzpicture}[z={(-4.85mm,-4.85mm)}]
            \draw (-0.5,0) -- (1.5,0);
            \foreach \x in {0,1}{
              \draw (\x,0,-0.5) -- (\x,0,0.5);
              \draw (\x,0) --++ (0,1);
              \node[tensor] at (\x,0) {};
            }
            \draw[fill=gray, rounded corners=1mm] (-0.2,0.4) rectangle (1.2,0.6);  
            \node at (0.5,0.8) {$h$};
        \end{tikzpicture} = 
        \begin{tikzpicture}[z={(-4.85mm,-4.85mm)}]
            \draw (-0.5,0) -- (1.5,0);
            \foreach \x in {0,1}{
              \draw (\x,0,-0.5) -- (\x,0,0.5);
              \draw (\x,0) --++ (0,0.6);
            }
          \node[tensor,fill=red, label = below right:$L$] at (0,0) {};
          \node[tensor] at (1,0) {};
        \end{tikzpicture} +
        \begin{tikzpicture}[z={(-4.85mm,-4.85mm)}]
            \draw (-0.5,0) -- (1.5,0);
            \foreach \x in {0,1}{
              \draw (\x,0,-0.5) -- (\x,0,0.5);
              \draw (\x,0) --++ (0,0.6);
            }
          \node[tensor] at (0,0) {};
          \node[tensor,fill=red,label = below right:$R$] at (1,0) {};
        \end{tikzpicture} +
        \begin{tikzpicture}[z={(-4.85mm,-4.85mm)}]
            \draw (-0.5,0) -- (1.5,0);
            \foreach \x in {0,1}{
              \draw (\x,0,-1) -- (\x,0,1);
              \draw (\x,0) --++ (0,0.6);
              \node[tensor] at (\x,0) {};
            }
            \filldraw[black,canvas is xz plane at y=0, rounded corners=1mm]  (-0.25,0.35) rectangle (1.25,0.65);
            \filldraw[red,canvas is xz plane at y=0, rounded corners=1mm]  (-0.2,0.4) rectangle (1.2,0.6);
            \node at (0.1,-0.5) {$X$};
        \end{tikzpicture}  -
        \begin{tikzpicture}[z={(-4.85mm,-4.85mm)}]
            \draw (-0.5,0) -- (1.5,0);
            \foreach \x in {0,1}{
              \draw (\x,0,-1) -- (\x,0,1);
              \draw (\x,0) --++ (0,0.6);
              \node[tensor] at (\x,0) {};
            }
            \filldraw[black,canvas is xz plane at y=0, rounded corners=1mm]  (-0.25,-0.35) rectangle (1.25,-0.65);
            \filldraw[red,canvas is xz plane at y=0, rounded corners=1mm]  (-0.2,-0.4) rectangle (1.2,-0.6);
            \node at (0.8,0.5) {$X$};
        \end{tikzpicture}, \\ 
        \begin{tikzpicture}[z={(-4.85mm,-4.85mm)}, baseline = -3.5mm]
            \draw (0,0,-0.5) -- (0,0,1.5);
            \foreach \z in {0,1}{
              \draw (-0.5,0,\z) -- (0.5,0,\z);
              \draw (0,0,\z) --++ (0,1);
              \node[tensor] at (0,0,\z) {};
            }
            \draw[fill=gray, rounded corners=1mm,canvas is yz plane at x=0] (0.4,-0.2) rectangle (0.6,1.2);
            \node at (-0.3,0.5) {$v$};
        \end{tikzpicture} = 
        \begin{tikzpicture}[z={(-4.85mm,-4.85mm)}, baseline = -3.5mm]
            \draw (0,0,-0.5) -- (0,0,1.5);
            \foreach \z in {0,1}{
              \draw (-0.45,0,\z) -- (0.5,0,\z);
              \draw (0,0,\z) --++ (0,0.6);
            }
          \node[tensor,fill=red,label=below right:$U$] at (0,0,0) {};
          \node[tensor] at (0,0,1) {};
        \end{tikzpicture} +
        \begin{tikzpicture}[z={(-4.85mm,-4.85mm)}, baseline = -3.5mm]
            \draw (0,0,-0.5) -- (0,0,1.5);
            \foreach \z in {0,1}{
              \draw (-0.45,0,\z) -- (0.5,0,\z);
              \draw (0,0,\z) --++ (0,0.6);
            }
          \node[tensor] at (0,0,0) {};
          \node[tensor,fill=red,label=below right:$D$] at (0,0,1) {};
        \end{tikzpicture} +
        \begin{tikzpicture}[z={(-4.85mm,-4.85mm)}, baseline = -3.5mm]
            \draw (0,0,-0.5) -- (0,0,1.5);
            \foreach \z in {0,1}{
              \draw (-0.6,0,\z) -- (0.8,0,\z);
              \draw (0,0,\z) --++ (0,0.6);
              \node[tensor] at (0,0,\z) {};
            }
            \filldraw[black,canvas is xz plane at y=0, rounded corners=1mm]
              (0.35,-0.25) rectangle (0.65,1.25);
            \filldraw[red,canvas is xz plane at y=0, rounded corners=1mm]
              (0.4,-0.2) rectangle (0.6,1.2);
            \node at (0.7,-0.25) {$Y$};
        \end{tikzpicture}
        -
        \begin{tikzpicture}[z={(-4.85mm,-4.85mm)}, baseline = -3.5mm]
            \draw (0,0,-0.5) -- (0,0,1.5);
            \foreach \z in {0,1}{
              \draw (-0.8,0,\z) -- (0.6,0,\z);
              \draw (0,0,\z) --++ (0,0.6);
              \node[tensor] at (0,0,\z) {};
            }
            \filldraw[black,canvas is xz plane at y=0, rounded corners=1mm]
              (-0.35,-0.25) rectangle (-0.65,1.25);
            \filldraw[red,canvas is xz plane at y=0, rounded corners=1mm]
              (-0.4,-0.2) rectangle (-0.6,1.2);
            \node at (-1.2,-0.25) {$Y$};
        \end{tikzpicture}.
    \end{align*}
\end{theorem}

\begin{proof}
  Let us consider the equation $O\ket{\psi} = 0$ and apply the inverse of $A$ on every lattice site, and then contract all neighboring indices of the resulting tensor network. We obtain that 
  \begin{equation*}
    \begin{tikzpicture}[z={(-4.85mm,-4.85mm)}, baseline = 4mm]
        \draw (-0.5,0) rectangle (1.5,1);
        \foreach \x in {0,1}{
          \draw (\x,0,-0.5) -- (\x,0,0.5);
          \draw (\x,1,-0.5) -- (\x,1,0.5);
          \draw (\x,0,-0.5) -- (\x,1,-0.5);
          \draw (\x,0) --++ (0,1);
          \node[tensor] at (\x,0) {};
          \node[tensor] at (\x,1) {};
        }
        \draw[fill=gray, rounded corners=1mm] (-0.2,0.4) rectangle (1.2,0.6);  
        \draw (0,0,0.5) -- (0,1,0.5);
        \draw (1,0,0.5) -- (1,1,0.5);
    \end{tikzpicture} + 
    \begin{tikzpicture}[z={(-4.3mm,-4.85mm)}, baseline = 3mm]
        \draw[canvas is yz plane at x=0] (0,-0.5) rectangle (1,1.5);
        \foreach \z in {0,1}{
          \draw (-0.5,0,\z) -- (0.5,0,\z);
          \draw (-0.5,0,\z) -- (-0.5,1,\z);
          \draw (0,0,\z) --++ (0,1);
          \node[tensor] at (0,0,\z) {};
          \node[tensor] at (0,1,\z) {};
        }
        \draw[fill=gray, rounded corners=1mm,canvas is yz plane at x=0] (0.4,-0.2) rectangle (0.6,1.2);
        \draw (-0.5,1,1) -- (0.5,1,1);
        \draw (-0.5,1,0) -- (0.5,1,0);
        \draw (0.5,0,0) -- (0.5,1,0);
        \draw (0.5,0,1) -- (0.5,1,1);
    \end{tikzpicture} = 0.
  \end{equation*}
  Let 
  \begin{equation*}
    \lambda = 
    \begin{tikzpicture}[z={(-4.85mm,-4.85mm)}, baseline = 4mm]
        \draw (-0.5,0) rectangle (1.5,1);
        \foreach \x in {0,1}{
          \draw (\x,0,-0.5) -- (\x,0,0.5);
          \draw (\x,1,-0.5) -- (\x,1,0.5);
          \draw (\x,0,-0.5) -- (\x,1,-0.5);
          \draw (\x,0) --++ (0,1);
          \node[tensor] at (\x,0) {};
          \node[tensor] at (\x,1) {};
        }
        \draw[fill=gray, rounded corners=1mm] (-0.2,0.4) rectangle (1.2,0.6);  
        \draw (0,0,0.5) -- (0,1,0.5);
        \draw (1,0,0.5) -- (1,1,0.5);
    \end{tikzpicture} = -  \ 
    \begin{tikzpicture}[z={(-4.3mm,-4.85mm)}, baseline = 3mm]
        \draw[canvas is yz plane at x=0] (0,-0.5) rectangle (1,1.5);
        \foreach \z in {0,1}{
          \draw (-0.5,0,\z) -- (0.5,0,\z);
          \draw (-0.5,0,\z) -- (-0.5,1,\z);
          \draw (0,0,\z) --++ (0,1);
          \node[tensor] at (0,0,\z) {};
          \node[tensor] at (0,1,\z) {};
        }
        \draw[fill=gray, rounded corners=1mm,canvas is yz plane at x=0] (0.4,-0.2) rectangle (0.6,1.2);
        \draw (-0.5,1,1) -- (0.5,1,1);
        \draw (-0.5,1,0) -- (0.5,1,0);
        \draw (0.5,0,0) -- (0.5,1,0);
        \draw (0.5,0,1) -- (0.5,1,1);
    \end{tikzpicture} \ .
  \end{equation*}
  Next we consider again the equation $O\ket{\psi} = 0$ and apply the inverse of $A$ on every lattice site, and then contract all neighboring indices except one pair over a vertical edge. There are 7 terms of the local operator $O$ that don't get completely contracted. The location of them w.r.t.\ the special edge is depicted below:
  \begin{equation*}
      \begin{tikzpicture}
          \draw (-0.5,-0.5) grid (4.5,3.5);
          \filldraw[red] (2,1.5) circle (3pt);
          \foreach \x in {(2,1.5), (2,2.5),(2,0.5), (2.5,1),  (2.5,2),  (1.5,1), (1.5,2)} {
            \filldraw \x circle (2pt);
          }
          \node[anchor=west,inner sep=5pt] at (2,1.5) {$m$}; 
          \node[anchor=west,inner sep=4pt] at (2,2.5) {$y_u$};           
          \node[anchor=west,inner sep=4pt] at (2,0.5) {$y_d$};
          \node[anchor=south,inner sep=4pt] at (2.5,1) {$x_{dr}$};
          \node[anchor=south,inner sep=4pt] at (2.5,2) {$x_{ur}$};
          \node[anchor=south,inner sep=4pt] at (1.5,1) {$x_{dl}$};
          \node[anchor=south,inner sep=4pt] at (1.5,2) {$x_{ul}$};
      \end{tikzpicture}
  \end{equation*}
  We labeled the links with the terms each operator term contributes with. On an $N \times M$ lattice we arrive at the equation
  \begin{equation*}
      x_{ul} + x_{dl} + x_{dr} + x_{ur} + y_u + y_d + m + (NM - 4)\cdot \lambda \cdot \id  + (NM-3) \cdot (-\lambda) \cdot \id= 0, 
  \end{equation*}
  as there are $NM - 4$ fully contracted horizontal terms and $NM - 3$ fully contracted vertical terms. That is, 
  \begin{equation}\label{eq:PEPS_NN_V}
      x_{ul} + x_{dl} + x_{dr} + x_{ur} + y_u + y_d + m - \lambda \cdot \id = 0. 
  \end{equation}

  Let us consider again the equation $O\ket{\psi} = 0$ and apply the inverse of $A$ on every lattice site, and then contract all neighboring indices except one pair over a horizontal edge. There are 7 terms of the local operator $O$ that don't get completely contracted. The location of them w.r.t.\ the special edge is depicted below:
  \begin{equation*}
      \begin{tikzpicture}
          \draw (-0.5,-0.5) grid (3.5,2.5);
          \filldraw[red] (1.5,1) circle (3pt);
          \foreach \x in {(1.5,1), (2.5,1),(0.5,1), (1,1.5),  (2,1.5),  (1,0.5), (2,0.5)} {
            \filldraw \x circle (2pt);
          }
          \node[anchor=south,inner sep=5pt] at (1.5,1) {$n$}; 
          \node[anchor=south,inner sep=4pt] at (2.5,1) {$x_l$};           
          \node[anchor=south,inner sep=4pt] at (0.5,1) {$x_r$};
          \node[anchor=west,inner sep=4pt] at (1,1.5) {$y_{ul}$};
          \node[anchor=west,inner sep=4pt] at (2,1.5) {$y_{ur}$};
          \node[anchor=west,inner sep=4pt] at (1,0.5) {$y_{dl}$};
          \node[anchor=west,inner sep=4pt] at (2,0.5) {$y_{dr}$};
      \end{tikzpicture}
  \end{equation*}
  We labeled the links with the terms each operator term contributes with. Similar to the above, we obtain the equation
  \begin{equation}\label{eq:PEPS_NN_H}
      y_{ul} + y_{dl} + y_{dr} + y_{ur} + x_l + x_r + n + \lambda \cdot \id = 0. 
  \end{equation}

  Let us now consider again the equation $O\ket{\psi} = 0$ and apply the inverse of $A$ on every lattice site, and then contract all neighboring indices except two pairs over two neighboring vertical edges. There are 12 terms of the local operator $O$ that don't get completely contracted. The location of them w.r.t.\ the two special edges is depicted below:
  \begin{equation*}
      \begin{tikzpicture}
          \draw (-0.5,-0.5) grid (5.5,3.5);
          \filldraw[red] (2,1.5) circle (3pt);
          \filldraw[red] (3,1.5) circle (3pt);
          \foreach \x in {      (1.5,1), (1.5,2), 
                           (2,1.5), (2,2.5),(2,0.5), 
                              (2.5,1),  (2.5,2),
                           (3,1.5), (3,2.5),(3,0.5), 
                              (3.5,1),  (3.5,2)} {
            \filldraw \x circle (2pt);
          }
          \node[anchor=west,inner sep=5pt] at (2,1.5) {$m$}; 
          \node[anchor=west,inner sep=4pt] at (2,2.5) {$y_u$};           
          \node[anchor=west,inner sep=4pt] at (2,0.5) {$y_d$};
          \node[anchor=south,inner sep=4pt] at (2.5,1) {$x_d$};
          \node[anchor=south,inner sep=4pt] at (2.5,2) {$x_u$};
          \node[anchor=south,inner sep=4pt] at (1.5,1) {$x_{dl}$};
          \node[anchor=south,inner sep=4pt] at (1.5,2) {$x_{ul}$};
          \node[anchor=west,inner sep=5pt] at (3,1.5) {$m$}; 
          \node[anchor=west,inner sep=4pt] at (3,2.5) {$y_u$};           
          \node[anchor=west,inner sep=4pt] at (3,0.5) {$y_d$};
          \node[anchor=south,inner sep=4pt] at (3.5,1) {$x_{dr}$};
          \node[anchor=south,inner sep=4pt] at (3.5,2) {$x_{ur}$};
      \end{tikzpicture}
  \end{equation*}
  We labeled the links with the terms each operator term contributes with. On an $N \times M$ lattice we arrive at the equation  
  \begin{equation*}
    x_u + x_d + (x_{ul} + x_{dl} + y_u + m + y_d) \otimes 1 + 1 \otimes (y_u + m + y_d + x_{ur} + x_{dr}) + (NM-6)\lambda \cdot 1\otimes 1 - (NM-6)\lambda\cdot 1\otimes 1 = 0
  \end{equation*}
  Using \cref{eq:PEPS_NN_V}, we obtain 
  \begin{equation}\label{eq:PEPS_NN_xu_xd}
    x_u + x_d - (x_{ur} + x_{dr}) \otimes 1 - 1 \otimes (x_{ul} + x_{dl}) + 2\lambda \ 1\otimes 1 = 0.
  \end{equation}

  Let us now consider again the equation $O\ket{\psi} = 0$ and apply the inverse of $A$ on every lattice site, and then contract all neighboring indices except two pairs over two neighboring horizontal edges. There are 12 terms of the local operator $O$ that don't get completely contracted. The location of them w.r.t.\ the two special edges is depicted below:
  \begin{equation*}
      \begin{tikzpicture}[rotate=90]
          \draw (0.5,-0.5) grid (4.5,3.5);
          \filldraw[red] (2,1.5) circle (3pt);
          \filldraw[red] (3,1.5) circle (3pt);
          \foreach \x in {      (1.5,1), (1.5,2), 
                           (2,1.5), (2,2.5),(2,0.5), 
                              (2.5,1),  (2.5,2),
                           (3,1.5), (3,2.5),(3,0.5), 
                              (3.5,1),  (3.5,2)} {
            \filldraw \x circle (2pt);
          }
          \node[anchor=south,inner sep=5pt] at (2,1.5) {$n$}; 
          \node[anchor=south,inner sep=4pt] at (2,2.5) {$x_l$};           
          \node[anchor=south,inner sep=4pt] at (2,0.5) {$x_r$};
          \node[anchor=west,inner sep=4pt] at (2.5,1) {$y_l$};
          \node[anchor=west,inner sep=4pt] at (2.5,2) {$y_r$};
          \node[anchor=west,inner sep=4pt] at (1.5,1) {$y_{dr}$};
          \node[anchor=west,inner sep=4pt] at (1.5,2) {$y_{dl}$};
          \node[anchor=south,inner sep=5pt] at (3,1.5) {$m$}; 
          \node[anchor=south,inner sep=4pt] at (3,2.5) {$x_l$};           
          \node[anchor=south,inner sep=4pt] at (3,0.5) {$x_r$};
          \node[anchor=west,inner sep=4pt] at (3.5,1) {$y_{ur}$};
          \node[anchor=west,inner sep=4pt] at (3.5,2) {$y_{ul}$};
      \end{tikzpicture}
  \end{equation*}
  We labeled the links with the terms each operator term contributes with. Just as above, we arrive at  
  \begin{equation*}
    y_l + y_r + (y_{ul} + y_{ur} + x_l + m + x_r) \otimes 1 + 1 \otimes (x_l + m + x_r + y_{dl} + y_{dr}) = 0
  \end{equation*}
  Using \cref{eq:PEPS_NN_H}, we obtain 
  \begin{equation}\label{eq:PEPS_NN_xu_xd}
    y_l + y_l - (y_{dl} + y_{dr}) \otimes 1 - 1 \otimes (y_{ul} + y_{ur})  = 0.
  \end{equation}

  Let us now consider again the equation $O\ket{\psi} = 0$ and apply the inverse of $A$ on every lattice site except one, and then contract all neighboring indices. The result is a PEPS tensor of the same shape as $A$. There are 16 terms of the local operator $O$ that don't get completely contracted. The location of them w.r.t.\ the special vertex is depicted below:
  \begin{equation*}
      \begin{tikzpicture}
          \draw (-0.5,-0.5) grid (4.5,4.5);
          \filldraw[red] (2,2) circle (3pt);
          \foreach \x in {(0.5,2),
                          (1,1.5), (1,2.5),
                          (1.5,1), (1.5,2), (1.5,3),
                          (2,0.5), (2,1.5), (2,2.5), (2,3.5), 
                          (2.5,1),  (2.5,2), (2.5,3),
                          (3,1.5), (3,2.5),
                          (3.5,2)} {
            \filldraw \x circle (2pt);
          }
          \node[anchor=south,inner sep=4pt] at (0.5,2) {$x_l$};
          \node[anchor=west,inner sep=5pt] at (1,1.5) {$y_{dl}$}; 
          \node[anchor=west,inner sep=5pt] at (1,2.5) {$y_{ul}$}; 
          \node[anchor=south,inner sep=4pt] at (1.5,1) {$x_{dl}$};
          \node[anchor=south,inner sep=4pt] at (1.5,2) {$l$};
          \node[anchor=south,inner sep=4pt] at (1.5,3) {$x_{ul}$};
          \node[anchor=west,inner sep=4pt] at (2,0.5) {$y_d$};
          \node[anchor=west,inner sep=5pt] at (2,1.5) {$d$}; 
          \node[anchor=west,inner sep=5pt] at (2,2.5) {$u$}; 
          \node[anchor=west,inner sep=4pt] at (2,3.5) {$y_u$};
          \node[anchor=south,inner sep=4pt] at (2.5,1) {$x_{dr}$};
          \node[anchor=south,inner sep=4pt] at (2.5,2) {$r$};
          \node[anchor=south,inner sep=4pt] at (2.5,3) {$x_{ur}$};
          \node[anchor=west,inner sep=5pt] at (3,1.5) {$y_{dr}$}; 
          \node[anchor=west,inner sep=5pt] at (3,2.5) {$y_{ur}$}; 
          \node[anchor=south,inner sep=4pt] at (3.5,2) {$x_r$};
      \end{tikzpicture}
  \end{equation*}
  As the number of completely contracted horizontal and vertical terms are the same, we obtain the equation
  \begin{equation}\label{eq:PEPS_NN_single_tensor}
      \begin{tikzpicture}[z ={(4.85mm,4.85mm)}]
          \draw (-0.5, 0) -- (0.5,0);
          \draw (0,-0.5) -- (0,0.5);
          \draw (0,0,0) -- (0,0,0.5);
          \node[tensor,red] at (0,0) {};
          \node[anchor=north west, inner sep=4pt, align=left] at (-0,0) {$l+r+$\\$u+d$};
      \end{tikzpicture} + 
      \begin{tikzpicture}[z ={(4.85mm,4.85mm)}]
          \draw (-1.5, 0) -- (0.5,0);
          \draw (0,-0.5) -- (0,0.5);
          \draw (0,0,0) -- (0,0,0.5);
          \node[tensor] at (0,0) {};
          \filldraw[red] (-1,0) circle (2pt);
          \node[anchor=south, inner sep=4pt] at (-1,0) {$x_l + y_{ul} + y_{dl}$};
      \end{tikzpicture} + 
      \begin{tikzpicture}[z ={(4.85mm,4.85mm)}]
          \draw (-0.5, 0) -- (1.5,0);
          \draw (0,-0.5) -- (0,0.5);
          \draw (0,0,0) -- (0,0,0.5);
          \node[tensor] at (0,0) {};
          \filldraw[red] (1,0) circle (2pt);
          \node[anchor=north, inner sep=4pt] at (1,0) {$x_r + y_{ur} + y_{dr}$};
      \end{tikzpicture} + 
      \begin{tikzpicture}[z ={(4.85mm,4.85mm)}]
          \draw (-0.5, 0) -- (0.5,0);
          \draw (0,-0.5) -- (0,1);
          \draw (0,0,0) -- (0,0,0.5);
          \node[tensor] at (0,0) {};
          \filldraw[red] (0,0.5) circle (2pt);
          \node[anchor=east, inner sep=4pt,align=center] at (0,0.5) {$y_u + $\\$x_{ul} + x_{ur}$};
      \end{tikzpicture} +
      \begin{tikzpicture}[z ={(4.85mm,4.85mm)}]
          \draw (-0.5, 0) -- (0.5,0);
          \draw (0,0.5) -- (0,-1);
          \draw (0,0,0) -- (0,0,0.5);
          \node[tensor] at (0,0) {};
          \filldraw[red] (0,-0.5) circle (2pt);
          \node[anchor=west, inner sep=4pt,align=center] at (0,-0.5) {$y_d +$\\$ x_{dl} + x_{dr}$};
      \end{tikzpicture} = 0.
  \end{equation}
  Let us now consider again the equation $O\ket{\psi} = 0$ and apply the inverse of $A$ on every lattice site except for two neighboring ones, and then contract all neighboring indices. The result is a contraction of two PEPS tensor of the same shape as $A$. There are 23 terms of the local operator $O$ that don't get completely contracted. The location of them w.r.t.\ the special vertex is depicted below:
  \begin{equation*}
      \begin{tikzpicture}
          \draw (-0.5,-0.5) grid (5.5,4.5);
          \filldraw[red] (2,2) circle (3pt);
          \filldraw[red] (3,2) circle (3pt);
          \foreach \x in {(0.5,2),
                          (1,1.5), (1,2.5),
                          (1.5,1), (1.5,2), (1.5,3),
                          (2,0.5), (2,1.5), (2,2.5), (2,3.5), 
                          (2.5,1),  (2.5,2), (2.5,3),
                          (3,0.5), (3,1.5), (3,2.5), (3,3.5), 
                          (3.5,1),  (3.5,2), (3.5,3),
                          (4,1.5), (4,2.5),
                          (4.5,2)} {
            \filldraw \x circle (2pt);
          }
          \node[anchor=south,inner sep=4pt] at (0.5,2) {$x_l$};
          \node[anchor=west,inner sep=5pt] at (1,1.5) {$y_{dl}$}; 
          \node[anchor=west,inner sep=5pt] at (1,2.5) {$y_{ul}$}; 
          \node[anchor=south,inner sep=4pt] at (1.5,1) {$x_{dl}$};
          \node[anchor=south,inner sep=4pt] at (1.5,2) {$l$};
          \node[anchor=south,inner sep=4pt] at (1.5,3) {$x_{ul}$};
          \node[anchor=west,inner sep=4pt] at (2,0.5) {$y_d$};
          \node[anchor=west,inner sep=5pt] at (2,1.5) {$d$}; 
          \node[anchor=west,inner sep=5pt] at (2,2.5) {$u$}; 
          \node[anchor=west,inner sep=4pt] at (2,3.5) {$y_u$};
          \node[anchor=south,inner sep=4pt] at (2.5,1) {$x_d$};
          \node[anchor=south,inner sep=4pt] at (2.5,2) {$h$};
          \node[anchor=south,inner sep=4pt] at (2.5,3) {$x_u$};
          \node[anchor=west,inner sep=4pt] at (3,0.5) {$y_d$};
          \node[anchor=west,inner sep=5pt] at (3,1.5) {$d$}; 
          \node[anchor=west,inner sep=5pt] at (3,2.5) {$u$}; 
          \node[anchor=west,inner sep=4pt] at (3,3.5) {$y_u$};
          \node[anchor=south,inner sep=4pt] at (3.5,1) {$x_{dr}$};
          \node[anchor=south,inner sep=4pt] at (3.5,2) {$r$};
          \node[anchor=south,inner sep=4pt] at (3.5,3) {$x_{ur}$};
          \node[anchor=west,inner sep=5pt] at (4,1.5) {$y_{dr}$}; 
          \node[anchor=west,inner sep=5pt] at (4,2.5) {$y_{ur}$}; 
          \node[anchor=south,inner sep=4pt] at (4.5,2) {$x_r$};
      \end{tikzpicture}
  \end{equation*}
  As there is one more completely contracted horizontal term than vertical term, we obtain
  \begin{equation*}
        \begin{tikzpicture}[z={(-4.85mm,-4.85mm)}]
            \draw (-0.5,0) -- (1.5,0);
            \foreach \x in {0,1}{
              \draw (\x,0,-0.5) -- (\x,0,0.5);
              \draw (\x,0) --++ (0,1);
              \node[tensor] at (\x,0) {};
            }
            \draw[fill=gray, rounded corners=1mm] (-0.2,0.4) rectangle (1.2,0.6);  
        \end{tikzpicture} +
        \begin{tikzpicture}[z={(-4.85mm,-4.85mm)}]
            \draw (-0.5,0) -- (1.5,0);
            \foreach \x in {0,1}{
              \draw (\x,0,-0.5) -- (\x,0,0.5);
              \draw (\x,0) --++ (0,0.6);
            }
          \node[tensor,red, label = below right:$\hat L$] at (0,0) {};
          \node[tensor] at (1,0) {};
        \end{tikzpicture} +
        \begin{tikzpicture}[z={(-4.85mm,-4.85mm)}]
            \draw (-0.5,0) -- (1.5,0);
            \foreach \x in {0,1}{
              \draw (\x,0,-0.5) -- (\x,0,0.5);
              \draw (\x,0) --++ (0,0.6);
            }
          \node[tensor] at (0,0) {};
          \node[tensor,red,label = below right:$\hat R$] at (1,0) {};
        \end{tikzpicture} +
        \begin{tikzpicture}[z={(-4.85mm,-4.85mm)}]
            \draw (-0.5,0) -- (1.5,0);
            \foreach \x in {0,1}{
              \draw (\x,0,-1) -- (\x,0,1);
              \draw (\x,0) --++ (0,0.6);
              \node[tensor] at (\x,0) {};
            }
            \filldraw[red,canvas is xz plane at y=0, rounded corners=1mm]  (-0.2,0.4) rectangle (1.2,0.6);
            \node at (0.1,-0.5) {$u_x$};
        \end{tikzpicture}  +
        \begin{tikzpicture}[z={(-4.85mm,-4.85mm)}]
            \draw (-0.5,0) -- (1.5,0);
            \foreach \x in {0,1}{
              \draw (\x,0,-1) -- (\x,0,1);
              \draw (\x,0) --++ (0,0.6);
              \node[tensor] at (\x,0) {};
            }
            \filldraw[red,canvas is xz plane at y=0, rounded corners=1mm]  (-0.2,-0.4) rectangle (1.2,-0.6);
            \node at (0.8,0.5) {$d_x$};
        \end{tikzpicture} = 0, 
  \end{equation*}
  where 
  \begin{align*}
      \begin{tikzpicture}[z ={(4.85mm,4.85mm)}]
          \draw (-0.5, 0) -- (0.5,0);
          \draw (0,-0.5) -- (0,0.5);
          \draw (0,0,0) -- (0,0,0.5);
          \node[tensor,red] at (0,0) {};
          \node[anchor=north west, inner sep=4pt, align=left] at (-0,0) {$\hat L$};
      \end{tikzpicture} &= 
      \begin{tikzpicture}[z ={(4.85mm,4.85mm)}]
          \draw (-0.5, 0) -- (0.5,0);
          \draw (0,-0.5) -- (0,0.5);
          \draw (0,0,0) -- (0,0,0.5);
          \node[tensor,red] at (0,0) {};
          \node[anchor=north west, inner sep=4pt, align=center] at (-0,0) {$l+$\\$u+d$};
      \end{tikzpicture} + 
      \begin{tikzpicture}[z ={(4.85mm,4.85mm)}]
          \draw (-1.5, 0) -- (0.5,0);
          \draw (0,-0.5) -- (0,0.5);
          \draw (0,0,0) -- (0,0,0.5);
          \node[tensor] at (0,0) {};
          \filldraw[red] (-1,0) circle (2pt);
          \node[anchor=south, inner sep=4pt] at (-1,0) {$x_l + y_{ul} + y_{dl}$};
      \end{tikzpicture} + \lambda \  
      \begin{tikzpicture}[z ={(4.85mm,4.85mm)}]
          \draw (-0.5, 0) -- (0.5,0);
          \draw (0,-0.5) -- (0,0.5);
          \draw (0,0,0) -- (0,0,0.5);
          \node[tensor] at (0,0) {};
      \end{tikzpicture} + 
      \begin{tikzpicture}[z ={(4.85mm,4.85mm)}]
          \draw (-0.5, 0) -- (0.5,0);
          \draw (0,-0.5) -- (0,1);
          \draw (0,0,0) -- (0,0,0.5);
          \node[tensor] at (0,0) {};
          \filldraw[red] (0,0.5) circle (2pt);
          \node[anchor=east, inner sep=4pt,align=center] at (0,0.5) {$y_u + x_{ul}$};
      \end{tikzpicture} +
      \begin{tikzpicture}[z ={(4.85mm,4.85mm)}]
          \draw (-0.5, 0) -- (0.5,0);
          \draw (0,0.5) -- (0,-1);
          \draw (0,0,0) -- (0,0,0.5);
          \node[tensor] at (0,0) {};
          \filldraw[red] (0,-0.5) circle (2pt);
          \node[anchor=west, inner sep=4pt,align=center] at (0,-0.5) {$y_d + x_{dl}$};
      \end{tikzpicture} ,\\
      \begin{tikzpicture}[z ={(4.85mm,4.85mm)}]
          \draw (-0.5, 0) -- (0.5,0);
          \draw (0,-0.5) -- (0,0.5);
          \draw (0,0,0) -- (0,0,0.5);
          \node[tensor,red] at (0,0) {};
          \node[anchor=north west, inner sep=4pt, align=left] at (-0,0) {$\hat R$};
      \end{tikzpicture} &= 
      \begin{tikzpicture}[z ={(4.85mm,4.85mm)}]
          \draw (-0.5, 0) -- (0.5,0);
          \draw (0,-0.5) -- (0,0.5);
          \draw (0,0,0) -- (0,0,0.5);
          \node[tensor,red] at (0,0) {};
          \node[anchor=north west, inner sep=4pt, align=center] at (-0,0) {$r+$\\$u+d$};
      \end{tikzpicture} + 
      \begin{tikzpicture}[z ={(4.85mm,4.85mm)}]
          \draw (-1.5, 0) -- (0.5,0);
          \draw (0,-0.5) -- (0,0.5);
          \draw (0,0,0) -- (0,0,0.5);
          \node[tensor] at (0,0) {};
          \filldraw[red] (-1,0) circle (2pt);
          \node[anchor=south, inner sep=4pt] at (-1,0) {$x_r + y_{ur} + y_{dr}$};
      \end{tikzpicture} + 
      \begin{tikzpicture}[z ={(4.85mm,4.85mm)}]
          \draw (-0.5, 0) -- (0.5,0);
          \draw (0,-0.5) -- (0,1);
          \draw (0,0,0) -- (0,0,0.5);
          \node[tensor] at (0,0) {};
          \filldraw[red] (0,0.5) circle (2pt);
          \node[anchor=east, inner sep=4pt,align=center] at (0,0.5) {$y_u + x_{ur}$};
      \end{tikzpicture} +
      \begin{tikzpicture}[z ={(4.85mm,4.85mm)}]
          \draw (-0.5, 0) -- (0.5,0);
          \draw (0,0.5) -- (0,-1);
          \draw (0,0,0) -- (0,0,0.5);
          \node[tensor] at (0,0) {};
          \filldraw[red] (0,-0.5) circle (2pt);
          \node[anchor=west, inner sep=4pt,align=center] at (0,-0.5) {$y_d + x_{dr}$};
      \end{tikzpicture} .
  \end{align*}
  Using now \cref{eq:PEPS_NN_xu_xd} let us define 
  \begin{equation*}
      X = x_u  - x_{ur}  \otimes 1 - 1 \otimes x_{ul}  = - x_d  + x_{dr} \otimes 1 + 1 \otimes x_{dl}.
  \end{equation*}
  We can then write
  \begin{equation*}
        \begin{tikzpicture}[z={(-4.85mm,-4.85mm)}]
            \draw (-0.5,0) -- (1.5,0);
            \foreach \x in {0,1}{
              \draw (\x,0,-0.5) -- (\x,0,0.5);
              \draw (\x,0) --++ (0,1);
              \node[tensor] at (\x,0) {};
            }
            \draw[fill=gray, rounded corners=1mm] (-0.2,0.4) rectangle (1.2,0.6);  
        \end{tikzpicture} =
        \begin{tikzpicture}[z={(-4.85mm,-4.85mm)}]
            \draw (-0.5,0) -- (1.5,0);
            \foreach \x in {0,1}{
              \draw (\x,0,-0.5) -- (\x,0,0.5);
              \draw (\x,0) --++ (0,0.6);
            }
          \node[tensor,red, label = below right:$L$] at (0,0) {};
          \node[tensor] at (1,0) {};
        \end{tikzpicture} +
        \begin{tikzpicture}[z={(-4.85mm,-4.85mm)}]
            \draw (-0.5,0) -- (1.5,0);
            \foreach \x in {0,1}{
              \draw (\x,0,-0.5) -- (\x,0,0.5);
              \draw (\x,0) --++ (0,0.6);
            }
          \node[tensor] at (0,0) {};
          \node[tensor,red,label = below right:$R$] at (1,0) {};
        \end{tikzpicture} -
        \begin{tikzpicture}[z={(-4.85mm,-4.85mm)}]
            \draw (-0.5,0) -- (1.5,0);
            \foreach \x in {0,1}{
              \draw (\x,0,-1) -- (\x,0,1);
              \draw (\x,0) --++ (0,0.6);
              \node[tensor] at (\x,0) {};
            }
            \filldraw[red,canvas is xz plane at y=0, rounded corners=1mm]  (-0.2,0.4) rectangle (1.2,0.6);
            \node at (0.1,-0.5) {$X$};
        \end{tikzpicture}  +
        \begin{tikzpicture}[z={(-4.85mm,-4.85mm)}]
            \draw (-0.5,0) -- (1.5,0);
            \foreach \x in {0,1}{
              \draw (\x,0,-1) -- (\x,0,1);
              \draw (\x,0) --++ (0,0.6);
              \node[tensor] at (\x,0) {};
            }
            \filldraw[red,canvas is xz plane at y=0, rounded corners=1mm]  (-0.2,-0.4) rectangle (1.2,-0.6);
            \node at (0.8,0.5) {$X$};
        \end{tikzpicture}, 
  \end{equation*}
  with
  \begin{align*}
      \begin{tikzpicture}[z ={(4.85mm,4.85mm)}]
          \draw (-0.5, 0) -- (0.5,0);
          \draw (0,-0.5) -- (0,0.5);
          \draw (0,0,0) -- (0,0,0.5);
          \node[tensor,red] at (0,0) {};
          \node[anchor=north west, inner sep=4pt, align=left] at (-0,0) {$L$};
      \end{tikzpicture}& = -  
      \begin{tikzpicture}[z ={(4.85mm,4.85mm)}]
          \draw (-0.5, 0) -- (0.5,0);
          \draw (0,-0.5) -- (0,0.5);
          \draw (0,0,0) -- (0,0,0.5);
          \node[tensor,red] at (0,0) {};
          \node[anchor=north west, inner sep=4pt, align=center] at (-0,0) {$l+$\\$u+d$};
      \end{tikzpicture} -
      \begin{tikzpicture}[z ={(4.85mm,4.85mm)}]
          \draw (-1.5, 0) -- (0.5,0);
          \draw (0,-0.5) -- (0,0.5);
          \draw (0,0,0) -- (0,0,0.5);
          \node[tensor] at (0,0) {};
          \filldraw[red] (-1,0) circle (2pt);
          \node[anchor=south, inner sep=4pt] at (-1,0) {$x_l + y_{ul} + y_{dl}$};
      \end{tikzpicture} - \lambda \  
      \begin{tikzpicture}[z ={(4.85mm,4.85mm)}]
          \draw (-0.5, 0) -- (0.5,0);
          \draw (0,-0.5) -- (0,0.5);
          \draw (0,0,0) -- (0,0,0.5);
          \node[tensor] at (0,0) {};
      \end{tikzpicture} - 
      \begin{tikzpicture}[z ={(4.85mm,4.85mm)}]
          \draw (-0.5, 0) -- (0.5,0);
          \draw (0,-0.5) -- (0,1);
          \draw (0,0,0) -- (0,0,0.5);
          \node[tensor] at (0,0) {};
          \filldraw[red] (0,0.5) circle (2pt);
          \node[anchor=east, inner sep=4pt,align=center] at (0,0.5) {$y_u + $\\$x_{ul} + x_{ur}$};
      \end{tikzpicture} -
      \begin{tikzpicture}[z ={(4.85mm,4.85mm)}]
          \draw (-0.5, 0) -- (0.5,0);
          \draw (0,0.5) -- (0,-1);
          \draw (0,0,0) -- (0,0,0.5);
          \node[tensor] at (0,0) {};
          \filldraw[red] (0,-0.5) circle (2pt);
          \node[anchor=west, inner sep=4pt,align=center] at (0,-0.5) {$y_d +$\\$ x_{dl} + x_{dr}$};
      \end{tikzpicture}, \\
      \begin{tikzpicture}[z ={(4.85mm,4.85mm)}]
          \draw (-0.5, 0) -- (0.5,0);
          \draw (0,-0.5) -- (0,0.5);
          \draw (0,0,0) -- (0,0,0.5);
          \node[tensor,red] at (0,0) {};
          \node[anchor=north west, inner sep=4pt, align=left] at (-0,0) {$R$};
      \end{tikzpicture} &= -
      \begin{tikzpicture}[z ={(4.85mm,4.85mm)}]
          \draw (-0.5, 0) -- (0.5,0);
          \draw (0,-0.5) -- (0,0.5);
          \draw (0,0,0) -- (0,0,0.5);
          \node[tensor,red] at (0,0) {};
          \node[anchor=north west, inner sep=4pt, align=center] at (-0,0) {$r+$\\$u+d$};
      \end{tikzpicture} - 
      \begin{tikzpicture}[z ={(4.85mm,4.85mm)}]
          \draw (1.5, 0) -- (-0.5,0);
          \draw (0,-0.5) -- (0,0.5);
          \draw (0,0,0) -- (0,0,0.5);
          \node[tensor] at (0,0) {};
          \filldraw[red] (1,0) circle (2pt);
          \node[anchor=north, inner sep=4pt] at (1,0) {$x_r + y_{ur} + y_{dr}$};
      \end{tikzpicture} - 
      \begin{tikzpicture}[z ={(4.85mm,4.85mm)}]
          \draw (-0.5, 0) -- (0.5,0);
          \draw (0,-0.5) -- (0,1);
          \draw (0,0,0) -- (0,0,0.5);
          \node[tensor] at (0,0) {};
          \filldraw[red] (0,0.5) circle (2pt);
          \node[anchor=east, inner sep=4pt,align=center] at (0,0.5) {$y_u + $\\$x_{ul} + x_{ur}$};
      \end{tikzpicture} -
      \begin{tikzpicture}[z ={(4.85mm,4.85mm)}]
          \draw (-0.5, 0) -- (0.5,0);
          \draw (0,0.5) -- (0,-1);
          \draw (0,0,0) -- (0,0,0.5);
          \node[tensor] at (0,0) {};
          \filldraw[red] (0,-0.5) circle (2pt);
          \node[anchor=west, inner sep=4pt,align=center] at (0,-0.5) {$y_d +$\\$ x_{dl} + x_{dr}$};
      \end{tikzpicture}.
  \end{align*}
  Using \cref{eq:PEPS_NN_single_tensor}, these expressions simplify:
  \begin{align*}
      \begin{tikzpicture}[z ={(4.85mm,4.85mm)}]
          \draw (-0.5, 0) -- (0.5,0);
          \draw (0,-0.5) -- (0,0.5);
          \draw (0,0,0) -- (0,0,0.5);
          \node[tensor,red] at (0,0) {};
          \node[anchor=north west, inner sep=4pt, align=left] at (-0,0) {$L$};
      \end{tikzpicture}& = 
      \begin{tikzpicture}[z ={(4.85mm,4.85mm)}]
          \draw (-0.5, 0) -- (0.5,0);
          \draw (0,-0.5) -- (0,0.5);
          \draw (0,0,0) -- (0,0,0.5);
          \node[tensor,red] at (0,0) {};
          \node[anchor=north west, inner sep=4pt, align=center] at (-0,0) {$r$};
      \end{tikzpicture} +
      \begin{tikzpicture}[z ={(4.85mm,4.85mm)}]
          \draw (1.5, 0) -- (-0.5,0);
          \draw (0,-0.5) -- (0,0.5);
          \draw (0,0,0) -- (0,0,0.5);
          \node[tensor] at (0,0) {};
          \filldraw[red] (1,0) circle (2pt);
          \node[anchor=north, inner sep=4pt] at (1,0) {$x_r + y_{ur} + y_{dr}$};
      \end{tikzpicture} - \lambda \  
      \begin{tikzpicture}[z ={(4.85mm,4.85mm)}]
          \draw (-0.5, 0) -- (0.5,0);
          \draw (0,-0.5) -- (0,0.5);
          \draw (0,0,0) -- (0,0,0.5);
          \node[tensor] at (0,0) {};
      \end{tikzpicture} , \\
      \begin{tikzpicture}[z ={(4.85mm,4.85mm)}]
          \draw (-0.5, 0) -- (0.5,0);
          \draw (0,-0.5) -- (0,0.5);
          \draw (0,0,0) -- (0,0,0.5);
          \node[tensor,red] at (0,0) {};
          \node[anchor=north west, inner sep=4pt, align=left] at (-0,0) {$R$};
      \end{tikzpicture} &= 
      \begin{tikzpicture}[z ={(4.85mm,4.85mm)}]
          \draw (-0.5, 0) -- (0.5,0);
          \draw (0,-0.5) -- (0,0.5);
          \draw (0,0,0) -- (0,0,0.5);
          \node[tensor,red] at (0,0) {};
          \node[anchor=north west, inner sep=4pt, align=center] at (-0,0) {$l$};
      \end{tikzpicture} +
      \begin{tikzpicture}[z ={(4.85mm,4.85mm)}]
          \draw (-1.5, 0) -- (0.5,0);
          \draw (0,-0.5) -- (0,0.5);
          \draw (0,0,0) -- (0,0,0.5);
          \node[tensor] at (0,0) {};
          \filldraw[red] (-1,0) circle (2pt);
          \node[anchor=south, inner sep=4pt] at (-1,0) {$x_l + y_{ul} + y_{dl}$};
      \end{tikzpicture} .
  \end{align*}

  Let us now consider again the equation $O\ket{\psi} = 0$ and apply the inverse of $A$ on every lattice site except for two vertical neighbors, and then contract all neighboring indices. The result is a contraction of two PEPS tensor of the same shape as $A$. There are 23 terms of the local operator $O$ that don't get completely contracted. The location of them w.r.t.\ the special vertex is depicted below:
  \begin{equation*}
      \begin{tikzpicture}
          \draw (-0.5,-0.5) grid (4.5,5.5);
          \filldraw[red] (2,2) circle (3pt);
          \filldraw[red] (2,3) circle (3pt);
          \foreach \x in {(0.5,2), (0.5,3),
                          (1,1.5), (1,2.5), (1,3.5),
                          (1.5,1), (1.5,2), (1.5,3), (1.5,4),
                          (2,0.5), (2,1.5), (2,2.5), (2,3.5), (2,4.5), 
                          (2.5,1),  (2.5,2), (2.5,3), (2.5,4),
                          (3,1.5), (3,2.5), (3,3.5), 
                          (3.5,2), (3.5,3)} {
            \filldraw \x circle (2pt);
          }
          \node[anchor=south,inner sep=4pt] at (0.5,2) {$x_l$};
          \node[anchor=south,inner sep=4pt] at (0.5,3) {$x_l$};
          \node[anchor=west,inner sep=5pt] at (1,1.5) {$y_{dl}$}; 
          \node[anchor=west,inner sep=5pt] at (1,2.5) {$y_{l}$}; 
          \node[anchor=west,inner sep=5pt] at (1,3.5) {$y_{ul}$}; 
          \node[anchor=south,inner sep=4pt] at (1.5,1) {$x_{dl}$};
          \node[anchor=south,inner sep=4pt] at (1.5,2) {$l$};
          \node[anchor=south,inner sep=4pt] at (1.5,3) {$l$};
          \node[anchor=south,inner sep=4pt] at (1.5,4) {$x_{ul}$};
          \node[anchor=west,inner sep=4pt] at (2,0.5) {$y_d$};
          \node[anchor=west,inner sep=5pt] at (2,1.5) {$d$}; 
          \node[anchor=west,inner sep=5pt] at (2,2.5) {$v$}; 
          \node[anchor=west,inner sep=5pt] at (2,3.5) {$u$}; 
          \node[anchor=west,inner sep=4pt] at (2,4.5) {$y_u$};
          \node[anchor=south,inner sep=4pt] at (2.5,1) {$x_{dr}$};
          \node[anchor=south,inner sep=4pt] at (2.5,2) {$r$};
          \node[anchor=south,inner sep=4pt] at (2.5,3) {$r$};
          \node[anchor=south,inner sep=4pt] at (2.5,4) {$x_{ur}$};
          \node[anchor=west,inner sep=5pt] at (3,1.5) {$y_{dr}$}; 
          \node[anchor=west,inner sep=5pt] at (3,2.5) {$y_r$}; 
          \node[anchor=west,inner sep=4pt] at (3,3.5) {$y_{ur}$};
          \node[anchor=south,inner sep=4pt] at (3.5,2) {$y_{r}$};
          \node[anchor=south,inner sep=4pt] at (3.5,3) {$y_r$};
      \end{tikzpicture}
  \end{equation*}
  As there is one more completely contracted horizontal term than vertical term, we obtain
  \begin{equation*}
              \begin{tikzpicture}[z={(-4.85mm,-4.85mm)}, baseline = -3.5mm]
            \draw (0,0,-0.5) -- (0,0,1.5);
            \foreach \z in {0,1}{
              \draw (-0.5,0,\z) -- (0.5,0,\z);
              \draw (0,0,\z) --++ (0,1);
              \node[tensor] at (0,0,\z) {};
            }
            \draw[fill=gray, rounded corners=1mm,canvas is yz plane at x=0] (0.4,-0.2) rectangle (0.6,1.2);
        \end{tikzpicture} +
        \begin{tikzpicture}[z={(-4.85mm,-4.85mm)}, baseline = -3.5mm]
            \draw (0,0,-0.5) -- (0,0,1.5);
            \foreach \z in {0,1}{
              \draw (-0.5,0,\z) -- (0.5,0,\z);
              \draw (0,0,\z) --++ (0,0.6);
            }
          \node[tensor,red,label=below right:$\hat{U}$] at (0,0,0) {};
          \node[tensor] at (0,0,1) {};
        \end{tikzpicture} +
        \begin{tikzpicture}[z={(-4.85mm,-4.85mm)}, baseline = -3.5mm]
            \draw (0,0,-0.5) -- (0,0,1.5);
            \foreach \z in {0,1}{
              \draw (-0.5,0,\z) -- (0.5,0,\z);
              \draw (0,0,\z) --++ (0,0.6);
            }
          \node[tensor] at (0,0,0) {};
          \node[tensor,red,label=below right:$\hat{D}$] at (0,0,1) {};
        \end{tikzpicture} +
        \begin{tikzpicture}[z={(-4.85mm,-4.85mm)}, baseline = -3.5mm]
            \draw (0,0,-0.5) -- (0,0,1.5);
            \foreach \z in {0,1}{
              \draw (-1,0,\z) -- (1,0,\z);
              \draw (0,0,\z) --++ (0,0.6);
              \node[tensor] at (0,0,\z) {};
            }
            \filldraw[red,canvas is xz plane at y=0, rounded corners=1mm]
              (0.4,-0.2) rectangle (0.6,1.2);
            \node at (0.7,-0.25) {$y_r$};
        \end{tikzpicture}
        +
        \begin{tikzpicture}[z={(-4.85mm,-4.85mm)}, baseline = -3.5mm]
            \draw (0,0,-0.5) -- (0,0,1.5);
            \foreach \z in {0,1}{
              \draw (-1,0,\z) -- (1,0,\z);
              \draw (0,0,\z) --++ (0,0.6);
              \node[tensor] at (0,0,\z) {};
            }
            \filldraw[red,canvas is xz plane at y=0, rounded corners=1mm]
              (-0.4,-0.2) rectangle (-0.6,1.2);
            \node at (-1.2,-0.25) {$y_l$};
        \end{tikzpicture} = 0,
  \end{equation*}
  with 
    \begin{align*}
      \begin{tikzpicture}[z ={(4.85mm,4.85mm)}]
          \draw (-0.5, 0) -- (0.5,0);
          \draw (0,-0.5) -- (0,0.5);
          \draw (0,0,0) -- (0,0,0.5);
          \node[tensor,red] at (0,0) {};
          \node[anchor=north west, inner sep=4pt, align=left] at (-0,0) {$\hat U$};
      \end{tikzpicture} &= 
      \begin{tikzpicture}[z ={(4.85mm,4.85mm)}]
          \draw (-0.5, 0) -- (0.5,0);
          \draw (0,-0.5) -- (0,0.5);
          \draw (0,0,0) -- (0,0,0.5);
          \node[tensor,red] at (0,0) {};
          \node[anchor=north west, inner sep=4pt, align=center] at (-0,0) {$l+$\\$u+r$};
      \end{tikzpicture} + 
      \begin{tikzpicture}[z ={(4.85mm,4.85mm)}]
          \draw (-1.5, 0) -- (0.5,0);
          \draw (0,-0.5) -- (0,0.5);
          \draw (0,0,0) -- (0,0,0.5);
          \node[tensor] at (0,0) {};
          \filldraw[red] (-1,0) circle (2pt);
          \node[anchor=south, inner sep=4pt] at (-1,0) {$x_l + y_{ul}$};
      \end{tikzpicture} - \lambda \  
      \begin{tikzpicture}[z ={(4.85mm,4.85mm)}]
          \draw (-0.5, 0) -- (0.5,0);
          \draw (0,-0.5) -- (0,0.5);
          \draw (0,0,0) -- (0,0,0.5);
          \node[tensor] at (0,0) {};
      \end{tikzpicture} + 
      \begin{tikzpicture}[z ={(4.85mm,4.85mm)}]
          \draw (-0.5, 0) -- (0.5,0);
          \draw (0,-0.5) -- (0,1);
          \draw (0,0,0) -- (0,0,0.5);
          \node[tensor] at (0,0) {};
          \filldraw[red] (0,0.5) circle (2pt);
          \node[anchor=east, inner sep=4pt,align=center] at (0,0.5) {$y_u + x_{ul} + x_{ur}$};
      \end{tikzpicture} +
      \begin{tikzpicture}[z ={(4.85mm,4.85mm)}]
          \draw (-0.5, 0) -- (1.5,0);
          \draw (0,0.5) -- (0,-0.5);
          \draw (0,0,0) -- (0,0,0.5);
          \node[tensor] at (0,0) {};
          \filldraw[red] (1,0) circle (2pt);
          \node[anchor=north, inner sep=4pt,align=center] at (1,0) {$y_r + y_{ur}$};
      \end{tikzpicture}, \\
      \begin{tikzpicture}[z ={(4.85mm,4.85mm)}]
          \draw (-0.5, 0) -- (0.5,0);
          \draw (0,-0.5) -- (0,0.5);
          \draw (0,0,0) -- (0,0,0.5);
          \node[tensor,red] at (0,0) {};
          \node[anchor=north west, inner sep=4pt, align=left] at (-0,0) {$\hat D$};
      \end{tikzpicture} &= 
      \begin{tikzpicture}[z ={(4.85mm,4.85mm)}]
          \draw (-0.5, 0) -- (0.5,0);
          \draw (0,-0.5) -- (0,0.5);
          \draw (0,0,0) -- (0,0,0.5);
          \node[tensor,red] at (0,0) {};
          \node[anchor=north west, inner sep=4pt, align=center] at (-0,0) {$r+$\\$l+d$};
      \end{tikzpicture} + 
      \begin{tikzpicture}[z ={(4.85mm,4.85mm)}]
          \draw (-1.5, 0) -- (0.5,0);
          \draw (0,-0.5) -- (0,0.5);
          \draw (0,0,0) -- (0,0,0.5);
          \node[tensor] at (0,0) {};
          \filldraw[red] (-1,0) circle (2pt);
          \node[anchor=south, inner sep=4pt] at (-1,0) {$x_l + y_{dl}$};
      \end{tikzpicture} + 
      \begin{tikzpicture}[z ={(4.85mm,4.85mm)}]
          \draw (-0.5, 0) -- (1.5,0);
          \draw (0,-0.5) -- (0,0.5);
          \draw (0,0,0) -- (0,0,0.5);
          \node[tensor] at (0,0) {};
          \filldraw[red] (1,0) circle (2pt);
          \node[anchor=north, inner sep=4pt,align=center] at (1,0) {$y_u + y_{dr}$};
      \end{tikzpicture} +
      \begin{tikzpicture}[z ={(4.85mm,4.85mm)}]
          \draw (-0.5, 0) -- (0.5,0);
          \draw (0,0.5) -- (0,-1);
          \draw (0,0,0) -- (0,0,0.5);
          \node[tensor] at (0,0) {};
          \filldraw[red] (0,-0.5) circle (2pt);
          \node[anchor=west, inner sep=4pt,align=center] at (0,-0.5) {$y_d + x_{dl} + x_{dr}$};
      \end{tikzpicture} .
  \end{align*}
  Using now \cref{eq:PEPS_NN_V}, let us define
  \begin{equation*}
    Y =  y_l  - y_{dl} \otimes 1 - 1 \otimes y_{ul} = -y_r + y_{dr} \otimes 1 + 1 \otimes y_{ur}.
  \end{equation*}
  We can then write 
  \begin{equation*}
              \begin{tikzpicture}[z={(-4.85mm,-4.85mm)}, baseline = -3.5mm]
            \draw (0,0,-0.5) -- (0,0,1.5);
            \foreach \z in {0,1}{
              \draw (-0.5,0,\z) -- (0.5,0,\z);
              \draw (0,0,\z) --++ (0,1);
              \node[tensor] at (0,0,\z) {};
            }
            \draw[fill=gray, rounded corners=1mm,canvas is yz plane at x=0] (0.4,-0.2) rectangle (0.6,1.2);
        \end{tikzpicture} = 
        \begin{tikzpicture}[z={(-4.85mm,-4.85mm)}, baseline = -3.5mm]
            \draw (0,0,-0.5) -- (0,0,1.5);
            \foreach \z in {0,1}{
              \draw (-0.5,0,\z) -- (0.5,0,\z);
              \draw (0,0,\z) --++ (0,0.6);
            }
          \node[tensor,red,label=below right:$U$] at (0,0,0) {};
          \node[tensor] at (0,0,1) {};
        \end{tikzpicture} +
        \begin{tikzpicture}[z={(-4.85mm,-4.85mm)}, baseline = -3.5mm]
            \draw (0,0,-0.5) -- (0,0,1.5);
            \foreach \z in {0,1}{
              \draw (-0.5,0,\z) -- (0.5,0,\z);
              \draw (0,0,\z) --++ (0,0.6);
            }
          \node[tensor] at (0,0,0) {};
          \node[tensor,red,label=below right:$D$] at (0,0,1) {};
        \end{tikzpicture} +
        \begin{tikzpicture}[z={(-4.85mm,-4.85mm)}, baseline = -3.5mm]
            \draw (0,0,-0.5) -- (0,0,1.5);
            \foreach \z in {0,1}{
              \draw (-1,0,\z) -- (1,0,\z);
              \draw (0,0,\z) --++ (0,0.6);
              \node[tensor] at (0,0,\z) {};
            }
            \filldraw[red,canvas is xz plane at y=0, rounded corners=1mm]
              (0.4,-0.2) rectangle (0.6,1.2);
            \node at (0.7,-0.25) {$Y$};
        \end{tikzpicture}
        -
        \begin{tikzpicture}[z={(-4.85mm,-4.85mm)}, baseline = -3.5mm]
            \draw (0,0,-0.5) -- (0,0,1.5);
            \foreach \z in {0,1}{
              \draw (-1,0,\z) -- (1,0,\z);
              \draw (0,0,\z) --++ (0,0.6);
              \node[tensor] at (0,0,\z) {};
            }
            \filldraw[red,canvas is xz plane at y=0, rounded corners=1mm]
              (-0.4,-0.2) rectangle (-0.6,1.2);
            \node at (-1.2,-0.25) {$Y$};
        \end{tikzpicture},
  \end{equation*}
  with
  \begin{align*}
      \begin{tikzpicture}[z ={(4.85mm,4.85mm)}]
          \draw (-0.5, 0) -- (0.5,0);
          \draw (0,-0.5) -- (0,0.5);
          \draw (0,0,0) -- (0,0,0.5);
          \node[tensor,red] at (0,0) {};
          \node[anchor=north west, inner sep=4pt, align=left] at (-0,0) {$U$};
      \end{tikzpicture}& = -  
      \begin{tikzpicture}[z ={(4.85mm,4.85mm)}]
          \draw (-0.5, 0) -- (0.5,0);
          \draw (0,-0.5) -- (0,0.5);
          \draw (0,0,0) -- (0,0,0.5);
          \node[tensor,red] at (0,0) {};
          \node[anchor=north west, inner sep=4pt, align=center] at (-0,0) {$l+$\\$u+r$};
      \end{tikzpicture} -
      \begin{tikzpicture}[z ={(4.85mm,4.85mm)}]
          \draw (-1.5, 0) -- (0.5,0);
          \draw (0,-0.5) -- (0,0.5);
          \draw (0,0,0) -- (0,0,0.5);
          \node[tensor] at (0,0) {};
          \filldraw[red] (-1,0) circle (2pt);
          \node[anchor=south, inner sep=4pt] at (-1,0) {$x_l + y_{ul} + y_{dl}$};
      \end{tikzpicture} + \lambda \  
      \begin{tikzpicture}[z ={(4.85mm,4.85mm)}]
          \draw (-0.5, 0) -- (0.5,0);
          \draw (0,-0.5) -- (0,0.5);
          \draw (0,0,0) -- (0,0,0.5);
          \node[tensor] at (0,0) {};
      \end{tikzpicture} - 
      \begin{tikzpicture}[z ={(4.85mm,4.85mm)}]
          \draw (-0.5, 0) -- (0.5,0);
          \draw (0,-0.5) -- (0,1);
          \draw (0,0,0) -- (0,0,0.5);
          \node[tensor] at (0,0) {};
          \filldraw[red] (0,0.5) circle (2pt);
          \node[anchor=east, inner sep=4pt,align=center] at (0,0.5) {$y_u + x_{ul} + x_{ur}$};
      \end{tikzpicture} -
      \begin{tikzpicture}[z ={(4.85mm,4.85mm)}]
          \draw (-0.5, 0) -- (1.5,0);
          \draw (0,0.5) -- (0,-0.5);
          \draw (0,0,0) -- (0,0,0.5);
          \node[tensor] at (0,0) {};
          \filldraw[red] (1,0) circle (2pt);
          \node[anchor=north, inner sep=4pt,align=center] at (1,0) {$y_r + y_{ur} + y_{dr}$};
      \end{tikzpicture}, \\
      \begin{tikzpicture}[z ={(4.85mm,4.85mm)}]
          \draw (-0.5, 0) -- (0.5,0);
          \draw (0,-0.5) -- (0,0.5);
          \draw (0,0,0) -- (0,0,0.5);
          \node[tensor,red] at (0,0) {};
          \node[anchor=north west, inner sep=4pt, align=left] at (-0,0) {$D$};
      \end{tikzpicture} &= -
      \begin{tikzpicture}[z ={(4.85mm,4.85mm)}]
          \draw (-0.5, 0) -- (0.5,0);
          \draw (0,-0.5) -- (0,0.5);
          \draw (0,0,0) -- (0,0,0.5);
          \node[tensor,red] at (0,0) {};
          \node[anchor=north west, inner sep=4pt, align=center] at (-0,0) {$r+$\\$l+d$};
      \end{tikzpicture} -
      \begin{tikzpicture}[z ={(4.85mm,4.85mm)}]
          \draw (-1.5, 0) -- (0.5,0);
          \draw (0,-0.5) -- (0,0.5);
          \draw (0,0,0) -- (0,0,0.5);
          \node[tensor] at (0,0) {};
          \filldraw[red] (-1,0) circle (2pt);
          \node[anchor=south, inner sep=4pt] at (-1,0) {$x_l + y_{dl}+y_{ul}$};
      \end{tikzpicture} - 
      \begin{tikzpicture}[z ={(4.85mm,4.85mm)}]
          \draw (-0.5, 0) -- (1.5,0);
          \draw (0,-0.5) -- (0,0.5);
          \draw (0,0,0) -- (0,0,0.5);
          \node[tensor] at (0,0) {};
          \filldraw[red] (1,0) circle (2pt);
          \node[anchor=north, inner sep=4pt,align=center] at (1,0) {$y_u + y_{dr} + y_{ur}$};
      \end{tikzpicture} -
      \begin{tikzpicture}[z ={(4.85mm,4.85mm)}]
          \draw (-0.5, 0) -- (0.5,0);
          \draw (0,0.5) -- (0,-1);
          \draw (0,0,0) -- (0,0,0.5);
          \node[tensor] at (0,0) {};
          \filldraw[red] (0,-0.5) circle (2pt);
          \node[anchor=west, inner sep=4pt,align=center] at (0,-0.5) {$y_d + x_{dl} + x_{dr}$};
      \end{tikzpicture} .
  \end{align*}
  Using \cref{eq:PEPS_NN_single_tensor}, these expressions simplify:
  \begin{align*}
      \begin{tikzpicture}[z ={(4.85mm,4.85mm)}]
          \draw (-0.5, 0) -- (0.5,0);
          \draw (0,-0.5) -- (0,0.5);
          \draw (0,0,0) -- (0,0,0.5);
          \node[tensor,red] at (0,0) {};
          \node[anchor=north west, inner sep=4pt, align=left] at (-0,0) {$L$};
      \end{tikzpicture}& = 
      \begin{tikzpicture}[z ={(4.85mm,4.85mm)}]
          \draw (-0.5, 0) -- (0.5,0);
          \draw (0,-0.5) -- (0,0.5);
          \draw (0,0,0) -- (0,0,0.5);
          \node[tensor,red] at (0,0) {};
          \node[anchor=north west, inner sep=4pt, align=center] at (-0,0) {$d$};
      \end{tikzpicture} +
      \begin{tikzpicture}[z ={(4.85mm,4.85mm)}]
          \draw (0.5, 0) -- (-0.5,0);
          \draw (0,-1.0) -- (0,0.5);
          \draw (0,0,0) -- (0,0,0.5);
          \node[tensor] at (0,0) {};
          \filldraw[red] (0,-0.5) circle (2pt);
          \node[anchor=west, inner sep=4pt] at (0,-0.5) {$y_d + x_{dr} + x_{dl}$};
      \end{tikzpicture} + \lambda \  
      \begin{tikzpicture}[z ={(4.85mm,4.85mm)}]
          \draw (-0.5, 0) -- (0.5,0);
          \draw (0,-0.5) -- (0,0.5);
          \draw (0,0,0) -- (0,0,0.5);
          \node[tensor] at (0,0) {};
      \end{tikzpicture} , \\
      \begin{tikzpicture}[z ={(4.85mm,4.85mm)}]
          \draw (-0.5, 0) -- (0.5,0);
          \draw (0,-0.5) -- (0,0.5);
          \draw (0,0,0) -- (0,0,0.5);
          \node[tensor,red] at (0,0) {};
          \node[anchor=north west, inner sep=4pt, align=left] at (-0,0) {$R$};
      \end{tikzpicture} &= 
      \begin{tikzpicture}[z ={(4.85mm,4.85mm)}]
          \draw (-0.5, 0) -- (0.5,0);
          \draw (0,-0.5) -- (0,0.5);
          \draw (0,0,0) -- (0,0,0.5);
          \node[tensor,red] at (0,0) {};
          \node[anchor=north west, inner sep=4pt, align=center] at (-0,0) {$u$};
      \end{tikzpicture} +
      \begin{tikzpicture}[z ={(4.85mm,4.85mm)}]
          \draw (-0.5, 0) -- (0.5,0);
          \draw (0,-0.5) -- (0,1);
          \draw (0,0,0) -- (0,0,0.5);
          \node[tensor] at (0,0) {};
          \filldraw[red] (0,0.5) circle (2pt);
          \node[anchor=west, inner sep=4pt] at (0,0.5) {$y_u + x_{ul} + x_{ur}$};
      \end{tikzpicture} .
  \end{align*}
  Finally notice that by \cref{eq:PEPS_NN_single_tensor}, $L+R+U+D=0$.
\end{proof}

\subsection{Two-body interaction in the hexagonal lattice.}

\begin{theorem}
    Let $O$ be a nearest neighbor operator on an hexagonal lattice: $O = \sum_i h_i + \sum_j l_j + \sum_m k_m$, where each term correspond to a different edge. Let $\ket{\psi}$ be a translationally invariant PEPS generated by two injective trivalent tensors $A_1$ and $A_2$ that we both depict in black such that $O \ket{\psi} = 0$. Then, there are four tensors $B,Y,R,W$ satisfying: 
\begin{align*}
        \begin{tikzpicture}[scale=0.7]
            \node[tensor,label=below :$A_2$] (l) at (1,0) {};
            \draw (l)--++(0,1);
            \draw (l)--++(60:0.7);
            \draw (l)--++(-60:0.7);
            \node[tensor,label=below:$A_1$] (r) at (0,0) {};
            \draw (r)--++(0,1);
            \draw (r)--++(120:0.7);
            \draw (r)--++(-120:0.7);
            \draw (l)--(r);
            \draw[line width=2.1mm,line cap=round] (0,0.5) -- (1,0.5);
            \draw[line width=1.9mm,line cap=round,gray] (0,0.5) -- (1,0.5);
            \node[anchor=south,inner sep=5pt] at (0.5,0.5) {$h$};
        \end{tikzpicture}  & = \ 
         \begin{tikzpicture}[scale=0.7]
            \node[tensor,fill=red,label=below :$R$] (l) at (1,0) {};
            \draw (l)--++(0,0.7);
            \draw (l)--++(60:0.7);
            \draw (l)--++(-60:0.7);
            \node[tensor,label=below :$A_1$] (r) at (0,0) {};
            \draw (r)--++(0,0.7);
            \draw (r)--++(120:0.7);
            \draw (r)--++(-120:0.7);
            \draw (l)--(r);
        \end{tikzpicture} +\
        \begin{tikzpicture}[scale=0.7]
            \node[tensor,label=below :$A_2$] (l) at (1,0) {};
            \draw (l)--++(0,0.7);
            \draw (l)--++(60:0.7);
            \draw (l)--++(-60:0.7);
            \node[tensor,fill=yellow,label=below :$Y$] (r) at (0,0) {};
            \draw (r)--++(0,0.7);
            \draw (r)--++(120:0.7);
            \draw (r)--++(-120:0.7);
            \draw (l)--(r);
        \end{tikzpicture} \ ,
\\
        \begin{tikzpicture}[scale=0.7]
            \node[tensor,label=below :$A_2$] (l) at (0,0) {};
            \draw (l)--++(0,1);
            \draw (l)--++(60:0.7);
            \draw (l)--++(-180:0.7);
            \draw (l)--++(-60:1);
            \node[tensor,label=below :$A_1$] (r) at (-60:1) {};
            \draw (r)--++(0,1);
            \draw (r)--++(0:0.7);
            \draw (r)--++(-120:0.7);
            \draw (l)--(r);
            \draw[line width=2.1mm,line cap=round] (0,0.65) --++ (-60:1);
            \draw[line width=1.9mm,line cap=round,gray] (0,0.65) --++ (-60:1);
            \node[anchor=south,inner sep=5pt] at (0.5,0.5) {$l$};
        \end{tikzpicture} &= \ 
        \begin{tikzpicture}[scale=0.7]
            \node[tensor,fill=white, label=below :$W$] (l) at (0,0) {};
            \draw (l)--++(0,0.7);
            \draw (l)--++(60:0.7);
            \draw (l)--++(-180:0.7);
            \draw (l)--++(-60:1);
            \node[tensor,label=below :$A_1$] (r) at (-60:1) {};
            \draw (r)--++(0,0.7);
            \draw (r)--++(0:0.7);
            \draw (r)--++(-120:0.7);
            \draw (l)--(r);
            \end{tikzpicture}\ - \
            \begin{tikzpicture}[scale=0.7]
            \node[tensor,label=below :$A_2$] (l) at (0,0) {};
            \draw (l)--++(0,0.7);
            \draw (l)--++(60:0.7);
            \draw (l)--++(-180:0.7);
            \draw (l)--++(-60:1);
            \node[tensor,fill=green,label=below right :$B+Y$] (r) at (-60:1) {};
            \draw (r)--++(0,0.7);
            \draw (r)--++(0:0.7);
            \draw (r)--++(-120:0.7);
            \draw (l)--(r);
            \end{tikzpicture} \ ,
\\
        \begin{tikzpicture}[scale=0.7]
            \node[tensor,label=below :$A_2$] (l) at (0,0) {};
            \draw (l)--++(0,1);
            \draw (l)--++(180:0.7);
            \draw (l)--++(-60:0.7);
            \node[tensor,label=below :$A_1$] (r) at (60:1) {};
            \draw (r)--++(0,1);
            \draw (r)--++(0:0.7);
            \draw (r)--++(120:0.7);
            \draw (l)--(r);
            \draw[line width=2.1mm,line cap=round] (0,0.65) --++ (60:1);
            \draw[line width=1.9mm,line cap=round,gray] (0,0.65) --++ (60:1);
            \node[anchor=south,inner sep=5pt] at (-0.3,0.7) {$k$};
        \end{tikzpicture} 
        & = \ - \
        \begin{tikzpicture}[scale=0.7]
    \node[tensor,fill=pink,label=below left :$R+W$] (l) at (0,0) {};
            \draw (l)--++(0,0.7);
            \draw (l)--++(180:0.7);
            \draw (l)--++(-60:0.7);
            \node[tensor,label=below :$A_1$] (r) at (60:1) {};
            \draw (r)--++(0,0.7);
            \draw (r)--++(0:0.7);
            \draw (r)--++(120:0.7);
            \draw (l)--(r);
        \end{tikzpicture} + \ 
        \begin{tikzpicture}[scale=0.7]
            \node[tensor,label=below :$A_2$] (l) at (0,0) {};
            \draw (l)--++(0,0.7);
            \draw (l)--++(180:0.7);
            \draw (l)--++(-60:0.7);
            \node[tensor,fill=blue,label=below :$B$] (r) at (60:1) {};
            \draw (r)--++(0,0.7);
            \draw (r)--++(0:0.7);
            \draw (r)--++(120:0.7);
            \draw (l)--(r);
        \end{tikzpicture}\ .
\end{align*} 
\end{theorem}

\begin{proof}
We first apply the inverse of $A_1$ and $A_2$ on every lattice site in the equation $O\ket{\psi} =0$, and then fully contract the resulting tensor network. On a lattice with $N$ hexagons we obtain the equation $N(\eta + \lambda+ \kappa ) = 0$, or equivalently,
\begin{equation}\label{eq:hexagon_hkl}
    \eta + \lambda+ \kappa  = 0,
\end{equation}
where $\eta$ is defined as 
\begin{align*} 
    \eta = 
    \begin{tikzpicture}[scale=0.7]
        \node[tensor,label=below :$A_2$] (l) at (1,0) {};
        \node[tensor,label=above :$A_2^{-1}$] (lu) at (1,1) {};
        \draw (l)--++(0,1);
        \draw (l)--++(60:0.7) --++ (0,1);
        \draw (l)--++(-60:0.7) --++ (0,1);
        \draw (lu)--++(60:0.7);
        \draw (lu)--++(-60:0.7);
        \node[tensor,label=below:$A_1$] (r) at (0,0) {};
        \node[tensor,label=above:$A_1^{-1}$] (ru) at (0,1) {};
        \draw (r)--++(0,1);
        \draw (r)--++(120:0.7) --++ (0,1);
        \draw (r)--++(-120:0.7) --++ (0,1);
        \draw (ru)--++(120:0.7);
        \draw (ru)--++(-120:0.7);
        \draw (l)--(r);
        \draw (lu)--(ru);
        \draw[line width=2mm,line cap=round] (0,0.5) -- (1,0.5);
        \node[anchor=south,inner sep=5pt] at (0.5,0.5) {$h$};
    \end{tikzpicture},
\end{align*} 
and $\lambda$ and $\kappa$ are defined analogously for the Hamiltonian terms $l$ and $k$, repsectively. 

Next we apply the inverse of $A_1$ and $A_2$ everywhere except on a single site site in the equation $O\ket{\psi} =0$ and contract the resulting tensor network. When the site left out is at the position of the tensor $A_1$, the result is a single tensor of the same shape as $A_1$. Using \cref{eq:hexagon_hkl}, we obtain an equation containing $9$ terms. The position of the Hamiltonians from which these 9 terms originate as well as the resulting eqution are depicted below: 
\begin{equation*}
        \begin{tikzpicture}[scale=0.7]
  \node[tensor,fill=red] (i) at (0,0) {};
            \draw (i)--++(0:1)--++(60:1)--++(120:1)--++(180:1)--++(240:1);
            \draw (i)--++(120:1)--++(180:1)--++(240:1)--++(300:1)--++(0:1);
            \draw (i)--++(240:1)--++(300:1)--++(0:1)--++(60:1)--++(120:1);
            \filldraw (0:0.5) circle (2pt);
          \node[anchor=north,inner sep=5pt] at (0:0.5) {$V$};
          \filldraw (120:0.5) circle (2pt);
          \node[anchor=west,inner sep=5pt] at (120:0.5) {$C$};
          \filldraw (240:0.5) circle (2pt);
          \node[anchor=east,inner sep=5pt] at (240:0.5) {$D$};
        \filldraw ($(0:1)+(60:0.5)$) circle (2pt);
          \node[anchor=north,inner sep=5pt] at ($(0:1)+(60:0.5)$) {$v_1$};
          \filldraw ($(0:1)+(-60:0.5)$) circle (2pt);
          \node[anchor=north,inner sep=5pt] at ($(0:1)+(-60:0.5)$) {$v_2$};
          \filldraw ($(120:1)+(180:0.5)$) circle (2pt);
          \node[anchor=north,inner sep=5pt] at ($(120:1)+(180:0.5)$) {$c_1$};
          \filldraw ($(120:1)+(60:0.5)$) circle (2pt);
          \node[anchor=north,inner sep=5pt] at ($(120:1)+(60:0.5)$) {$c_2$};
          \filldraw ($(240:1)+(180:0.5)$) circle (2pt);
          \node[anchor=north,inner sep=5pt] at ($(240:1)+(180:0.5)$) {$d_1$};
          \filldraw ($(240:1)+(300:0.5)$) circle (2pt);
          \node[anchor=north,inner sep=5pt] at ($(240:1)+(300:0.5)$) {$d_2$};
        \end{tikzpicture} \quad , \quad
        \begin{tikzpicture}[scale=0.7]
            \node[tensor,fill=red,label=below:$V+C+D$] (r) at (0,0) {};
            \draw (0,0.5)--(r)--++(120:0.7);
            \draw (0.7,0)--(r)--++(-120:0.7);
        \end{tikzpicture}
        +
             \begin{tikzpicture}[scale=0.7]
            \node[tensor,label=left:$A_1$] (r) at (0,0) {};
            \draw (0,0.5)--(r)--++(120:0.7);
            \draw (0.7,0)--(r)--++(-120:0.7);
            \node[tensor,scale=0.7,fill=red,label= below:$v_1+v_2$] at (0.5,0) {};
        \end{tikzpicture}
                +
             \begin{tikzpicture}[scale=0.7]
            \node[tensor,label=left:$A_1$] (r) at (0,0) {};
            \draw (0,0.5)--(r)--++(120:0.7);
            \draw (0.7,0)--(r)--++(-120:0.7);
            \node[tensor,scale=0.7,fill=red,label= right:$d_1+d_2$] at (240:0.5) {};
        \end{tikzpicture}
        +
             \begin{tikzpicture}[scale=0.7]
            \node[tensor,label=below:$A_1$] (r) at (0,0) {};
            \draw (0,0.5)--(r)--++(120:0.7);
            \draw (0.7,0)--(r)--++(-120:0.7);
            \node[tensor,scale=0.7,fill=red,label= right:$c_1+c_2$] at (120:0.5) {};
        \end{tikzpicture}
        = \ 0 \ .
        \end{equation*}
Similarly, if we apply the inverse of $A_1$ and $A_2$ on every lattice site except for one that is at a position of a tensor $A_2$, we obtain:
    \begin{equation*}
        \begin{tikzpicture}[scale=0.7]
  \node[tensor,fill=red] (i) at (0,0) {};
            \draw (i)--++(180:1)--++(240:1)--++(300:1)--++(0:1)--++(60:1);
            \draw (i)--++(60:1)--++(120:1)--++(180:1)--++(240:1)--++(300:1);
            \draw (i)--++(300:1)--++(0:1)--++(60:1)--++(120:1)--++(180:1);
            \filldraw (180:0.5) circle (2pt);
          \node[anchor=north,inner sep=5pt] at (180:0.5) {$E$};
        \filldraw ($(180:1)+(120:0.5)$) circle (2pt);
          \node[anchor=north,inner sep=5pt] at ($(180:1)+(120:0.5)$) {$e_1$};
          \filldraw ($(180:1)+(-120:0.5)$) circle (2pt);
          \node[anchor=north,inner sep=5pt] at ($(180:1)+(-120:0.5)$) {$e_2$};
        \filldraw (60:0.5) circle (2pt);
          \node[anchor=north west,inner sep=2pt] at (60:0.5) {$F$};
            \filldraw ($(60:1)+(120:0.5)$) circle (2pt);
        \node[anchor=south,inner sep=4pt] at ($(60:1)+(120:0.5)$) {$f_1$};
          \filldraw ($(60:1)+(0:0.5)$) circle (2pt);
          \node[anchor=south,inner sep=4pt] at ($(60:1)+(0:0.5)$) {$f_2$};
          \filldraw (300:0.5) circle (2pt);
          \node[anchor=north,inner sep=5pt] at (300:0.5) {$G$};
          \filldraw ($(300:1)+(0:0.5)$) circle (2pt);
          \node[anchor=north,inner sep=5pt] at ($(300:1)+(0:0.5)$) {$g_1$};
          \filldraw ($(300:1)+(240:0.5)$) circle (2pt);
          \node[anchor=north,inner sep=5pt] at ($(300:1)+(240:0.5)$) {$g_2$};
        \end{tikzpicture}\quad , \quad
        \begin{tikzpicture}[scale=0.7]
            \node[tensor,fill=red,label=below:$E+F+G$] (r) at (0,0) {};
            \draw (0,0.5)--(r)--++(60:0.7);
            \draw (-0.7,0)--(r)--++(-60:0.7);
        \end{tikzpicture}
        +
        \begin{tikzpicture}[scale=0.7]
            \node[tensor,label=right:$A_2$] (r) at (0,0) {};
            \draw (0,0.5)--(r)--++(60:0.7);
            \draw (-0.7,0)--(r)--++(-60:0.7);
            \node[tensor,scale=0.7,fill=red,label= below:$e_1+e_2$] at (-0.5,0) {};
        \end{tikzpicture}
         +
        \begin{tikzpicture}[scale=0.7]
            \node[tensor,label=below left:$A_2$] (r) at (0,0) {};
            \draw (0,0.5)--(r)--++(60:0.7);
            \draw (-0.7,0)--(r)--++(-60:0.7);
            \node[tensor,scale=0.7,fill=red,label= right:$f_1+f_2$] at (60:0.5) {};
        \end{tikzpicture}
        +
        \begin{tikzpicture}[scale=0.7]
            \node[tensor,label=below left:$A_2$] (r) at (0,0) {};
            \draw (0,0.5)--(r)--++(60:0.7);
            \draw (-0.7,0)--(r)--++(-60:0.7);
            \node[tensor,scale=0.7,fill=red,label= right:$g_1+g_2$] at (-60:0.5) {};
        \end{tikzpicture}
        = \ 0 \ .
        \end{equation*}
We define $X_1= X+x_1+x_2$ with $X\in\{V,C,D,E,F,G \}$. Then:
$$ E_1+F_1+G_1=0 \quad , \quad  V_1+C_1+D_1=0$$
Now we apply the inverse everywhere except on two consecutive sites, in all three directions. Using \cref{eq:hexagon_hkl} in all three cases, there are 13 terms of $O$ that contribute non-trivially and two terms contributing as the identity. Below we depict the terms that contribute non-trivially and the resulting equations:
    \begin{equation*}
        \begin{tikzpicture}[scale=0.7]
  \node[tensor,fill=red] (i) at (0,0) {};
            \draw (i)--++(180:1)--++(240:1)--++(300:1)--++(0:1)--++(60:1);
            \draw (i)--++(60:1)--++(120:1)--++(180:1)--++(240:1)--++(300:1);
            \draw (i)--++(300:1)--++(0:1)--++(60:1)--++(120:1)--++(180:1);
        \node[tensor,fill=red] (f) at (60:1) {};
        \draw (f)--++(0:1)--++(60:1)--++(120:1)--++(180:1)--++(240:1);
        \filldraw (180:0.5) circle (2pt);
          \node[anchor=north,inner sep=5pt] at (180:0.5) {$E$};
        \filldraw ($(180:1)+(120:0.5)$) circle (2pt);
          \node[anchor=north,inner sep=5pt] at ($(180:1)+(120:0.5)$) {$e_1$};
          \filldraw ($(180:1)+(-120:0.5)$) circle (2pt);
          \node[anchor=north,inner sep=5pt] at ($(180:1)+(-120:0.5)$) {$e_2$};
        \filldraw (60:0.5) circle (2pt);
            \filldraw ($(60:1)+(120:0.5)$) circle (2pt);
        \node[anchor=west,inner sep=4pt] at ($(60:1)+(120:0.5)$) {$C$};
        \filldraw ($(60:1)+(120:1)+(180:0.5)$) circle (2pt);
        \node[anchor=north,inner sep=5pt] at ($(60:1)+(120:1)+(180:0.5)$) {$c_1$};
        \filldraw ($(60:1.5)+(120:1)$) circle (2pt);
        \node[anchor=north,inner sep=5pt] at ($(60:1.5)+(120:1)$) {$c_2$};

          \filldraw (300:0.5) circle (2pt);
          \node[anchor=north,inner sep=5pt] at (300:0.5) {$G$};
          \filldraw ($(300:1)+(0:0.5)$) circle (2pt);
          \node[anchor=north,inner sep=5pt] at ($(300:1)+(0:0.5)$) {$g_1$};
          \filldraw ($(300:1)+(240:0.5)$) circle (2pt);
          \node[anchor=north,inner sep=5pt] at ($(300:1)+(240:0.5)$) {$g_2$};
        \filldraw ($(60:1)+(0:0.5)$) circle (2pt);
          \node[anchor=north,inner sep=5pt] at ($(60:1)+(0:0.5)$) {$V$};
        \filldraw ($(60:1.5)+(0:1)$) circle (2pt);
        \node[anchor=north,inner sep=5pt] at ($(60:1.5)+(0:1)$) {$v_1$};
          \filldraw ($(60:1)+(0:1)+(300:0.5)$) circle (2pt);
          \node[anchor=north,inner sep=5pt] at ($(60:1)+(0:1)+(300:0.5)$) {$v_2$};
        \end{tikzpicture}
\quad , \quad 
 \begin{tikzpicture}[scale=0.7]
            \node[tensor,label=below :$A_2$] (l) at (0,0) {};
            \draw (l)--++(0,1);
            \draw (l)--++(180:0.7);
            \draw (l)--++(-60:0.7);
            \node[tensor,label=below :$A_1$] (r) at (60:1) {};
            \draw (r)--++(0,1);
            \draw (r)--++(0:0.7);
            \draw (r)--++(120:0.7);
            \draw (l)--(r);
            \draw[line width=2mm,line cap=round] (0,0.65) --++ (60:1);
            \node[anchor=south,inner sep=5pt] at (-0.3,0.7) {$k$};
        \end{tikzpicture} 
      = \ - \
        \begin{tikzpicture}[scale=0.7]
            \node[tensor,fill=red,label=below left :$E_1+G_1$] (l) at (0,0) {};
            \draw (l)--++(0,0.7);
            \draw (l)--++(180:0.7);
            \draw (l)--++(-60:0.7);
            \node[tensor,label=below :$A_1$] (r) at (60:1) {};
            \draw (r)--++(0,0.7);
            \draw (r)--++(0:0.7);
            \draw (r)--++(120:0.7);
            \draw (l)--(r);
        \end{tikzpicture} - \ 
        \begin{tikzpicture}[scale=0.7]
            \node[tensor,label=below :$A_2$] (l) at (0,0) {};
            \draw (l)--++(0,0.7);
            \draw (l)--++(180:0.7);
            \draw (l)--++(-60:0.7);
            \node[tensor,fill=red,label=below right :$C_1+V_1$] (r) at (60:1) {};
            \draw (r)--++(0,0.7);
            \draw (r)--++(0:0.7);
            \draw (r)--++(120:0.7);
            \draw (l)--(r);
        \end{tikzpicture}\ - (\eta + \lambda)
        \begin{tikzpicture}[scale=0.7]
            \node[tensor,label=below :$A_2$] (l) at (0,0) {};
            \draw (l)--++(0,0.7);
            \draw (l)--++(180:0.7);
            \draw (l)--++(-60:0.7);
            \node[tensor,label=below right :$A_1$] (r) at (60:1) {};
            \draw (r)--++(0,0.7);
            \draw (r)--++(0:0.7);
            \draw (r)--++(120:0.7);
            \draw (l)--(r);
        \end{tikzpicture}\ .
        \end{equation*}

    \begin{equation*}
        \begin{tikzpicture}[scale=0.7]
  \node[tensor,fill=red] (i) at (0,0) {};
            \draw (i)--++(180:1)--++(240:1)--++(300:1)--++(0:1)--++(60:1);
            \draw (i)--++(60:1)--++(120:1)--++(180:1)--++(240:1)--++(300:1);
            \draw (i)--++(300:1)--++(0:1)--++(60:1)--++(120:1)--++(180:1);
        \node[tensor,fill=red] (f) at (300:1) {};
        \draw (f)--++(240:1)--++(300:1)--++(0:1)--++(60:1)--++(120:1);
        \filldraw (180:0.5) circle (2pt);
          \node[anchor=north,inner sep=5pt] at (180:0.5) {$E$};
        \filldraw ($(180:1)+(120:0.5)$) circle (2pt);
          \node[anchor=north,inner sep=5pt] at ($(180:1)+(120:0.5)$) {$e_1$};
          \filldraw ($(180:1)+(-120:0.5)$) circle (2pt);
          \node[anchor=north,inner sep=5pt] at ($(180:1)+(-120:0.5)$) {$e_2$};
        \filldraw (60:0.5) circle (2pt);
          \node[anchor=north west,inner sep=2pt] at (60:0.5) {$F$};
            \filldraw ($(60:1)+(120:0.5)$) circle (2pt);
        \node[anchor=south,inner sep=5pt] at ($(60:1)+(120:0.5)$) {$f_1$};
          \filldraw ($(60:1)+(0:0.5)$) circle (2pt);
          \node[anchor=south,inner sep=5pt] at ($(60:1)+(0:0.5)$) {$f_2$};
          \filldraw (300:0.5) circle (2pt);
          \filldraw ($(300:1)+(0:0.5)$) circle (2pt);
          \node[anchor=north,inner sep=5pt] at ($(300:1)+(0:0.5)$) {$V$};
          \filldraw ($(300:1)+(0:1)+(60:0.5)$) circle (2pt);
          \node[anchor=north,inner sep=5pt] at ($(300:1)+(0:1)+(60:0.5)$) {$v_1$};
          \filldraw ($(300:1.5)+(0:1)$) circle (2pt);
          \node[anchor=north,inner sep=5pt] at ($(300:1.5)+(0:1)$) {$v_2$};
          \filldraw ($(300:1)+(240:0.5)$) circle (2pt);
          \node[anchor=north,inner sep=5pt] at ($(300:1)+(240:0.5)$) {$D$};
          \filldraw ($(300:1.5)+(240:1)$) circle (2pt);
          \node[anchor=north,inner sep=5pt] at ($(300:1.5)+(240:1)$) {$d_2$};
          \filldraw ($(300:1)+(240:1)+(180:0.5)$) circle (2pt);
          \node[anchor=north,inner sep=5pt] at ($(300:1)+(240:1)+(180:0.5)$) {$d_1$};
        \end{tikzpicture} \quad , \quad 
\begin{tikzpicture}[scale=0.7]
            \node[tensor,label=below :$A_2$] (l) at (0,0) {};
            \draw (l)--++(0,1);
            \draw (l)--++(60:0.7);
            \draw (l)--++(-180:0.7);
            \draw (l)--++(-60:1);
            \node[tensor,label=below :$A_1$] (r) at (-60:1) {};
            \draw (r)--++(0,1);
            \draw (r)--++(0:0.7);
            \draw (r)--++(-120:0.7);
            \draw (l)--(r);
            \draw[line width=2mm,line cap=round] (0,0.65) --++ (-60:1);
            \node[anchor=south,inner sep=5pt] at (0.5,0.5) {$l$};
        \end{tikzpicture} = \ - \ 
        \begin{tikzpicture}[scale=0.7]
            \node[tensor,fill=red, label=below left :$E_1+F_1$] (l) at (0,0) {};
            \draw (l)--++(0,0.7);
            \draw (l)--++(60:0.7);
            \draw (l)--++(-180:0.7);
            \draw (l)--++(-60:1);
            \node[tensor,label=below :$A_1$] (r) at (-60:1) {};
            \draw (r)--++(0,0.7);
            \draw (r)--++(0:0.7);
            \draw (r)--++(-120:0.7);
            \draw (l)--(r);
            \end{tikzpicture}\ - \
            \begin{tikzpicture}[scale=0.7]
            \node[tensor,label=below :$A_2$] (l) at (0,0) {};
            \draw (l)--++(0,0.7);
            \draw (l)--++(60:0.7);
            \draw (l)--++(-180:0.7);
            \draw (l)--++(-60:1);
            \node[tensor,fill=red,label=below right :$D_1+V_1$] (r) at (-60:1) {};
            \draw (r)--++(0,0.7);
            \draw (r)--++(0:0.7);
            \draw (r)--++(-120:0.7);
            \draw (l)--(r);
            \end{tikzpicture} \ - (\eta + \kappa) 
            \begin{tikzpicture}[scale=0.7]
            \node[tensor,label=below :$A_2$] (l) at (0,0) {};
            \draw (l)--++(0,0.7);
            \draw (l)--++(60:0.7);
            \draw (l)--++(-180:0.7);
            \draw (l)--++(-60:1);
            \node[tensor,fill=red,label=below right :$A_1$] (r) at (-60:1) {};
            \draw (r)--++(0,0.7);
            \draw (r)--++(0:0.7);
            \draw (r)--++(-120:0.7);
            \draw (l)--(r);
            \end{tikzpicture} \ .
        \end{equation*}

\begin{equation*}
        \begin{tikzpicture}[scale=0.7]
  \node[tensor,fill=red] (i) at (0,0) {};
            \draw (i)--++(0:1)--++(60:1)--++(120:1)--++(180:1)--++(240:1);
            \draw (i)--++(120:1)--++(180:1)--++(240:1)--++(300:1)--++(0:1);
            \draw (i)--++(240:1)--++(300:1)--++(0:1)--++(60:1)--++(120:1);
        \node[tensor,fill=red] (f) at (1,0) {};
        \draw (f)--++(300:1)--++(0:1)--++(60:1)--++(120:1)--++(180:1);
        \filldraw (0.5,0) circle (2pt);
          \filldraw (120:0.5) circle (2pt);
          \node[anchor=west,inner sep=5pt] at (120:0.5) {$C$};
          \filldraw ($(120:1)+(180:0.5)$) circle (2pt);
          \node[anchor=north,inner sep=5pt] at ($(120:1)+(180:0.5)$) {$c_1$};
          \filldraw ($(120:1)+(60:0.5)$) circle (2pt);
          \node[anchor=north,inner sep=5pt] at ($(120:1)+(60:0.5)$) {$c_2$};
          \filldraw (240:0.5) circle (2pt);
          \node[anchor=north,inner sep=5pt] at (240:0.5) {$D$};
          \filldraw ($(240:1)+(180:0.5)$) circle (2pt);
          \node[anchor=north,inner sep=5pt] at ($(240:1)+(180:0.5)$) {$d_1$};
          \filldraw ($(240:1)+(300:0.5)$) circle (2pt);
          \node[anchor=north,inner sep=5pt] at ($(240:1)+(300:0.5)$) {$d_2$};
           \filldraw ($(0:1)+(60:0.5)$) circle (2pt);
          \node[anchor=north west,inner sep=2pt] at ($(0:1)+(60:0.5)$) {$F$};
          \filldraw ($(0:1.5)+(60:1)$) circle (2pt);
          \node[anchor=south,inner sep=5pt] at ($(0:1.5)+(60:1)$) {$f_2$};
          \filldraw ($(0:1)+(60:1)+(120:0.5)$) circle (2pt);
          \node[anchor=south,inner sep=5pt] at ($(0:1)+(60:1)+(120:0.5)$) {$f_1$};
           \filldraw ($(0:1)+(300:0.5)$) circle (2pt);
          \node[anchor=north,inner sep=5pt] at ($(0:1)+(300:0.5)$) {$G$};
          \filldraw ($(0:1.5)+(300:1)$) circle (2pt);
          \node[anchor=north,inner sep=5pt] at ($(0:1.5)+(300:1)$) {$g_1$};
          \filldraw ($(0:1)+(300:1)+(240:0.5)$) circle (2pt);
          \node[anchor=north,inner sep=5pt] at ($(0:1)+(300:1)+(240:0.5)$) {$g_2$};
        \end{tikzpicture} 
        \quad , \quad 
         \begin{tikzpicture}[scale=0.7]
            \node[tensor,label=below :$A_2$] (l) at (1,0) {};
            \draw (l)--++(0,1);
            \draw (l)--++(60:0.7);
            \draw (l)--++(-60:0.7);
            \node[tensor,label=below:$A_1$] (r) at (0,0) {};
            \draw (r)--++(0,1);
            \draw (r)--++(120:0.7);
            \draw (r)--++(-120:0.7);
            \draw (l)--(r);
            \draw[line width=2mm,line cap=round] (0,0.5) -- (1,0.5);
            \node[anchor=south,inner sep=5pt] at (0.5,0.5) {$h$};
        \end{tikzpicture}   = \ - \
         \begin{tikzpicture}[scale=0.7]
            \node[tensor,fill=red,label=right :$F_1+G_1$] (l) at (1,0) {};
            \draw (l)--++(0,0.7);
            \draw (l)--++(60:0.7);
            \draw (l)--++(-60:0.7);
            \node[tensor,label=below :$A_1$] (r) at (0,0) {};
            \draw (r)--++(0,0.7);
            \draw (r)--++(120:0.7);
            \draw (r)--++(-120:0.7);
            \draw (l)--(r);
        \end{tikzpicture} - \
        \begin{tikzpicture}[scale=0.7]
            \node[tensor,label=below :$A_2$] (l) at (1,0) {};
            \draw (l)--++(0,0.7);
            \draw (l)--++(60:0.7);
            \draw (l)--++(-60:0.7);
            \node[tensor,fill=red,label=left :$C_1+D_1$] (r) at (0,0) {};
            \draw (r)--++(0,0.7);
            \draw (r)--++(120:0.7);
            \draw (r)--++(-120:0.7);
            \draw (l)--(r);
        \end{tikzpicture} \ - (\lambda + \kappa) 
        \begin{tikzpicture}[scale=0.7]
            \node[tensor,label=below :$A_2$] (l) at (1,0) {};
            \draw (l)--++(0,0.7);
            \draw (l)--++(60:0.7);
            \draw (l)--++(-60:0.7);
            \node[tensor,label=left :$A_1$] (r) at (0,0) {};
            \draw (r)--++(0,0.7);
            \draw (r)--++(120:0.7);
            \draw (r)--++(-120:0.7);
            \draw (l)--(r);
        \end{tikzpicture} \ .
        \end{equation*}

By using $\eta + \lambda + \kappa = 0$ and defining $Y=-(C_1+D_1+\kappa A_1)$, $R=-(F_1+G_1+ \lambda A_2 )$, $W=-(E_1+F_1+ \eta A_2)$ and $B=-( C_1+V_1)$ we obtain the desired equations.
\end{proof}

\bibliography{bibliography}

\end{document}